\documentclass{lmcs}
\pdfoutput=1

\usepackage{lastpage}
\lmcsdoi{17}{3}{9}
\lmcsheading{}{\pageref{LastPage}}{}{}%
{Jun.~25,~2020}{Jul.~21,~2021}{}

\usepackage[utf8]{inputenc}

\renewcommand{\sfdefault}{cmss}

\makeatletter%
\@ifundefined{basedir}{%
\newcommand\basedir{}%
}{}%
\makeatother%

\usepackage{amsfonts}
\usepackage{amsthm}
\usepackage{amssymb}
\usepackage{stmaryrd}

\usepackage{mathpartir}

\usepackage{\basedir pftools}
\usepackage{\basedir iris}
\usepackage{bussproofs}

\usepackage{xcolor}
\usepackage{xspace}
\usepackage{xpunctuate}

\usepackage{graphicx}
\usepackage[inline]{enumitem}
\usepackage{semantic}
\usepackage{csquotes}
\usepackage{mdframed}
\usepackage{hyperref}
\usepackage[nameinlink]{cleveref}
\usepackage{listings}
\usepackage{lstcoq}
\usepackage{xparse}

\usepackage{subcaption}
\usepackage[nomargin,inline]{fixme}

\SetSymbolFont{stmry}{bold}{U}{stmry}{m}{n} % this fixes warnings when \boldsymbol is used with stmaryrd included

\extrarowheight=\jot	% else, arrays are scrunched compared to, say, aligned
\newcolumntype{.}{@{}}
\newcolumntype{M}{@{\mskip\thickmuskip}}

\definecolor{StringRed}{rgb}{.637,0.082,0.082}
\definecolor{CommentGreen}{rgb}{0.0,0.55,0.3}
\definecolor{KeywordBlue}{rgb}{0.0,0.3,0.55}
\definecolor{LinkColor}{rgb}{0.55,0.0,0.3}
\definecolor{CiteColor}{rgb}{0.55,0.0,0.3}
\definecolor{HighlightColor}{rgb}{0.0,0.0,0.0}

\definecolor{grey}{rgb}{0.5,0.5,0.5}

\hypersetup{%
  linktocpage=true, pdfstartview=FitV,
  breaklinks=true, pageanchor=true, pdfpagemode=UseOutlines,
  plainpages=false, bookmarksnumbered, bookmarksopen=true, bookmarksopenlevel=3,
  hypertexnames=true, pdfhighlight=/O,
  colorlinks=true,linkcolor=LinkColor,citecolor=CiteColor,
  urlcolor=LinkColor
}

\fxsetup{theme=color,mode=multiuser,draft}
\FXRegisterAuthor{lars}{alars}{\colorbox{blue!70}{\color{white}Lars}}
\FXRegisterAuthor{dan}{adan}{\colorbox{green!70}{\color{black}Dan}}
\FXRegisterAuthor{robbert}{arobbert}{\colorbox{orange!70}{\color{black}Robbert}}

\newcommand{\ie}{\emph{i.e.,} }
\newcommand{\cf}{\emph{c.f.} }
\newcommand{\eg}{\emph{e.g.,} }
\newcommand{\etc}{\emph{etc}\xperiod}
\newcommand{\etal}{\emph{et~al}\xperiod}
\newcommand{\wrt}{w.r.t.~}

\newcommand{\reloc}{ReLoC\xspace}
\newcommand{\relocTwo}{ReLoC Reloaded\xspace}

\newcommand{\PersistentProp}{\Prop_{\always}}

\newcommand{\ctxrel}{\precsim_{\mathit{ctx}}}
\newcommand{\logrelsymb}{\precsim_{\mathit{log}}}
\newcommand{\ctxref}[4]{\ifthenelse{\equal{#1}{}}{}{#1 \proves} #2 \ctxrel #3 : #4}
\newcommand{\ctxequivsymb}{\mathrel{\simeq}_{\mathit{ctx}}}
\newcommand{\ctxequiv}[4]{\ifthenelse{\equal{#1}{}}{}{#1 \proves} #2 \ctxequivsymb #3 : #4}
\newcommand{\lrel}{\precsim}

\DeclareDocumentCommand{\REL}{O{} m m m}
  {\mathsf{REL}_{#1}\spac#2 \lrel #3 \ifthenelse{\equal{#4}{}}{}{ : #4}}

\DeclareDocumentCommand{\logrel}{O{} m O{} m m m}
  {
  \ifthenelse{\equal{#2}{}}{%
    \ifthenelse{\equal{#3}{}}{%
      {}% No gamma
    }{#3}
  }{#2 \ifthenelse{\equal{#3}{}}{%
    }{\mid #3}}
  \ifthenelse{\equal{#1}{}}{\ifthenelse{\equal{#2}{}\and\equal{#3}{}}{}{\mathrel{\models}}}{\mathrel{\models_{#1}}}
    #4 \lrel #5 \ifthenelse{\equal{#6}{}}{}{ : #6}}

\newcommand{\logrelV}[4]{\Sem{#1}_{#2}(#3, #4)}

\newcommand{\mapstoI}{\mathrel{\mapsto_{\mathsf{i}}}}
\newcommand{\mapstoS}{\mathrel{\mapsto_{\mathsf{s}}}}

\newcommand{\fmapstoS}[1][]{\fmapsto[#1]_\mathsf{s}}
\newcommand{\loc}{\ell}

\newcommand{\validSubst}[3]{\Sem{#2}^{\ast}_{#1}(#3)}

\newcommand{\HeapLang}{\textlog{HeapLang}\xspace}

\newcommand{\Type}{\textdom{Type}}
\newcommand{\TVar}{\textdom{TVar}}
\newcommand{\Pctx}{\textdom{Ctx}}
\newcommand{\Ectx}{\textdom{ECtx}}

\newcommand{\efs}{\vec{\expr_f}}
\newcommand{\cfgg}{\rho}

\newcommand{\pureexec}[2]{#1 \ra_{\mathsf{pure}} #2}
\newcommand{\conS}[2]{{#1}\mathrel{\textnormal{::}}{#2}}

\newcommand{\lvar}[1]{\textit{#1}}

\newcommand{\langkwStyle}[1]{\textsf{\color{blue} #1}}
\newcommand{\langconstStyle}[1]{\textsf{\textbf{#1}}}
\reservestyle{\langkw}{\langkwStyle}
\reservestyle{\langconst}{\langconstStyle}
\reservestyle{\langother}{\mathit}
\langkw{let[let ],in[ in ],
  match,with,
  if[if ],then[ then ],else[ else ],
  inl,inr,fold,unfold,pack,unpack,and,
  ref,rec[rec ],fork,CAS,FAA,val,type}

\langconst{assert,false,true,None,Some,sig,end,cons,nil}

\def\Match#1with#2{\<match> \spac #1 \spac \<with> \spac #2}
\def\Let#1=#2in{\<let> #1 \mathrel{=} #2 \<in> }
\def\If#1then{\<if> #1 \<then>}
\def\Else{\<else>}
\def\Ref(#1){\<ref>(#1)}
\def\Rec#1 #2={\<rec> {#1} \  {#2} \mathrel{=} }
\newcommand*\Fork[1]{\<fork>\spac\set{#1}}
\newcommand\deref{\mathop{!}}
\let\gets\leftarrow

\newcommand{\diverge}{\textlog{diverge}}

\newcommand{\unpack}[2]{\<unpack>\;{#1} \<in>#2}
\newcommand{\CAS}[3]{\<CAS>(#1, #2, #3)}
\newcommand{\tapp}[1]{#1 \langle\rangle}
\newcommand{\tlam}[1]{\Lambda. #1}
\newcommand{\proj}{\pi}
\newcommand{\NONE}{\<None>}
\newcommand{\SOME}{\<Some>}

\newcommand{\textType}[1]{\mathsf{#1}}
\newcommand{\tvar}{\textType{\alpha}}

\newcommand{\tunit}{\textType{unit}}
\newcommand{\tint}{\textType{int}}
\newcommand{\tbool}{\textType{bool}}
\newcommand{\trec}[1]{\mu #1}
\newcommand{\tforall}[1]{\forall #1}
\newcommand{\texists}[1]{\exists #1}
\newcommand{\tref}{\textType{ref}}

\newcommand{\EqType}{\mathrm{EqType}}

\newcommand{\none}{\bot}

\newcommand\bigcdot{\mathrel{\raisebox{1pt}{$\scriptscriptstyle\bullet$}}}
\newcommand\holed[1]{[\,#1\,]}
\newcommand\hole{\holed\bigcdot}
\newcommand\contextE[1]{\kern1pt\holed{#1}}

\newcommand{\typedctx}[7]{#1 : (#2 \mid #3 \vdash #4) \Rightarrow (#5 \mid #6 \vdash #7)}
\newcommand{\typed}[4]{\ifthenelse{\equal{#2}{}}{#1}{#1 \mid #2} \vdash #3 : #4}
\newcommand\doubleplus{+\kern-1.3ex+\kern0.8ex}

\newcommand{\specctx}{\textlog{specCtx}}
\newcommand{\specinv}{\textlog{spec\_inv}}
\newcommand{\tpto}{\mathbin{\Mapsto}}

\newcommand{\lkvar}{\mathit{lk}}
\newcommand{\newlock}{\mathsf{newlock}}
\newcommand{\acquire}{\mathsf{acquire}}
\newcommand{\release}{\mathsf{release}}
\newcommand{\Ticket}{\mathsf{TL}}
\newcommand{\newlockTicket}{\newlock_\Ticket}
\newcommand{\acquireTicket}{\acquire_\Ticket}
\newcommand{\releaseTicket}{\release_\Ticket}

\newcommand{\bitbool}{\textlog{bitbool}}
\newcommand{\bitnat}{\textlog{bitnat}}
\newcommand{\flipnat}{\textlog{flipnat}}
\newcommand{\tbit}{\textlog{TBit}}

\newcommand{\CGincrement}{\textlog{inc}_s}
\newcommand{\CGcounter}{\textlog{counter}_s}
\newcommand{\FGincrement}{\textlog{inc}_i}
\newcommand{\FGcounter}{\textlog{counter}_i}
\newcommand{\wkincr}{\textlog{wkincr}}
\newcommand{\counterread}{\textlog{read}}
\newcommand{\counterinv}{I_{\textlog{cnt}}}
\newcommand{\counterN}{\mathcal{N}}

\newcommand{\rand}{\mathsf{rand}}
\newcommand{\orFn}{\mathsf{or}}
\newcommand{\orOp}{\mathrel{\oplus}}

\newcommand{\newcoin}{\textlog{new\_coin}}
\newcommand{\flipcoin}{\textlog{flip}}
\newcommand{\readcoin}{\textlog{read}}
\newcommand{\newcoinLazy}{\textlog{new\_coin\_lazy}}
\newcommand{\flipcoinLazy}{\textlog{flip\_lazy}}
\newcommand{\readcoinLazy}{\textlog{read\_lazy}}
\newcommand{\newcoinLazyI}{\widehat{\newcoinLazy}}
\newcommand{\flipcoinLazyI}{\widehat{\flipcoinLazy}}
\newcommand{\readcoinLazyI}{\widehat{\readcoinLazy}}

\newcommand{\dinssingle}[2]{\left[#1:=#2\right]}
\DeclareDocumentCommand{\deltainsert}{m m m}
{\dinssingle{#1}{#2}\ifthenelse{\equal{#3}{}}{}{, #3}}
\newcommand{\vctxinsert}[3]{#1\mathop{:}#2, #3}
\newcommand{\rel}{R}
\newcommand{\tenv}{\Xi}
\newcommand{\pctx}{\mathcal{C}}

\newcommand{\lockInv}{\mathsf{lockInv}}
\newcommand{\lockInt}{\mathsf{lockInt}}

\newcommand{\isLock}[2]{\mathsf{isLock}_{\mathsf{s}}(#1, #2)}
\newcommand{\CLOSE}[2]{\mathsf{closeInv}_{#1}(#2)}
\newcommand{\islockI}{\mathsf{isLock}_{\mathsf{i}}}
\newcommand{\locked}[1]{\mathsf{locked}_{\mathsf{i}}(#1)}

\newcommand{\ticket}[2]{\mathsf{ticket}_{#1}({#2})}
\newcommand{\issuedTickets}[2]{\mathsf{issuedTickets}_{#1}({#2})}
\newcommand{\llo}{l_{\mathsf o}}
\newcommand{\lln}{l_{\mathsf n}}
\newcommand{\loname}{\gname_{\mathsf o}}
\newcommand{\lnname}{\gname_{\mathsf n}}
\newcommand{\lonamesp}{\namesp\!.{\mathsf o}}
\newcommand{\lnnamesp}{\namesp\!.{\mathsf n}}
\newcommand{\locknamesp}{\namesp\!.{\mathsf{inv}}}

\newcommand{\newProph}{\textlog{newproph}}
\newcommand{\resolveProph}[2]{\textlog{resolve}\spac #1\spac \textlog{to}\ #2}
\newcommand{\Proph}[2]{\textlog{proph}(#1, #2)}

\newcommand{\cnt}[3]{\textlog{cnt}_{#1}(#2, #3)}
\newcommand{\cntAuth}[2]{\textlog{cntAuth}_{#1}(#2)}
\newcommand{\Cnt}[2]{\textlog{isCnt}_{#2}(#1, \namesp)}
\newcommand{\CntN}[3]{\textlog{isCnt}_{#2}(#1, #3)}
\newcommand{\cntvar}{\mathit{c}}

\definecolor{precond}{HTML}{03aa07}
\definecolor{postcond}{HTML}{ec6262}

\crefname{thm}{Theorem}{Theorems}
\crefname{lem}{Lemma}{Lemmata}
\crefname{prop}{Proposition}{Propositions}
\crefname{defi}{Definition}{Definitions}

\begin{document}

\title[ReLoC Reloaded: A Mechanized Relational Logic for Concurrency]{ReLoC Reloaded: A Mechanized Relational Logic for Fine-Grained Concurrency and Logical Atomicity}

\author{Dan Frumin\rsuper{a}}
\address{\lsuper{a}University of Groningen and Radboud University}
\email{d.frumin@rug.nl}

\author{Robbert Krebbers\rsuper{b}}
\address{\lsuper{b}Radboud University and Delft University of Technology}
\email{mail@robbertkrebbers.nl}

\author{Lars Birkedal\rsuper{c}}
\address{\lsuper{c}Aarhus University}
\email{birkedal@cs.au.dk}

\maketitle

\begin{abstract}
We present a new version of \reloc: a relational separation logic for proving refinements of programs with higher-order state, fine-grained concurrency, polymorphism and recursive types.
The core of \reloc is its \emph{refinement judgment} $\logrel{}{\expr}{\expr'}{\type}$, which states that a program $\expr$ refines a program $\expr'$ at type $\type$.
\reloc provides type-directed structural rules and symbolic execution rules in separation-logic style for manipulating the judgment, whereas in prior work on refinements for languages with higher-order state and concurrency, such proofs were carried out by unfolding the judgment into its definition in the model.
\reloc's abstract proof rules make it simpler to carry out refinement proofs, and enable us to generalize the notion of logically atomic specifications to the relational case, which we call \emph{logically atomic relational specifications}.

We build \reloc on top of the Iris framework for separation logic in Coq, allowing us
to leverage features of Iris to prove soundness of \reloc, and to carry out refinement proofs in \reloc.
We implement tactics for interactive proofs in \reloc, allowing us to mechanize several case studies in Coq, and thereby demonstrate the practicality of \reloc.

\relocTwo extends \reloc (LICS'18) with various technical improvements, a new Coq mechanization, and support for Iris's prophecy variables.
The latter allows us to carry out refinement proofs that involve reasoning about the program's \emph{future}.
We also expand \reloc's notion of logically atomic relational specifications with a new flavor based on the HOCAP pattern by Svendsen \etal.

%%% Local Variables:
%%% mode: latex
%%% TeX-master: "reloc"
%%% End:

\end{abstract}

\section{Introduction}
\label{reloc:sec:intro}
A fundamental question in computer science is \emph{when two programs are equivalent}?
The ``golden standard'' of program equivalence is \emph{contextual equivalence}, stated directly in terms of the operational semantics.
Intuitively, expressions $\expr$ and $\expr'$ are contextually equivalent
%, notation $\expr \ctxequivsymb \expr' : \type$,
if no well-typed client can distinguish them, which formally means that for all well-typed contexts $\pctx$, the expression $\fillctx\pctx[\expr]$ has same observable behaviors as $\fillctx\pctx[\expr']$.
Contextual equivalence can be further decomposed into \emph{contextual refinement}.
An expression $\expr$ \emph{contextually refines} $\expr'$ %, notation $\expr \ctxrel \expr' : \type$,
if, for all contexts $\pctx$, if $\fillctx\pctx[\expr]$ has some observable behavior, then so does $\fillctx\pctx[\expr']$.
Expressions $\expr$ and $\expr'$ are contextually equivalent iff $\expr$ contextually refines $\expr'$ and \textit{vice versa}.

Contextual refinement and contextual equivalence have many applications in computer science.
One such application is to specify programs in terms of other programs.
For example, one can specify an implementation of a program module (say, a map) that internally uses an efficient but complicated data structure (say, a balanced search tree) by stating that it refines an implementation that internally uses an inefficient but easy to understand data structure (say, an unordered list).
In the context of a typed language that supports data abstraction, a specification of a program module in terms of refinement shows that clients of the program module cannot depend on the internal representation of data.
This can be seen as an instance of the \emph{representation independence} principle \cite{reynolds:1974,mitchell:1986}.

In the context of concurrency, contextual refinement is often used to specify a fine-grained concurrent program module by stating that it contextually refines a coarse-grained version.
This is similar to showing that a fine-grained program module is \emph{linearizable} \cite{herlihy:wing:1990,filipovic:ohearn:rinetzky:yang:2010}, \ie each fine-grained operation appears to take place instantaneously.
 A simple example is the specification of a fine-grained concurrent counter by a coarse-grained one, see \Cref{reloc:fig:impl} for the code.
The increment operation of the fine-grained version, $\FGcounter$, takes an ``optimistic'' lock-free approach to incrementing the value using a compare-and-set operation inside a loop.
If the value of the counter has been changed (for instance, by some other thread), then the fine-grained counter reattempts the increment from the beginning.
The increment operation of the coarse-grained version, $\CGcounter$, is performed inside a critical section guarded by a lock.
We can state the desired refinement as follows:
\begin{equation*}
\FGcounter \ctxrel \CGcounter : (\tunit \to \tint) \times (\tunit \to \tint).
%\label{reloc:eq:cnt_refinement}
\end{equation*}
Due to the instrumentation of the coarse-grained version with locks, this refinement expresses that each operation of the fine-grained version takes place instantaneously.
We will use the counter as a simple running example throughout the paper.

Another application of contextual refinement and contextual equivalence is to state algebraic properties of program constructs.
For example, let us consider the non-deterministic choice operator $\expr_1 \orOp \expr_2$, which non-deterministically executes the expression $\expr_1$ or $\expr_2$.
Using contextual equivalence, we can state that this operator is
commutative (${\expr_1 \orOp \expr_2} \ctxequivsymb {\expr_2 \orOp \expr_1}$),
associative (${\expr_1 \orOp (\expr_2 \orOp \expr_3)} \ctxequivsymb {(\expr_1 \orOp \expr_2) \orOp \expr_3}$),
and that sequential composition distributes over the operator
$({(\expr_1 \orOp \expr_2) ; \expr_3} \ctxequivsymb {(\expr_1 ; \expr_3) \orOp (\expr_2 ; \expr_3)}$).

\begin{figure}
\begin{align*}
\intertext{\raggedright\textbf{Fine-grained grained version (\ie the implementation)}:}
\counterread \eqdef{}& \Lam \cntvar. \deref \cntvar \\
\FGincrement \eqdef{}& \Rec {\lvar{inc}} \cntvar =
	\begin{array}[t]{@{} l}
  \Let n = \deref \cntvar in \\
  \If \CAS{\cntvar}{n}{1 + n} then n \Else \lvar{inc}\ \cntvar
  \end{array} \\
\FGcounter \eqdef{}& \Let \cntvar = \Ref(0) in ((\Lam \unittt.\counterread\ \cntvar),(\Lam \unittt. \FGincrement\ \cntvar)) \\
\intertext{\raggedright\textbf{Coarse-grained version (\ie the specification)}:}
\CGincrement \eqdef{}& \Lam \cntvar\, l.
	%\begin{array}[t]{@{} l}
	\acquire\ l;\  %\\
	\Let n = \deref \cntvar in \cntvar \gets (1 + n);\
	\release\ l;\ n
	%\end{array}
	\\
\CGcounter \eqdef{}& \Let l = \newlock\ \unittt in
  \Let \cntvar = \Ref(0) in
  ((\Lam \unittt.\counterread\ \cntvar),(\Lam \unittt. \CGincrement\ \cntvar\ l))
\end{align*}
\caption{A fine-grained and coarse-grained concurrent counter.
(Note that the $\counterread$ operation is shared by both.)}
\label{reloc:fig:impl}
\end{figure}

\subsection*{Proving contextual refinement and contextual equivalence}
Contextual refinement $\expr \ctxrel {\expr'} : \type$ (and contextual equivalence $\expr \ctxequivsymb {\expr'} : \type$) are very strong notions because they relate the expressions $\expr$ and $\expr'$ in \emph{any} well-typed context $\pctx$ with a hole of type $\type$.
As a consequence, proving contextual refinement and equivalence directly is challenging---one has to consider arbitrary contexts $\pctx$, which are only known to be well-typed.
Contextual refinement and equivalence are therefore typically proved indirectly using approaches based on bisimulations (\eg~\cite{gordon:1999,pitts:2000,koutavas:wand:2006,sumii:pierce:2007}) or logical relations (\eg~\cite{pitts:2005,ahmed:2006,dreyer:ahmed:birkedal:2009,birkedal:Sieczkowski:thamsborg:2012,turon:thamsborg:ahmed:birkedal:dreyer:2013}).
In the present paper we focus on approaches based on logical relations because they scale well to increasingly rich programming languages with features such as impredicative polymorphism, recursive types, higher-order state, and fine-grained concurrency.
% bisimulation-based approaches do not scale very well, and logical relation-based approaches can be quite intricate.

In the approaches based on logical relations, the key is a notion of \emph{logical refinement}, notation $\logrel {} \expr {\expr'} \type$.
Logical refinement is defined by structural recursion over the type $\type$, rather than by quantification over all contexts.
The soundness theorem of logical relations states that logical refinement implies contextual refinement, \ie that $\logrel {} \expr {\expr'} \type$ implies $\expr \ctxrel {\expr'} : \type$.
As a result, proving contextual refinement can be reduced to proving logical refinement, which is generally much easier.

Unfortunately, it is difficult to construct a suitable notion of logical refinement when considering language features like recursive types and higher-order state.
In the presence of (general) recursive types, no structurally-recursive definition over the type exists, and in the presence of higher-order references, one needs some notion of \emph{recursively-defined worlds}~\cite{birkedal:etal:11}.
The technique of \emph{step-indexing}~\cite{ahmed:appel:virga:2002,ahmed:thesis} has been used to stratify the definitions using recursion over a natural number, called the \emph{step-index}, which corresponds to the number of computation steps performed by the program.

Step-indexing has shown to be very effective by a large body of work on step-indexed logical relations, \eg~\cite{neis:dreyer:rossberg:2011,hur:dreyer:2011,birkedal:Sieczkowski:thamsborg:2012,cicek:paraskevopoulou:garg:2016,rajani:garg:2018}.
However, definitions and proofs are intricate because step-indices appear practically everywhere---they even appear in definitions and proofs related to language features (say, products or sums) for which step-indexing is orthogonal.
Dreyer \etal thus proposed the ``logical approach'' to logical relations~\cite{dreyer:ahmed:birkedal:2009,dreyer:neis:rossberg:birkedal:2010} to hide step-indices by abstracting and internalizing them in a modal logic using the \emph{later} modality ($\later$)~\cite{appel:mellies:richards:vouillon:2007}.
Turon \etal~\cite{turon:thamsborg:ahmed:birkedal:dreyer:2013,CaReSL} further developed the logical approach by using separation logic~\cite{ohearn:reynolds:yang:2001,ohearn:2007,brookes:2007} to abstract over program states and to handle (fine-grained) concurrency.

More recently, Krebbers \etal~\cite{irisIPM} and Timany~\cite{amin:thesis} defined a logical relation for program refinement based on the work of Turon \etal in the state-of-the-art higher-order concurrent separation logic Iris~\cite{irisJFP,iris1,iris2,iris3}.
Iris supports impredicative invariants~\cite{svendsen:birkedal:2014} and used-defined ghost state, which can be used to streamline the definition of the logical relation, and to carry out proofs of challenging program refinements.
The meta theory of Iris is mechanized in the Coq proof assistant, and Iris comes equipped with a \emph{proof mode}~\cite{irisIPM,MoSeL}---an extensive set of Coq tactics for separation logic proofs---which allowed them to mechanize all their results in Coq.

\subsection*{Problem statement and key idea}
To prove refinements of complicated program modules in a scalable fashion, it is important to decompose refinement proofs into smaller refinements that can be proved in isolation.
As a simple example, let us consider the refinement of the fine-grained and coarse-grained concurrent counter from \Cref{reloc:fig:impl}:
\begin{equation*}
\FGcounter \ctxrel \CGcounter : (\tunit \to \tint) \times (\tunit \to \tint).
%\label{reloc:eq:cnt_refinement}
\end{equation*}
We wish to decompose the proof of this refinement into refinements for the read and increment operations.
Naively, one might consider proving contextual refinements for these operations.
Unfortunately, such contextual refinements do not hold---they only hold \emph{conditionally} under the assumption that the internal state in both of the implementations is related (including the state of the lock used by the coarse-grained version).
%This shows that contextual refinement is not a suitable notion to compose refinement proofs of individual operations of a program module.

Instead of performing composition at the level of contextual refinement, our key idea is to perform composition at the level of logical refinement.
By generalizing logical refinement to become an internal (\ie first-class) notion in (the Iris) separation logic, we can use the connectives of separation logic to express conditional refinements.
Logical refinements for the operations of the concurrent counter are as follows:
\begin{align*}
  \knowInv{}{\counterinv} \wand{}&
  \logrel{}{(\Lam \unittt. \counterread\ c_i)}{(\Lam \unittt. \counterread\ c_s)}{\tunit \to \tint}
   \\
  \knowInv{}{\counterinv} \wand{}&
  \logrel{}{(\Lam \unittt.\FGincrement\ c_i)}{(\Lam \unittt.\CGincrement\ c_s\ \lkvar)}{\tunit \to \tint}.
\end{align*}
We use the \emph{magic wand} ($\wand$, also known as \emph{separating implication}) to make these refinements conditional under the invariant $\counterinv$ (expressed using Iris's invariant connective $\knowInv {} I$), which is defined as $\counterinv \eqdef \Exists n \in \nat. c_i \mapstoI n * c_s \mapstoS n * \isLock{\lkvar}{\<false>}$.
The invariant $\counterinv$ intuitively expresses that in between function calls, the values of both counters are equal, and the lock (used in the coarse-grained implementation) is in unlocked state.
With logical refinements for the individual operations at hand, we can compose them into the logical refinement ${\logrel {} \FGcounter \CGcounter {(\tunit \to \tint) \times (\tunit \to \tint)}}$, which using soundness gives us the desired contextual refinement ${\FGcounter \ctxrel \CGcounter : (\tunit \to \tint) \times (\tunit \to \tint)}$.

Treating logical refinement as an internal notion in separation logic succinctly distinguishes our work from prior work.
In prior work on refinements for rich languages, \eg the aforementioned work by Turon \etal~\cite{turon:thamsborg:ahmed:birkedal:dreyer:2013,CaReSL}, Krebbers \etal~\cite{irisIPM}, and Timany~\cite{amin:thesis}, logical refinement is an external notion (\ie a proposition in ordinary mathematics, rather than in separation logic), which means that one cannot concisely state refinements that are conditional on the program state.
To state and prove such refinements, one needs to unfold the definition of the logical refinement into the model.

Apart from being able to decompose refinement proofs, internalizing logical refinement gives us a number of other tangible benefits.
First, it allows us to develop type-directed structural rules and symbolic execution rules for proving logical refinements.
Our symbolic execution rules closely resemble the typical rules for symbolic execution in separation logic, but come in two forms: for the program on the left-hand side and right-hand side of the refinement, making it possible to write concise proofs.

Second, by internalizing logical refinement we can state logical refinements that apply to the situation when the expression on the one side of the refinement contains a program subject to specification, while the expression on the other side is arbitrary.
We call such specifications \emph{relational specifications}.
Relational specifications take the ability to decompose refinement proofs one step further.
As a simple example, let us consider the example from \Cref{reloc:fig:impl}, where we proved that a fine-grained concurrent counter refines a coarse-grained version.
This refinement is insufficient if we want to prove that a program module that \emph{uses internally} the fine-grained counter (say, a ticket lock) refines another module that does not use the coarse-grained counter (say, a spin lock).
However, we can instead formulate a \emph{relational specification} for the program module that is proven just once, and derive different logical (and thus by soundness, contextual) refinements from it.

A key challenge in stating relational specifications for operations is to concisely capture that they behave as-if they were atomic, \ie they appear to take place instantaneously.
There has been a long line of work on \emph{logically atomic specifications} to reason about atomicity in the context of Hoare-style logics~\cite{jacobs:piessens:2011,HOCAP,daRochaPinto:dinsdale-young:gardner:2014,iris1,irisProph}.
We show that such logically atomic specifications generalize to the relational case, and call them \emph{logically atomic relational specifications}.
Concretely, we introduce relational specification patterns based on da Rocha Pinto \etal's TaDA-style~\cite{daRochaPinto:dinsdale-young:gardner:2014} and Svendsen \etal's HOCAP-style~\cite{HOCAP} logically atomic specifications.

\subsection*{The \reloc logic}

Based on the previously described key ideas, we develop a relational separation logic called \textbf{\reloc}.
\reloc is built on top of Iris, allowing the user to leverage the features of Iris such as invariants, (higher-order) ghost state, and prophecy variables.
Invariants and ghost state state are powerful mechanisms that support reasoning about concurrent programs through used-defined protocols.
Prophecy variables~\cite{abadi:lamport:1991,irisProph} allow for speculative reasoning about the future state of concurrent programs.
In Iris they come in the form of ghost variables whose value can be referenced before they are specified, thus allowing one to ``prophesize'' their potential value.
We show how these features can be used in \reloc to prove challenging refinements.

We have implemented \reloc as a shallow embedding on top of Iris in Coq~\cite{irisIPM,MoSeL}.
In addition to mechanizing all meta-theoretic results of \reloc, like its soundness theorem, we have implemented new tactics that support mechanized interactive reasoning about program refinements in \reloc in a practical and modular way.
To our knowledge, \reloc is the first fully mechanized relational logic enabling reasoning about contextual refinements of programs in a fine-grained concurrent higher-order imperative programming language.
The mechanization can be found at~\cite{appendix}.

\subsection*{Contributions and structure of the paper}%
\begin{itemize}
\item We present a relational logic \textbf{\reloc} for reasoning about contextual refinements of fine-grained concurrent higher-order imperative programs.
We present our target programming language (\Cref{reloc:sec:language}), an overview of \reloc (\Cref{reloc:sec:tour}), and a detailed description of its type-directed structural rules and symbolic execution rules (\Cref{reloc:sec:rules}).
\item We introduce relational specification patterns based on TaDA~\cite{daRochaPinto:dinsdale-young:gardner:2014} and  HOCAP-style~\cite{HOCAP} logically atomic specifications (\Cref{reloc:sec:atomicity}).
\item We show how to integrate \emph{prophecy variables} into ReLoC, thereby enabling speculative reasoning in proofs of program refinements (\Cref{reloc:sec:prophecies}).
\item We describe the logical relations model of \reloc in Iris (\Cref{reloc:sec:model}).
\item We describe the mechanization of \reloc in Coq, and explain how we support mechanized interactive reasoning in \reloc in a practical and modular way (\Cref{reloc:sec:formalization}).
\end{itemize}
We discuss further related work in \Cref{reloc:sec:related-work} and conclude in \Cref{reloc:sec:conclusion}.

In addition to the case studies presented in this paper, we have also verified a collection of refinements of concurrent programs from the literature.
We give a brief overview of these examples in \Cref{reloc:sec:other_examples}; and the proofs can be found in the accompanying Coq sources.

\subsection*{Differences with the conference version of this paper}
In the conference version of this paper~\cite{reloc} we described the first version of \reloc.
This paper extends the conference paper in two ways.
First, we introduce \relocTwo (in this paper referred to as just \reloc), which has several new features, especially in terms its Coq mechanization.
Second, we have expanded the presentation of, as well as the material covered by, the paper significantly.
Concretely, \relocTwo has the following new features compared to its original version:
\begin{itemize}
\item \relocTwo's primitive refinement judgment $\logrel{}{\expr}{\expr'}{\type}$ is defined for closed expressions (\ie without free variables), and the version for open expressions (\ie with free variables) is a derived notion (see \Cref{reloc:def:open_term_refinement}).
\item \relocTwo's underlying programming language is \HeapLang---the default language of Iris's Coq mechanization.
  By having a tight integration of \reloc with Iris's Coq ecosystem we managed to reuse more Coq code and integrate novel Iris features.
\item One such feature that we have integrated into \relocTwo is the support for prophecy
  variables (\Cref{reloc:sec:prophecies}), which was recently added to Iris~\cite{irisProph}.
\end{itemize}

Compared to the conference paper, we have significantly expanded \Cref{reloc:sec:language,reloc:sec:rules,reloc:sec:formalization}, and added \Cref{reloc:sec:prophecies,reloc:sec:model,reloc:sec:escape-hatch}, which are completely new.
We have extended \Cref{reloc:sec:atomicity} with HOCAP-style specifications, which we put into action by verifying a refinement between a ticket lock and a spin lock in \Cref{reloc:sec:ticket_lock}.

%%% Local Variables:
%%% mode: latex
%%% TeX-master: "reloc"
%%% End:

\section{The programming language}
\label{reloc:sec:language}
We consider a typed version of \HeapLang, the default language that is shipped with Iris's Coq development~\cite{irisWWW}.
\HeapLang is a call-by-value $\lambda$-calculus, with higher-order references, fork-based unstructured concurrency, and atomic operations for fine-grained concurrency, equipped with System-$\mathsf{F}$-style types.
The syntax is shown in \Cref{reloc:fig:language}.
We let $\tvar$ range over a countably infinite set $\TVar$ of type variables, which can be bound by the universal type $\tforall{\tvar. \type}$, existential type $\texists{\tvar . \type}$, and recursive type $\trec{\tvar. \type}$.
We omit the usual Boolean and arithmetic operations such as addition, multiplication, equality, negation.

\begin{figure}[t]
\begin{align*}
\type \in \Type \bnfdef{}&
  \tvar \ALT
  \tunit \ALT
  \tbool \ALT
  \tint \ALT
  \type_1 \times \type_2 \ALT
  \type_1 + \type_2 \ALT
  \type_1 \to \type_2 \ALT
  \tforall{\tvar. \type} \ALT
  \texists{\tvar. \type} \ALT
  \trec{\tvar. \type} \ALT
  \tref\ \type \\
\val \in \Val \bnfdef{}&
  i \ALT
  \loc \ALT
  \<true> \ALT
  \<false> \ALT
  (\val_1, \val_2) \ALT
  \<inl>(\val) \ALT
  \<inr>(\val) \ALT
  \Rec f \var = \expr \qquad i \in \mathbb{Z}, \loc \in \Loc \\ & \ALT
  \tlam{\expr}\ALT
  \<pack>\ \val \ALT
  \<fold>\ \val \\
\expr \in \Expr \bnfdef{}&
	\var \ALT
  \val \ALT
  \If \expr then \expr_1 \Else \expr_2 \ALT
  (\expr_1, \expr_2) \ALT
  \proj_i(\expr) \ALT
  \<inl>(\expr) \ALT
  \<inr>(\expr) \qquad i \in \{1, 2\} \\ & \ALT
  (\Match \expr with \<inl>(\var) \to \expr_1 \mid \<inr>(\var) \to \expr_2) \ALT
  \expr_1 (\expr_2) \ALT
  \tapp{\expr} \\ & \ALT
  \<pack>(\expr) \ALT
  \unpack{\expr_1}{\Ret \var. \expr_2} \ALT
  \<fold>\ \expr \ALT
  \<unfold>\ \expr \\ & \ALT
  \Ref(\expr) \ALT
  \deref \expr \ALT
  \expr_1 \gets \expr_2 \ALT
  \CAS{\expr_1}{\expr_2}{\expr_3} \ALT
  \Fork{\expr} \ALT
  \dots
\end{align*}
\caption{The syntax of the \HeapLang language.}
\label{reloc:fig:language}
\end{figure}

Most of the operations are standard, so we only discuss some subtleties.
Type abstraction $\tlam{\expr}$, type application $\tapp{\expr}$, and the $\<pack>$/$\<unpack>$ constructs for packing/unpacking existential types do not contain type annotations, following \eg~\cite{ahmed:2006}.
The $\<fold>$/$\<unfold>$ constructs are used to fold/unfold iso-recursive types.
The language includes standard operation on references $\Ref(\expr)$ for allocation, $\deref\expr$ for dereferencing, and $\expr_1 \gets \expr_2$ for assignment.
The \emph{atomic} compare-and-set operation $\CAS{\expr_1}{\expr_2}{\expr_3}$ checks if the value stored at the location $\expr_1$ is equal to $\expr_2$, and, if so, sets the value at $\expr_1$ to $\expr_3$.
The $\Fork\expr$ construct creates a new thread, which will execute the expression $\expr$.

\begin{figure}
\raggedright
\textbf{Selected typing rules}:
\begin{mathpar}
\inferH{var-typed}
  {\vctx(\var) = \type}
  {\typed{\tenv}{\vctx}{\var}{\type}}
\and
\inferH{proj-typed}
  {\typed{\tenv}{\vctx}{\expr}{\type_1 \times \type_2} \and i \in \{1, 2\}}
  {\typed{\tenv}{\vctx}{\proj_i(\expr)}{\type_i}}
\and
\inferH{rec-typed}
  {\typed{\tenv}{\vctxinsert{x}{\type_1}{\vctxinsert{f}{\type_1\to\type_2}{\vctx}}}{\expr}{\type_2}}
  {\typed{\tenv}{\vctx}{\Rec f x={\expr}}{\type_1 \to \type_2}}
\and
\inferH{tlam-typed}
  {\typed{\tenv, \tvar}{\vctx}{\expr}{\type} }
  {\typed{\tenv}{\vctx}{\tlam{\expr}}{\tforall{\tvar.\type}}}
\and
\inferH{tapp-typed}
  {\typed{\tenv}{\vctx}{\expr}{\tforall{\tvar.\type}}
   \and \tenv \vdash \type'}
  {\typed{\tenv}{\vctx}{\tapp{\expr}}{\subst{\type}{\tvar}{\type'}}}
\and
\inferH{tpack-typed}
  {\typed{\tenv}{\vctx}{\expr}{\subst{\type}{\tvar}{\type'}}}
  {\typed{\tenv}{\vctx}{\<pack>\ \expr}{\texists{\tvar.\type}}}
\and
\inferH{tunpack-typed}
  {\typed{\tenv}{\vctx}{\expr_1}{\texists{\tvar.\type_1}}
   \and
   \typed{\tvar, \tenv}{\vctxinsert{\var}{\type_1}{\vctx}}{\expr_2}{\type_2}
   \and \mbox{$\tvar$ is not free in $\vctx$ or $\type_2$}}
  {\typed{\tenv}{\vctx}{\unpack{\expr_1}{\Ret \var. \expr_2}}{\type_2}}
\and
\inferH{fold-typed}
  {\typed{\tenv}{\vctx}{\expr}{\subst{\type}{\tvar}{\trec{\type}}}}
  {\typed{\tenv}{\vctx}{\<fold>\ \expr}{\trec{\tvar . \type}}}
\and
\inferH{unfold-typed}
  {\typed{\tenv}{\vctx}{\expr}{\trec{\tvar. \type}}}
  {\typed{\tenv}{\vctx}{\<unfold>\ \expr}{\subst{\type}{\tvar}{\trec{\tvar. \type}}}}
\and
\inferH{alloc-typed}
  {\typed{\tenv}{\vctx}{\expr}{\type}}
  {\typed{\tenv}{\vctx}{\Ref(\expr)}{\tref\ \type}}
\and
\inferH{load-typed}
  {\typed{\tenv}{\vctx}{\expr}{\tref\ \type}}
  {\typed{\tenv}{\vctx}{\deref\expr}{\type}}
\and
\inferH{store-typed}
  {\typed{\tenv}{\vctx}{\expr_1}{\tref\ \type}
   \and \typed{\tenv}{\vctx}{\expr_2}{\type}}
  {\typed{\tenv}{\vctx}{\expr_1 \gets \expr_2}{\tunit}}
\and
\inferH{cas-typed}
  {\typed{\tenv}{\vctx}{\expr_1}{\tref\ \type}
   \quad \typed{\tenv}{\vctx}{\expr_2}{\type}
   \quad \typed{\tenv}{\vctx}{\expr_3}{\type}
   \quad \EqType(\type)}
  {\typed{\tenv}{\vctx}{\CAS{\expr_1}{\expr_2}{\expr_3}}{\tbool}}
\and
\inferH{fork-typed}
  {\typed{\tenv}{\vctx}{\expr}{\tunit}}
  {\typed{\tenv}{\vctx}{\Fork{\expr}}{\tunit}}
\end{mathpar}

\medskip
\textbf{Selected rules of pure reduction}
  $\pureexec {\expr_1} {\expr_2}$
  \textbf{and thread-local call-by-value head-reduction}
  $(\expr,\stateS) \hstep (\expr',\stateS')$:
\medskip
\begin{mathpar}
\axiomH{proj}
  {\pureexec {\proj_i\ (\val_1, \val_2)} {\val_i}}
\and
\axiomH{beta}
  {\pureexec {(\Rec f x = e)\ \val}
             {\subst{\subst{e}{x}{v}}{f}{\Rec f x = e}}}
\and
\axiomH{tbeta}
  {\pureexec {\tapp{(\tlam{\expr})}} \expr}
\and
\axiomH{unpack}
  {\pureexec {\unpack{(\<pack>\ \val)}{\Ret \var. \expr}}{\subst{\expr}{\var}{\val}}}
\and
\axiomH{unfold}
  {\pureexec {\<unfold>\ (\<fold>\ \val)} \val}
\and
\inferH{pure}
  {\pureexec {\expr_1} {\expr_2}}
  {(\expr_1, \stateS) \hstep (\expr_2, \stateS)}
\and
\inferH{alloc}
  {\stateS(\loc) = \none}
  {(\Ref(\val), \stateS) \hstep (\loc, \mapinsert{\loc}{\val}{\stateS})}
\and
\inferH{deref}
  {\stateS(\loc) = \val}
  {(\deref \loc, \stateS) \hstep (\val, \stateS)}
\and
\inferH{store}
  {\stateS(\loc) = \val}
  {(\loc \gets \val', \stateS) \hstep
  (\unittt, \mapinsert{\loc}{\val'}{\stateS})}
\and
\inferH{cas-fail}
  {\stateS(\loc) \neq \val_1}
  {(\CAS{\loc}{\val_1}{\val_2}, \stateS) \hstep (\<false>, \stateS)}
\and
\inferH{cas-suc}
  {\stateS(\loc) = \val_1}
  {(\CAS{\loc}{\val_1}{\val_2}, \stateS) \hstep (\<true>, \mapinsert{\loc}{\val_2}{\stateS})}
\end{mathpar}

\medskip
\textbf{Thread-pool reduction}
\((\vec\expr,\stateS) \tpstep (\vec{\expr'},\stateS')\):
\medskip
\begin{mathpar}
\infer
  {(\expr,\stateS) \hstep (\expr',\stateS')}
  {(\vec{\expr_1}\ \fillctx \lctx[\expr]\ \vec{\expr_2}, \stateS)
    \tpstep (\vec{\expr_1}\ \fillctx \lctx[\expr']\ \vec{\expr_2}, \stateS')}
\and
\infer
  {}
  {(\vec{\expr_1}\ \fillctx\lctx[\Fork{\expr}]\ \vec{\expr_2}, \stateS)
     \tpstep
   (\vec{\expr_1}\ \fillctx\lctx[\unittt]\ \vec{\expr_2}\ \expr, \stateS)}
\end{mathpar}
\caption{The type system and operational semantics of \HeapLang.}
\label{reloc:fig:heaplang}
\end{figure}

\paragraph{Syntactic sugar.}
We use syntactic sugar to define non-recursive functions, let-bindings, and sequential composition.
We let $(\Lam \var. \expr) \eqdef (\Rec \_ \var = \expr)$ and
$(\Let \var = \expr_1 in \expr_2) \eqdef ((\Lam \var. \expr_2)\ \expr_1)$ and
$(\expr_1 ; \expr_2) \eqdef (\Let \_ = \expr_1 in \expr_2)$.
The underscore $\_$ denotes an anonymous binder, \ie a fresh variable that is unused in the body of the binding expression.

\paragraph{Type system.}
Typing judgments take the form \(\typed \tenv \vctx \expr \type\),
where \(\vctx\) is a context assigning types to program variables, and $\tenv$ is a context of type variables.
The inference rules for the typing judgments are standard; a selection of representative rules is given in \Cref{reloc:fig:heaplang}.
The typing rule for the compare-and-set ($\<CAS>$) operation has a side-condition $\EqType(\type)$, which ensures that a compare-and-set can only be performed on word-sized data types, \ie the unit, Boolean, integer, and reference type.

\paragraph{Operational semantics.}
The operational semantics involves three reduction relations:
pure head reduction $\ra_{\mathsf{pure}}$,
thread-local head reduction \(\hstep\), and thread-pool reduction \(\tpstep\), see
\Cref{reloc:fig:heaplang} for the rules.
Head reduction \(\hstep\) is lifted to thread-pool reduction \(\tpstep\) using standard \emph{call-by-value evaluation contexts} (in the style of Felleisen and Hieb \cite{felleisen:hieb:1992}):
\[
\lctx \in \Ectx \bnfdef{}
  \hole \ALT
  \expr_1(\lctx) \ALT
  \lctx(\val_2) \ALT
  \expr_1 \gets \lctx \ALT
  \lctx \gets \val_2 \ALT \dots
\]
Thread-pool reduction \(\tpstep\) is defined on configurations \(\cfgg = (\vec\expr,\stateS)\)
consisting of a state \(\stateS\) (a finite partial map from locations to
values) and a thread-pool \(\vec\expr\) (a list of expressions corresponding to the
threads)
by interleaving, \ie by picking a thread and executing it, thread-locally, for one step.
The only special case is \(\Fork{\expr}\),
which spawns a thread \(\expr\), and reduces itself to the unit value \(\unittt\).

\paragraph{Contextual refinement.}
The notion of contextual refinement that we use is standard (see, \eg \cite{pitts:2005} or \cite[Chapters 46 \& 47]{harper:pfpl}).
It formalizes the situation when the set of observations that can be made about the first program is a subset of observations that can be made about the second program.
An observation about a program are made using a \emph{program context} $\pctx$, which is a program with a hole:
\begin{align*}
\pctx \in \Pctx \bnfdef{}&
	\Box \ALT
	\Rec f \var = \pctx \ALT
	\pctx (\expr_2) \ALT
	\expr_1 (\pctx) \ALT
	\tlam\pctx \ALT
	\tapp{\pctx} \ALT
  \dots
\end{align*}
Since we are in a typed setting, we consider only \emph{typed contexts}.
A program context is well-typed, denoted as $\typedctx{\pctx}{\tenv}{\vctx}{\type}{\tenv'}{\vctx'}{\type'}$, if
for any term $\typed{\tenv}{\vctx}{\term}{\type}$ we have $\typed{\tenv'}{\vctx'}{\fillctx\pctx[\term]}{\type'}$.
The typing relation on contexts is standard, and can be derived from the typing rules in \Cref{reloc:fig:heaplang}.

We then define contextual refinement as follows.
An expression $\expr_1$ \emph{contextually refines} an expression $\expr_2$ at type $\type$, denoted as $\ctxref{\tenv \mid \vctx}{\expr_1}{\expr_2}{\type}$, if no well-typed \emph{program context} $\pctx$ resulting in a closed program can distinguish the two:
\begin{align*}
\ctxref{\tenv \mid \vctx}{\expr_1}{\expr_2}{\type} \eqdef{}&
  \All \type'\,
  (\typedctx{\pctx}{\tenv}{\vctx}{\type}{\emptyset}{\emptyset}{\type'})\;
    \val\; \vec{\expr_f}\; \stateS. \\& \qquad\qquad
  (\fillctx\pctx[\expr_1], \emptyset) \tpstep^{\ast} (\conS{\val}{\efs}, \stateS) \implies \\ & \qquad\qquad
    \Exists \val'\; {\vec{\expr'_{f}}}\; \stateS'. (\fillctx\pctx[\expr_2], \emptyset) \tpstep^{\ast} (\conS{\val'}{\vec{\expr'_{f}}}, \stateS').
\end{align*}
\emph{Contextual equivalence} $\ctxequiv{\tenv \mid \vctx}{\expr_1}{\expr_2}{\type}$ is defined as the symmetric closure of contextual refinement, \ie
$(\ctxref{\tenv \mid \vctx}{\expr_1}{\expr_2}{\type}) \wedge (\ctxref{\tenv \mid \vctx}{\expr_2}{\expr_1}{\type})$.

Note that contextual refinement only takes termination into account, and does not require the resulting values $\val$ and $\val'$ to be equal.
Demanding the equality on the resulting values would make contextual refinement too strong.
For example, the terms $(\Lam x. x + 1)$ and $(\Lam x. 1 + x)$ of function type would not be deemed contextually equivalent, because they terminate to syntactically different values in the empty program context.

There are, however, equivalent formulations of contextual refinement which equate the resulting values $\val$ and $\val'$.
In order to do that, it is necessary to restrict the typed context $\typedctx{\pctx}{\tenv}{\vctx}{\type}{\emptyset}{\emptyset}{\type'}$ to those for which $\type'$ is a directly observable type, like Booleans or integers.
For example, we could have used the following equivalent\footnote{Proving that this definition is equivalent to the one presented earlier is not a very complicated, albeit laborious, task. See the Coq mechanization for the formal proof.} definition (a variation of $\<true>$-adequate contextual equivalence from \cite[Exercise 7.5.10]{pitts:2005}):
\begin{align*}
% \ctxref{\tenv \mid \vctx}{\expr_1}{\expr_2}{\type} \eqdef{}&
&  \All (\typedctx{\pctx}{\tenv}{\vctx}{\type}{\emptyset}{\emptyset}{\tbool})\;
    \vec{\expr_f}\; \stateS. \\& \qquad\qquad
  (\fillctx\pctx[\expr_1], \emptyset) \tpstep^{\ast} (\conS{\<true>}{\efs}, \stateS) \implies
    \Exists {\vec{\expr'_{f}}}\; \stateS'. (\fillctx\pctx[\expr_2], \emptyset) \tpstep^{\ast} (\conS{\<true>}{{\vec{\expr'_{f}}}}, \stateS').
\end{align*}

%%% Local Variables:
%%% mode: latex
%%% TeX-master: "reloc"
%%% End:

%  LocalWords:  CBV

\section{A tour of \reloc}
\label{reloc:sec:tour}
This section gives a tour of \reloc by demonstrating its key logical connectives and proof rules.
We first describe \reloc's grammar, soundness statement, and rule format (\Cref{reloc:subsec:grammar}).
After that, we put \reloc to action by proving contextual refinements of two program modules.
The first is a bit module, which demonstrates \reloc's type-directed structural rules and symbolic execution rules for reasoning about pure programs (\Cref{reloc:subsec:representation_independence}).
The second is the concurrent counter module from \Cref{reloc:sec:intro}, which involves reasoning about internal state and concurrency.
Specifically we demonstrate how \reloc is used to reason about stateful programs
using symbolic execution (\Cref{reloc:sec:counter:symbolic_execution}),
concurrency using invariants (\Cref{reloc:sec:tour:invariants}),
and recursive functions and loops using L\"ob induction (\Cref{reloc:sec:counter:loeb_induction}).

\subsection{Grammar and soundness}
\label{reloc:subsec:grammar}
\reloc is based on higher-order intuitionistic separation logic, and the grammar of its propositions is:
\begin{align*}
\propMV, \propMVB \in \Prop \bnfdef{}&
	\TRUE \ALT \FALSE \ALT
	\All \var.\propMV \ALT
	\Exists \var.\propMV \ALT
        \propMV \wedge \propMVB \ALT
        \propMV \vee \propMVB \ALT
        \propMV \implies \propMVB \\ & \ALT
	\propMV * \propMVB \ALT
	\propMV \wand \propMVB \ALT
	\loc \mapstoI \val \ALT
	\loc \mapstoS \val \ALT
	(\logrel[\mask]{\Delta}{\expr_1}{\expr_2}{\type}) \\ & \ALT
        \logrelV{\type}{\Delta}{\val_1}{\val_2} \ALT
	\knowInv\namesp\propMV \ALT
	\later \propMV \ALT
	\always \propMV \ALT
	\pvs[\mask_1][\mask_2] \propMV \ALT \dots
\end{align*}

\reloc is an extension of Iris and therefore includes all connectives of Iris, in particular, the \emph{later} modality $\later$, \emph{persistence} modality $\always$, \emph{update} modality $\pvs[\mask_1][\mask_2]$, and \emph{invariant assertion} $\knowInv \namesp \propMV$.
We introduce these connectives in passing throughout this section.
Some of these connectives are annotated by \emph{invariant masks} $\mask \subseteq \InvName$ and \emph{invariant names} $\namesp \in \InvName$, which are needed for bookkeeping related to Iris's invariant mechanism.
Until we introduce invariants in \Cref{reloc:sec:tour:invariants}, we will omit these annotations.
Similarly, we will ignore the later modality $\later$ until we explain it in \Cref{reloc:sec:counter:loeb_induction}.

An essential difference to vanilla Iris is that \reloc has internal (or first-class) \emph{refinement judgments} $\logrel{\Delta}{\expr_1}{\expr_2}{\type}$, which should be read as ``the expression $\expr_1$ refines the expression $\expr_2$ at type $\type$''.
Just like contextual refinement, the refinement judgment in \reloc is indexed by a type $\type$.
The judgment contains an environment $\Delta$ which assigns \emph{interpretations} to type variables.
These interpretations are given by an Iris relation of type %
$\Val \times \Val \to \Prop$.
One such kind of relation, the \emph{value interpretation} relation $\logrelV{\type}{\Delta}{-}{-} : \Val \times \Val \to \Prop$ (for each syntactic type $\type$ of \HeapLang) will be discussed in \Cref{reloc:sec:rules}.
We elide the contexts $\Delta$ in refinement judgments whenever they are empty.

The intuitive meaning of $\logrel{\Delta}{\expr_1}{\expr_2}{\type}$ is that $\expr_1$ is safe, and all of its behaviors can be simulated by $\expr_2$.
It is a simulation in the sense that any execution step of $\expr_1$ can be matched by a (possibly empty) sequence of execution steps of $\expr_2$.
Borrowing the terminology from languages with non-determinism, we think of $\expr_1$ as being \emph{demonic} and $\expr_2$ as being \emph{angelic}.
That is, the non-deterministic choices of $\expr_1$ (\eg scheduling of forked-off threads) are selected by an external demon; whereas for the non-deterministic choices of $\expr_2$, an angle blesses the person proving the refinement with an ability to select a choice themselves.

Since we often use refinement judgments to specify programs, we refer to the left-hand side $\expr_1$ as the \emph{implementation}, and to right-hand side $\expr_2$ as the \emph{specification}.
The intuitive meaning is formally reflected by the soundness theorem \wrt contextual refinement.
\begin{thm}[Soundness]
\label[thm]{reloc:thm:soundness}
If the refinement judgment $\logrel{\emptyset}{\expr_1}{\expr_2}{\type}$ is derivable in \reloc, then $\ctxref{\emptyset \mid \emptyset}{\expr_1}{\expr_2}{\type}$.
\end{thm}

In this section we only consider closed programs $\expr_1$ and $\expr_2$; we will see how \reloc (and its soundness theorem) generalize to open terms in \Cref{reloc:sec:fundamental}.

Like ordinary separation logic, \reloc has \emph{heap assertions}.
Since \reloc is relational, these come in two forms: $\loc \mapstoI \val$ and \mbox{$\loc \mapstoS \val$}, which signify ownership of a location $\loc$ with value $\val$ on the implementation and specification side, respectively.

Contrary to earlier work on logical refinements in Iris, \eg~\cite{irisIPM,amin:thesis}, refinement judgments $\logrel{\Delta}{\expr_1}{\expr_2}{\type}$ in \reloc are first-class propositions.
As such, we can combine them in arbitrary ways with the other logical connectives, and state conditional refinements.
For example, the proposition
\begin{equation}
\label{reloc:eq:cond_refinement}
(\loc_1 \mapstoI \val_1 * \loc_2 \mapstoS \val_2 *
\logrel{\Delta}{\expr_1'}{\expr_2'}{\typeB})
\wand{}
\logrel{\Delta}{\expr_1}{\expr_2}{\type},
\end{equation}
states that the $\expr_1$ refines $\expr_2$, under the assumption of another refinement and that certain locations have specified values in the heap.
Having conditional refinements is crucial for modularity, as it allows us to formulate and prove refinements of individual methods of a data structure under the assumptions provided by the internal invariant of the data structure.
The fact that refinement judgments are first class also plays an important role in the presentation of \reloc's proof rules.

\subsection{Derivability and inference rules}
As standard in logic, Iris/\reloc has a derivability relation $P \proves Q$.
We say that $Q$ is derivable if $\TRUE \proves Q$.
In many situations, we use magic wand $\wand$ instead of the derivability relation $\proves$, because we have the standard deduction property:
\[
  P \proves Q \wand R \qquad\textnormal{iff}\qquad P \ast Q \proves R
\]
Most of the inference rules we present can be internalized as \reloc propositions by a magic wand or a derivability relation between the separating conjunction of the antecedents and the consequent.
We thus use the following notations:
\[
\begin{array}{c @{\quad \textnormal{is notation for}\quad} l}
\infer{P_1 \and \dotsb \and P_n}{Q} & (P_1 \ast \dotsb \ast P_n) \wand Q, \\[1.5em]
\inferB{P}{Q} & (P \wand Q) \wedge (Q \wand P).
\end{array}
\]
For instance, the conditional refinement in Formula \eqref{reloc:eq:cond_refinement} is presented as the following inference rule:
\[
  \infer
  {\loc_1 \mapstoI \val_1 \and \loc_2 \mapstoS \val_2 \and
\logrel{\Delta}{\expr_1'}{\expr_2'}{\typeB}}
  {\logrel{\Delta}{\expr_1}{\expr_2}{\type}}
\]
In rules like this, it is useful to think of premises $\loc_1 \mapstoI \val_1$ and $\loc_2 \mapstoS \val_2$ as side conditions, and of the premise $\logrel{\Delta}{\expr_1'}{\expr_2'}{\typeB}$ as the new goal that you get when you apply the rule.
This \emph{backwards-style} reasoning integrates well in the Coq proof assistant; we discuss it more in detail in \Cref{{reloc:sec:formalization}}.

We use the derivability relation $\proves$ explicitly to state rules that cannot be internalized, \eg
$\infer{\vdash P}{\vdash Q}$ states that if $P$ is derivable, then $Q$ is derivable.
This is weaker than $\infer{P}{Q}$, which denotes that $Q$ can be derived from $P$, \ie $P \vdash Q$.

\subsection{Example: Contextual equivalance of a bit module}
\label{reloc:subsec:representation_independence}
\begin{figure}
\raggedright
\textbf{Value interpretation rules}:
\medskip
\begin{mathpar}
\inferhrefB{val-var}{val-var'}
  {\Delta(\tvar)(\val_1, \val_2)}{\logrelV{\tvar}{\Delta}{\val_1}{\val_2}}
\and
\inferhrefB{val-unit}{val-unit'}
  {\val_1 = \val_2 = \unittt}
  {\logrelV{\tunit}{\Delta}{\val_1}{\val_2}}
\and
\inferhrefB{val-bool}{val-bool'}
  {\Exists b \in \mathbb{B}. \val_1 = \val_2 = b}
  {\logrelV{\tbool}{\Delta}{\val_1}{\val_2}}
\and
\inferhrefB{val-int}{val-int'}
  {\Exists n \in \mathbb{Z}. \val_1 = \val_2 = n}
  {\logrelV{\tint}{\Delta}{\val_1}{\val_2}}
\end{mathpar}

\smallskip
\textbf{Type-directed structural rules}:
\medskip
\begin{mathpar}
\inferhref{rel-return}{rel-return'}
  {\logrelV{\type}{\Delta}{\val_1}{\val_2}}
  {\logrel{\Delta}{\val_1}{\val_2}{\type}}
\and\!\!
\inferH{rel-pair}
  {\logrel{\Delta}{\expr_1}{\expr_2}{\type}
   \and \logrel{\Delta}{\expr'_1}{\expr'_2}{\typeB}}
  {\logrel{\Delta}{(\expr_1, \expr'_1)}{(\expr_2,\expr'_2)}{\type \times \typeB}}
\and\!\!
\inferH{rel-pack}
{\All \val_1,\val_2. \persistent{R(\val_1,\val_2)} \and
\logrel{\deltainsert{\tvar}{\rel}{\Delta}}{\expr_1}{\expr_2}{\type}}
 {\logrel{\Delta}{\<pack>\ \expr_1}{\<pack>\ \expr_2}{\texists\tvar.\type}}
\and
\inferH{rel-rec}
  {\always \big(\All \val_1, \val_2. \logrelV{\type}{\Delta}{\val_1}{\val_2} \wand
    \left(\logrel{\Delta}{(\Rec {f_1} {x_1} = e_1)\ \val_1}{(\Rec {f_2} {x_2} = e_2)\ \val_2}{\typeB}\right)\big)}
  {\logrel{\Delta}{(\Rec {f_1} {x_1} = e_1)}{(\Rec {f_2} {x_2} = e_2)}{\type \to \typeB}}
\end{mathpar}

\smallskip
\textbf{Symbolic execution rules}:
\medskip
\begin{mathpar}
\inferH{rel-pure-l}
  {\pureexec{\expr_1}{\expr_1'} \and \later (\logrel{\Delta}{\fillctx\lctx[\expr_1']}{\expr_2}{\type})}
  {\logrel{\Delta}{\fillctx\lctx [\expr_1]}{\expr_2}{\type}}
\and
\inferH{rel-pure-r}
  {\pureexec{\expr_2}{\expr_2'} \and \logrel[\mask]{\Delta}{\expr_1}{\fillctx\lctx [\expr_2']}{\type}}
  {\logrel[\mask]{\Delta}{\expr_1}{\fillctx\lctx [\expr_2]}{\type}}
\and
\inferH{rel-alloc-l'}
  {\All \loc. \loc \mapstoI \val \wand
    \logrel{\Delta}{\fillctx\lctx [\loc]}{\expr_2}{\type}}
  {\logrel{\Delta}{\fillctx\lctx [\Ref(\val)]}{\expr_2}{\type}}
\and
\inferH{rel-alloc-r}
  {\All \loc. \loc \mapstoS \val \wand
    \logrel[\mask]{\Delta}{\expr_1}{\fillctx\lctx [\loc]}{\type}}
  {\logrel[\mask]{\Delta}{\expr_1}{\fillctx\lctx [\Ref(\val)]}{\type}}
\and
\inferH{rel-load-l-inv}
	{\knowInv{\namesp}{\propMV}
	 \and
	 \big(\later\propMV * \CLOSE{\namesp}{\propMV} \big) \wand
	 \Exists \val.\loc \mapstoI \val \mathrel{*}
	 \later \left(\loc \mapstoI \val \wand \logrel[\top \setminus \namesp]{\Delta}{\fillctx\lctx[\val]}{\expr_2}{\type}\right)
	}
	{\logrel{\Delta}{\fillctx\lctx[\deref \loc]}{\expr_2}{\type}}
\and
\inferH{rel-load-r}
  {\loc \mapstoS \val
    \and
    \loc \mapstoS \val \wand
    \logrel[\mask]{\Delta}{\expr_1}{\fillctx\lctx [\val]}{\type}}
  {\logrel[\mask]{\Delta}{\expr_1}{\fillctx\lctx [\deref \loc]}{\type}}
\and
\inferH{rel-store-r}
  {\loc \mapstoS -
    \and
    \loc \mapstoS \val \wand
    \logrel[\mask]{\Delta}{\expr_1}{\fillctx\lctx [\unittt]}{\type}}
  {\logrel[\mask]{\Delta}{\expr_1}{\fillctx\lctx [\loc \gets \val]}{\type}}
\and
\inferH{rel-cas-l-inv}
	{\knowInv{\namesp}{\propMV}
	\and {
		\begin{array}{@{} l @{}}
		\later\propMV * \CLOSE{\namesp}{\propMV} \wand \\
		\qquad \Exists \val. \loc \mapstoI \val \mathrel{*} 
		\later \left( \begin{array}{@{} l @{} l @{}}
		\big(\val = \val_1 * \loc \mapstoI \val_2 &{}\wand
		  \logrel[\top \setminus \namesp]{\Delta}{\fillctx\lctx[\<true>]}{\expr_2}{\type} \big) \mathrel{\land} \\
		\big(\val \neq \val_1 * \loc \mapstoI \val &{}\wand
		  \logrel[\top \setminus \namesp]{\Delta}{\fillctx\lctx[\<false>]}{\expr_2}{\type} \big)
		\end{array} \right)
		\end{array}
	}}
	{\logrel{\Delta}{\fillctx\lctx[\CAS{\loc}{\val_1}{\val_2}]}{\expr_2}{\type}}
\end{mathpar}

\smallskip
\textbf{Invariants rules}:
\medskip
\begin{mathpar}
\inferhref{rel-inv-alloc}{rel-inv-alloc'}
  {\later\propMV \and
  \knowInv{\namesp}{\propMV} \wand \logrel{\Delta}{\expr_1}{\expr_2}{\type}}
  {\logrel{\Delta}{\expr_1}{\expr_2}{\type}}
\and
\inferH{rel-inv-restore}
	{\CLOSE{\namesp}{\propMV} \and \later\propMV \and
	\logrel[\mask]{\Delta}{\expr_1}{\expr_2}{\type}}
	{\logrel[\mask\setminus \namesp]{\Delta}{\expr_1}{\expr_2}{\type}}
\end{mathpar}
\caption{Selected rules of \reloc.}
\label{reloc:fig:rules_tour}
\end{figure}

We demonstrate the basic usage of \reloc by using its \emph{type-directed structural} and \emph{symbolic execution} rules to prove contextual equivalence of two implementations of a simple program module (representation independence).
The module we consider represents a single bit data structure---it contains an initial value for the bit, an operation for flipping the bit, and an operation for converting the values of the abstract type to Booleans.
We use an existential type (\ie abstract type) to hide the representation type and thus the type of the module:
\[
  \tbit \eqdef \texists{\tvar. \tvar \times (\tvar \to \tvar) \times (\tvar \to \tbool)}.
\]

Perhaps the simplest implementation of the bit interface is the one that uses Booleans for the internal state:
\[
  \bitbool : \tbit \eqdef \<pack>(\<true>, (\Lam b. \neg b), (\Lam b. b)).
\]
The second implementation models a bit by a number from the set $\{0, 1\}$:
\begin{align*}
  \flipnat : \tint \to \tint \eqdef{}& \Lam n. \If (n = 0) then 1 \Else 0 \\
  \bitnat : \tbit \eqdef{}& \<pack>(1, \flipnat, (\Lam n. n = 1)).
\end{align*}

Before we explain how the contextual equivalence of these two implementation is formally proved in \reloc, let us informally discuss why these implementations are equivalent.
Note that the underlying types ($\tint$ and $\tbool$) are not isomorphic.
This, however, is not going to be a problem, because the underlying types are hidden/existentially abstracted in the module signature.
As a consequence of that, a (well-typed) client has to be polymorphic in the type $\tvar$, and can thus only create and modify values of $\tvar$ through the functions provided by the module.
A client that uses the $\bitnat$ module can only construct integers $0$ and $1$ (using the initial value and applying the flip function a number of times).
Thus, requiring an isomorphism between the underlying types is too strict---for example, we do not care what Boolean value an integer $7$ might correspond to, because the number $7$ can never be constructed using the functions provided by $\bitnat$.

This intuitive reasoning signals the key idea behind the \emph{representation independence} principle \cite{mitchell:1986}, which states that in order to prove that two modules are equivalent, it suffices to pick a \emph{relation} between the underlying types and demonstrate that all the methods preserve this relation.
For this example, a sensible candidate for such a relation is
$
  \set{ (\<true>, 1), (\<false>, 0) }
$.
Note that our relation does not include any integers other than $0$ or~$1$, because as we previously explained, a well-typed client of $\bitnat$ cannot construct other integers.
With the relation at hand, the informal proof is as follows.
The initial values offered by the modules are related.
The flip function preserves this relation.
The function that converts ``bits'' to Booleans sends related values to the same Boolean.

We will now demonstrate how to carry out this argument formally in \reloc.
Specifically, we prove the following refinement using the rules in \Cref{reloc:fig:rules_tour}:
\[
  \logrel{}{\bitbool}{\bitnat}{\tbit}
\]
The other direction can be proved in a similar way, which using soundness (\Cref{reloc:thm:soundness}), gives us the contextual equivalence $\ctxequiv{}{\bitbool}{\bitnat}{\tbit}$.

From a high-level point of view, the proof of this example involves applying \reloc's type-directed structural rules following the structure of $\tbit$.
At the leaves of the proof, we continue with \reloc's symbolic execution rules to perform computation steps.

Since $\tbit$ is an existential type, and both $\bitbool$ and $\bitnat$ are $\<pack>$'s, we start off by applying the type-directed structural rule \ruleref{rel-pack}.
For that we need to pick a relation $R$, which will be the interpretation for the type variable $\tvar$, and should link together the underlying representations of bits in $\bitbool$ and $\bitnat$.
We define the relation $R$ as follows:
\[
  R(b, n) \eqdef (b = \<true> \wedge n = 1) \vee (b = \<false> \wedge n = 0).
\]
Starting with the initial goal $\logrel{}{\bitbool}{\bitnat}{\tbit}$, we apply \ruleref{rel-pack}.
As a side-condition, we have to prove that $R$ is \emph{persistent} for any $\val_1, \val_2$, written as $\persistent{R(\val_1,\val_2)}$, intuitively meaning that the proposition $R(\val_1, \val_2)$ does not assert ownership of any resources.
We discuss persistent propositions in more detail in \Cref{reloc:sec:tour:invariants,reloc:sec:persistence}, and for now we just note that the relation $R$ is indeed persistent.
After application of the \ruleref{rel-pack} rule the goal becomes:
\[
\logrel{\deltainsert{\tvar}{R}{}}
  {(\<true>, (\Lam b. \neg b), (\Lam b.b))}
  {(1, \flipnat, (\Lam n. n = 1))}
  {\tvar \times (\tvar \to \tvar) \times (\tvar \to \tbool)}.
\]
By repeatedly applying the type-directed structural rule \ruleref{rel-pair} we get three new goals:
\begin{enumerate}
\item $\logrel{\deltainsert{\tvar}{R}{}}{\<true>}{1}{\tvar}$;
\item $\logrel{\deltainsert{\tvar}{R}{}}{(\Lam b. \neg b)}{\flipnat}{\tvar \to \tvar}$;
\item $\logrel{\deltainsert{\tvar}{R}{}}{(\Lam b. b)}{(\Lam n. n = 1)}{\tvar \to \tbool}$.
\end{enumerate}
For the first goal, we can use the rules \ruleref{rel-return'} and \ruleref{val-var'}, leaving us with the obligation $R(\<true>, 1)$, which holds by the definition of $R$.

For the second and the third goal we need to prove refinements of two closures,
for which we use the type-directed structural rule \ruleref{rel-rec}.
Let us look at the third goal in detail.
After the application of \ruleref{rel-rec} we have to show:
\[
\always \left( \All \val_1, \val_2. \logrelV{\tvar}{{\deltainsert{\tvar}{R}{}}}{\val_1}{\val_2} \wand
\logrel{{\deltainsert{\tvar}{R}{}}}{(\Lam b. b)\ \val_1}{(\Lam n. n = 1)\ \val_2}{\tbool}
\right).
\]
The goal is wrapped in Iris's \emph{persistence modality} $\always$, which turns any proposition into a persistent one.
Once again, we postpone the details about the persistence modality until \Cref{reloc:sec:tour:invariants,reloc:sec:persistence}, and only remark that here we are allowed to prove the goal without the $\always$ modality.
Using this information, and the rule \ruleref{val-var'} we reduce our goal to show:
\[
R(\val_1, \val_2) \wand \logrel{{\deltainsert{\tvar}{R}{}}}{(\Lam b. b)\ \val_1}{(\Lam n. n = 1)\ \val_2}{\tbool},
\]
for arbitrary $\val_1, \val_2$.
We then unfold the definition of $R$ and observe that we need to distinguish two cases:
\begin{enumerate*}
\item $\val_1 = \<true>$ and $\val_2 = 1$;
\item $\val_1 = \<false>$ and $\val_2 = 0$.
\end{enumerate*}
Suppose we are in the first case (the second case is similar).
We have to show:
\[
\logrel{{\deltainsert{\tvar}{R}{}}}{(\Lam b. b)\ \<true>}{(\Lam n. n = 1)\ 1}{\tbool}.
\]
At this point we apply \reloc's \emph{symbolic execution} rules: we symbolically reduce both the left-hand and the right-hand side of the refinement.
For this we use the rules \ruleref{rel-pure-l} and \ruleref{rel-pure-r} (the later modalities ($\later$) in these rules can be ignored for now, they will be explained in \Cref{reloc:sec:counter:loeb_induction}).
These rules perform \emph{pure reductions}, \ie reductions that do not depend on the heaps.
In our case we have a $\beta$-reduction on the left-hand side, and a $\beta$-reduction and an evaluation of the binary operation (equality testing) on the right-hand side:
\[
  \pureexec{(\Lam b. b)\ \<true>}{\<true>} \qquad
  \pureexec{(\Lam n. n = 1)\ 1}{\pureexec{(1 = 1)}{\<true>}}.
\]
After the repeated application of the said rules we arrive at a goal
\[
  \logrel{{\deltainsert{\tvar}{R}{}}}{\<true>}{\<true>}{\tbool},
\]
which we discharge by \ruleref{rel-return'} and \ruleref{val-bool'}.
This completes the proof of the refinement.

\subsection{Example: Contextual refinement of a concurrent counter}
\label{reloc:sec:counterproof}
The previous example showcased how \reloc can be used to show contextual refinement and equivalence of pure program modules.
%Another application of \reloc is proving linearizability of concurrent modules~\cite{herlihy:wing:1990}.
%Linearizability is a global property of concurrent objects, but it can be reduced to showing contextual refinement~\cite{filipovic:ohearn:rinetzky:yang:2010}.
In this subsection we prove contextual refinement of the fine-grained concurrent counter in \Cref{reloc:fig:impl} from \Cref{reloc:sec:intro} by showing that it refines the coarse-grained counter.
Specifically, we prove the following refinement:
\[
\ctxref{}{\FGcounter}{\CGcounter}
	{(\tunit \to \tint) \times (\tunit \to \tint)}.
\]
Using soundness (\Cref{reloc:thm:soundness}), this contextual refinement can be reduced to proving the refinement judgment $\logrel{}{\FGcounter}{\CGcounter}
	{(\tunit \to \tint) \times (\tunit \to \tint)}$ in \reloc.

The previous example demonstrated the basic usage of symbolic execution rules of \reloc.
Those symbolic execution rules were confined to the pure fragment of the programming language.
In this example we show how to use \reloc's symbolic execution rules for stateful computations and concurrency primitives.
In addition to the type-directed structural rules and symbolic execution rules, the proof will require the usage of \emph{invariants} for linking together the values of the two counters.
We will use selected \reloc rules from \Cref{reloc:fig:rules_tour}.
To symbolically execute the operations on locks that appear in $\CGcounter$, we will also make use of the \emph{relational specification} for locks in \Cref{reloc:fig:locks}.
The lock specification %(which can be proved for \eg a spinlock)
is stated in terms of an abstract predicate $\isLock{\lkvar}{\<false>}$ (resp., $\isLock{\lkvar}{\<true>}$) stating that $\lkvar$ is a lock which is unlocked (resp., locked).
The relational specification for locks can then be seen as consisting of symbolic execution rules that manipulate that abstract predicate.\footnote{Because this specification is for the ``angelic'' right-hand side, it does not express mutual exclusion as it is common for separation logic specifications. We explain this by contrasting the specification with the one for the left-hand side in \Cref{{reloc:sec:simpl_specs_lhs}}.}
We will see in \Cref{reloc:subsec:specs_rhs} that these specifications can be proven for a simple spin lock.

\begin{figure*}
\begin{mathpar}
\inferH{newlock-r}
{\All \lkvar. \isLock{\lkvar}{\<false>} \wand \logrel[\mask]{\Delta}{\expr_1}{\fillctx\lctx[\lkvar]}{\type}}
{\logrel[\mask]{\Delta}{\expr_1}{\fillctx\lctx[\newlock\ \unittt]}{\type}}
\and
\inferH{acquire-r}
{\isLock{\lkvar}{\<false>} \and
\isLock{\lkvar}{\<true>} \wand
    \logrel[\mask]{\Delta}{\expr_1}{\fillctx\lctx[\unittt]}{\type}}
{\logrel[\mask]{\Delta}{\expr_1}{\fillctx\lctx[\acquire\ \lkvar]}{\type}}
\and
\inferH{release-r}
{\isLock{\lkvar}{b} \and
\isLock{\lkvar}{\<false>} \wand
   \logrel[\mask]{\Delta}{\expr_1}{\fillctx\lctx[\unittt]}{\type}}
{\logrel[\mask]{\Delta}{\expr_1}{\fillctx\lctx[\release\ \lkvar]}{\type}}
\end{mathpar}
\caption{Right-hand side relational specification for locks.}
\label{reloc:fig:locks}
\end{figure*}

\subsubsection{Symbolic execution}
\label{reloc:sec:counter:symbolic_execution}
Recall that performing symbolic execution means reducing the left-hand or right-hand side of the refinement according to the computational rules.
We have already seen the usage of \ruleref{rel-pure-l}, which allows us to perform pure computations.
For this example we also use stateful symbolic execution rules in \Cref{reloc:fig:rules_tour}.
To start with the refinement proof, we apply the stateful symbolic execution rule \ruleref{rel-alloc-l'} to the left-hand side to obtain:
\[
c_i \mapstoI 0 \wand
	\logrel{}{((\Lam \unittt. \counterread\ c_i),(\Lam \unittt. \FGincrement\ c_i))}{\CGcounter}{(\tunit \to \tint) \times (\tunit \to \tint)}.
\]
Note that after the application of the rule we gain access to the resource $c_i \mapstoI 0$ representing the value of the counter on the left-hand side.
Subsequently, using the symbolic execution rules \ruleref{rel-pure-r}, \ruleref{rel-alloc-r} and \ruleref{newlock-r} on the right-hand side the goal becomes:
\begin{multline*}
c_i \mapstoI 0 * c_s \mapstoS 0 * \isLock{\lkvar}{\<false>}
\wand \\
\logrel{}{((\Lam \unittt. \counterread\;c_i),(\Lam \unittt.\FGincrement\ c_i))}
  {(\Lam \unittt. \counterread\ c_s),(\Lam \unittt.\CGincrement\ c_s\ \lkvar))}
  {(\tunit \to \tint) \times (\tunit \to \tint)}.
\end{multline*}
In addition to gaining the resource $c_s \mapstoS 0$, representing the value of the right-hand side counter, we get access to the abstract predicate $\isLock{\lkvar}{\<false>}$, which keeps track of the state of the lock $\lkvar$ on the right-hand side.

\reloc's symbolic execution rules are inspired by the ``backwards''-style Hoare rules of~\cite{ishtiaq:ohearn:2001} and the weakest-precondition rules in Iris~\cite{iris3,irisJFP}.

\subsubsection{Invariants and persistent propositions}
\label{reloc:sec:tour:invariants}

At this point we wish to prove a refinement of two closures.
By the rule \ruleref{rel-pair} it would suffice to prove that both closures refine each other.
However, if we were to apply \ruleref{rel-pair}, we would be forced to split our resources in two: the resources needed for the refinement proof of the read function, and the resources needed for the refinement proof of the increment function.
But both of those operations require access to the counter locations $c_i \mapstoI -$ and $c_s \mapstoS -$.
To circumvent this issue we put said resources in a global \emph{invariant} $\knowInv{}{\propMV}$, which allows $\propMV$ to be shared between different parts of the program (and between different threads).
In our running example, we establish the invariant $\knowInv{\namesp}{\counterinv}$ (using \ruleref{rel-inv-alloc'}), where:
\[
\counterinv \eqdef \Exists n \in \nat. c_i \mapstoI n * c_s \mapstoS n * \isLock{\lkvar}{\<false>}.
\]
The invariant $\knowInv{\namesp}{\counterinv}$ not only allows us to share access to $c_i$ and $c_s$, but also ensures that the values of the respective counters match up.
For our invariant we pick any fresh invariant name $\namesp \in \InvName$ (more on the invariant names below).

Invariants $\knowInv {} \propMV$ are \emph{persistent}: once established, they will remain valid for the rest of the verification.
This differentiates them from \emph{ephemeral} propositions like $\loc \mapstoI \val$ and $\loc \mapstoS \val$, which could be invalidated in the future by actions of the program or proof.

The notion of being persistent is expressed in \reloc (and Iris) by means of the \emph{persistence} modality $\always$.
The purpose of $\always \propMV$ is to say that $\propMV$ holds without asserting any ephemeral propositions.
The most important rules for the $\always$ modality are $\always \propMV =  \always\propMV * \always\propMV$ and $\always\propMV \wand \propMV$, which allow to freely duplicate $\always \propMV$ and finally get $\propMV$ out.
We say that $\propMV$ is \emph{persistent}, written as $\persistent{\propMV}$, if $\propMV \proves \always \propMV$; otherwise, we say that $\propMV$ is \emph{ephemeral}.
To prove $\always \propMV$, one can only use persistent resources like $\knowInv{}{\propMV}$, and not ephemeral resources like $\loc \mapstoI \val$.
We refer to stripping off the persistence modality in the context of persistent hypotheses as \emph{introducing the $\always$ modality}.
We make that precise and give rules for the $\always$ modality in \Cref{reloc:sec:calculus}.

Once the invariant $\knowInv{}{\counterinv}$ for our running example has been established, we can duplicate it, and apply \ruleref{rel-pair} to obtain two goals:
\begin{align*}
  \knowInv{}{\counterinv} \wand{}&
  \logrel{}{(\Lam \unittt. \counterread\ c_i)}{(\Lam \unittt. \counterread\ c_s)}{\tunit \to \tint}
   \\
  \knowInv{}{\counterinv} \wand{}&
  \logrel{}{(\Lam \unittt.\FGincrement\ c_i)}{(\Lam \unittt.\CGincrement\ c_s\ \lkvar)}{\tunit \to \tint}.
\end{align*}
We first describe how to prove the refinement of $\counterread$.
As $\Lam x.e$ is syntactic sugar for $\Rec \_ x = e$, we can apply \ruleref{rel-rec} at the function type $\tunit \to \tint$ and obtain the new goal:
\[
  \knowInv{}{\counterinv} \wand
  \always \big(
    \All \val\, \val'. \logrelV{\tunit}{\Delta}{\val}{\val'} \wand
    \logrel{}{(\Lam \unittt. \deref c_i)\ \val}{(\Lam \unittt.\deref c_s)\ \val'}{\tint} \big).
\]
By \ruleref{val-unit'}, we obtain that $\logrelV{\tunit}{\Delta}{\val}{\val'}$ implies $\val = \val' = \unittt$.
Moreover, since $\knowInv{}{\counterinv}$ is our only hypothesis, and it is persistent, we can strip off the $\always$ modality% (we defer to \Cref{reloc:sec:calculus} for a detailed discussion about the $\always$ modality)
, arriving at the following goal:
\[
  \knowInv{}{\counterinv} \wand{} {\logrel{}{(\Lam \unittt. \deref c_i)\ \unittt}{(\Lam \unittt.\deref c_s)\ \unittt}{\tint}}.
\]

\paragraph{Accessing invariants.}
The fact that invariants are persistent (and thus can be duplicated, \ie $\knowInv {} \propMV \mathrel{=} \knowInv {} \propMV \mathrel{*} \knowInv {} \propMV$) comes with a cost---once a proposition $\propMV$ has been turned into an invariant $\knowInv {} \propMV$, one is only allowed to access $\propMV$ during a single \emph{atomic} execution step on the left-hand side.
This restriction is crucial as the scheduling of threads on the left-hand side is demonic.
When proving a refinement, we have to consider all possible interleavings of threads.
If we were to be able to access an invariant for the duration of multiple steps, another thread could be scheduled in between, and observe that the invariant was temporarily broken.

Scheduling on the right-hand side, however, is angelic.
That is, when proving a refinement, we have the ability to select the choice of scheduling.
As a consequence, \reloc allows us to execute multiple steps on the right-hand side while accessing an invariant.

Let us take a look at the way accessing invariants in \reloc works.
We do so by continuing the proof of our running example (after introducing $\always$ and performing pure symbolic execution steps):
\[
  \knowInv{}{\counterinv} \wand{}
  \logrel{}{\deref c_i}{\deref c_s}{\tint}.
\]
At this point we would like to access the locations $c_i$ and $c_s$ stored in the invariant $\knowInv{}{\counterinv}$.
For this we use the rule \ruleref{rel-load-l-inv} in \Cref{reloc:fig:rules_tour}.

This rule is quite a mouthful, so let us first take a look at its shape before going into detail about the mask annotations and later modalities $\later$.
The essence of \ruleref{rel-load-l-inv} is that it provides temporary access to the resources $\propMV$ guarded by the invariant.
In addition, it provides the \emph{invariant closing resource} $\CLOSE{\namesp}{\propMV}$, which can restore the invariant (using the rule \ruleref{rel-inv-restore}).
The resources $\propMV$ can be used to prove $\loc \mapstoI \val$, which is needed to justify the symbolic execution step on the left.
Afterwards, we are left with the goal $\logrel[\top \setminus \namesp]{\Delta}{\fillctx\lctx[\val]}{\expr_2}{\type}$.
We typically do not immediately restore the invariant (using \ruleref{rel-inv-restore}), but first use the resources $\propMV$ to perform matching symbolic execution steps on the right.

In our example, by applying \ruleref{rel-load-l-inv}, we obtain $c_i \mapstoI n$ and $c_s \mapstoS n$ and $\isLock{\lkvar}{\<false>}$, for some $n \in \nat$, reducing our goal to
$
\logrel[\top \setminus \counterN]{}{n}{\deref c_s}{\tint}
$.
We then use \ruleref{rel-load-r} to reduce our goal to
$
\logrel[\top \setminus \counterN]{}{n}{n}{\tint}.
$
Because these steps did not change the heap, \ruleref{rel-inv-restore}'s premises for closing the invariant are trivially met.
The refinement proof is then concluded by applying the structural rules \ruleref{rel-return'} and \ruleref{val-int'}.

Let us take a look at the rules \ruleref{rel-load-l-inv} and \ruleref{rel-inv-restore} in more detail.
A crucial aspect of these rules is that they ensure that access to the invariant $\knowInv {\namesp} \propMV$ is \emph{temporary}, \ie that $\propMV$ is only used during a single symbolic execution step on the left-hand side (but possibly several steps on the right), and that the same invariant cannot be opened twice.
This is achieved by tagging each invariant $\knowInv \namesp \propMV$ with a name $\namesp \in \InvName$, and by keeping track of which invariants have been accessed.
The latter is done in a way similar to Iris---like Iris's Hoare triples $\hoare \propMV \expr \propMVB[\mask]$, our refinement judgments $\logrel[\mask]{\Delta}{\expr_1}{\expr_2}{\type}$ are annotated with a \emph{mask} $\mask \subseteq \InvName$ of accessible invariants.
By default all invariants are accessible, so we write $\logrel{\Delta}{\expr_1}{\expr_2}{\type}$ for $\logrel[\top]{\Delta}{\expr_1}{\expr_2}{\type}$, where $\top$ is the set of all invariant names.

An invariant namespace is a (non-empty) list of strings or values: $\InvName = \List(\textdom{String} + \Val)$.
When opening an invariant and removing it from a mask, we coerce an invariant namespace $\namesp$ into a mask by taking its upwards extension $\namecl{\namesp} = \{ \namesp.x_1.\dots.x_n \mid n \in \mathbb{N},\, x_i \in \textdom{String} + \Val \}$.
Abusing the notation, we write $\mask \setminus \namesp$ for $\mask \setminus \namecl{\namesp}$.

When accessing an invariant, \eg using \ruleref{rel-load-l-inv} or \ruleref{rel-cas-l-inv}, its namespace is removed from the mask annotation of the judgment.
The removal of the namespace from the mask guarantees that invariants are only used for a single execution step on the left-hand side.
After all, all rules for symbolic execution on the left-hand side require a $\top$ mask, whereas those for the right-hand side allow for an arbitrary mask.
The only way of performing a subsequent step on the left-hand side is thus by first restoring the mask to $\top$, which can only be done by restoring the invariants that have been accessed (using the rule \ruleref{rel-inv-restore}).

One may wonder why refinement judgments are annotated with a mask instead of a Boolean that indicates if an invariant has been opened.
As we will show in \Cref{reloc:sec:rules}, \reloc allows one to access multiple invariants simultaneously.
To avoid \emph{reentrancy}---which means accessing the same invariant twice in a nested fashion---we need to know exactly which invariants are opened.

An additional aspect to note is that invariants $\knowInv \namesp \propMV$ in \reloc (and Iris) are \emph{impredicative}~\cite{svendsen:birkedal:2014,irisJFP}.
This means that $\propMV$ is allowed to contain other invariant assertions $\knowInv {\namesp'} \propMVB$ or even refinement judgments $\logrel{}{\expr}{\exprB}{\type}$.
As a consequence, to ensure soundness of the logic, all rules for invariants only provide access to $\later\propMV$, \ie $\propMV$ ``guarded'' by the \emph{later} modality $\later$.
When invariants are not used impredicatively (\ie invariants over so called \emph{timeless} propositions, which include connectives of first-order logic and heap assertions), these modalities can be soundly omitted.

\subsubsection{Later modality and L\"ob induction}
\label{reloc:sec:counter:loeb_induction}
The later modality $\later$ is not only used for resolving the impredicativity issues, but also for handling general recursion.
As is custom in logics based on step-indexing~\cite{appel:mcallester:2001}, such as Iris, the later modality $\later$ and L\"ob induction are used to reason about recursive functions.
Specifically, Iris provides the following rules for~$\later$:
\begin{mathpar}
\inferhref{$\later$-intro}{later-intro}
  {P}
  {\later P}
\and
\inferhref{$\later$-mono}{later-mono}
  {P \proves Q}
  {\later P \proves \later Q}
\and
\inferhref{L\"ob}{later-loeb}
  {\later P \proves P}
  {\proves P}
\end{mathpar}

In our example, this means that by L\"ob induction (rule \ruleref{later-loeb}), we may prove
$\logrel{} {\FGincrement\ c_i} {\CGincrement\ c_s\ \lkvar}{\tint}$,
under the assumption of the induction hypothesis
$\later (\logrel{} {\FGincrement\ c_i} {\CGincrement\ c_s\ \lkvar}{\tint})$.
The induction hypothesis is `guarded' by a $\later$, and can only be used after we have performed a step of symbolic execution on the left-hand side.
That is why the symbolic execution rules for the left-hand side contain the later modality in the premises.
Let us see how it works in the example.
We use \ruleref{rel-pure-l} to arrive at:
\begin{align*}
&\later (\logrel{} {\FGincrement\ c_i} {\CGincrement\ c_s\ \lkvar}{\tint}) \wand \\
&\later (\logrel{}%
{\Let c = \deref c_i in \If \CAS{c_i}{c}{1 + c} then c \Else \FGincrement\,c_i}%
{\CGincrement\ c_s\ \lkvar}{\tint}).
\end{align*}
By monotonicity (rule \ruleref{later-mono}), we can now remove $\later$ both from the induction hypothesis and from the goal.
Subsequently, we symbolically execute the load operation using the invariant, just like in the previous section, reaching the goal
\[
\logrel{}%
{\If \CAS{c_i}{n}{1 + n} then n \Else \FGincrement\,c_i}%
{\CGincrement\ c_s\ \lkvar}{\tint}
\]
for some $n \in \nat$.
To symbolically execute the compare-and-set ($\<CAS>$), we use \ruleref{rel-cas-l-inv}.
By this rule, we have to consider two outcomes, depending on whether the original value of the counter has changed between the load and compare-and-set operations or not.

\begin{enumerate}
\item Suppose that the value of the counter $c_i$ has changed.
In that case the compare-and-set operation fails and we are left with
\begin{align*}
& c_i \mapstoI m * c_s \mapstoS m * \isLock{\lkvar}{\<false>} \wand \\
& \logrel[\top \setminus \counterN]{}%
{\If \<false> then n \Else \FGincrement\,c_i}%
{\CGincrement\ c_s\ \lkvar}{\tint}
\end{align*}
for some $m \neq n$.
Because the symbolic heap has not been changed, we can easily restore the invariant and execute the $\If \<false> then \dots \Else \dots$ to obtain
$\logrel{}{\FGincrement\,c_i}{\CGincrement\ c_s\ \lkvar}{\tint}$,
which is exactly our induction hypothesis.

\item If the value has not changed, then the compare-and-set succeeds and we are left with the new goal:
\begin{align*}
&c_i \mapstoI (1 + n) * c_s \mapstoS n * \isLock{\lkvar}{\<false>} \wand \\
&\logrel[\top \setminus \counterN]{}%
{\If \<true> then n \Else \FGincrement\,c_i}%
{\CGincrement\ c_s\ \lkvar}{\tint}.
\end{align*}
At this point we use the symbolic execution rules \ruleref{rel-store-r}, \ruleref{rel-load-r} and the lock specifications from \Cref{reloc:fig:locks} to symbolically execute the right-hand side of the refinement and update the resources to match:
\begin{align*}
&c_i \mapstoI (1 + n) * c_s \mapstoS (1 + n) * \isLock{\lkvar}{\<false>} \wand \\
&\logrel[\top \setminus \counterN]{}%
{\If \<true> then n \Else \FGincrement\,c_i}%
{n}{\tint}.
\end{align*}
We can then restore the invariant and symbolically execute the left-hand side to finish the proof.

\end{enumerate}
Note that the point in the proof when we symbolically execute $\CGincrement\ c_s\ \lkvar$ on the right-hand side corresponds to the linearization point of $\FGincrement$.

This concludes the proof of the counter refinement.
For the purposes of the proof, we have used some derived rules and principles in \reloc.
In the next section we will present an overview of primitive rules---the very core of \reloc---and show how they can be used to recover the kind of intuitive reasoning we employed in this section.

%%% Local Variables:
%%% mode: latex
%%% TeX-master: "reloc"
%%% End:

\section{A closer look at \reloc}
\label{reloc:sec:rules}
We now explain some of the more technical details of \reloc, and show how the principles that we have used in \Cref{reloc:sec:tour} can be obtained from \reloc's primitive proof rules.
First, we describe how to work with invariants using Iris's update modality $\pvs$
(\Cref{reloc:sec:invariants_update}).
Then we explain the role and rules of persistent propositions (\Cref{reloc:sec:persistence}), and go through a selection of \reloc's primitive proof rules and explain how the symbolic execution and structural rules can be derived from them (\Cref{reloc:sec:calculus}).
Finally, we demonstrate how \reloc's rules can be used to prove the \emph{fundamental property}: if we can derive a typing judgment $\typed{}{}{\expr}{\type}$, then $\expr$ refines itself, \ie $\logrel{}{\expr}{\expr}{\type}$.
To prove the fundamental property, we need to generalize the relational judgment to open terms, and prove the structural rules for open terms as well (\Cref{reloc:sec:fundamental}).

A selection of \reloc's primitive proof rules are shown in \Cref{reloc:fig:rules}.
\begin{figure*}
\raggedright
\textbf{Value interpretation rules}:
\medskip
\begin{mathpar}
\inferHB{val-var}
  {\Delta(\tvar)(\val_1, \val_2)}
  {\logrelV{\tvar}{\Delta}{\val_1}{\val_2}}
\and
\inferHB{val-unit}
  {\val_1 = \val_2 = \unittt}
  {\logrelV{\tunit}{\Delta}{\val_1}{\val_2}}
\and
\inferHB{val-bool}
  {\Exists b \in \mathbb{B}. \val_1 = \val_2 = b}
  {\logrelV{\tbool}{\Delta}{\val_1}{\val_2}}
\and
\inferHB{val-int}
  {\Exists n \in \mathbb{Z}. \val_1 = \val_2 = n}
  {\logrelV{\tint}{\Delta}{\val_1}{\val_2}}
\and
\inferHB{val-prod}
  {\Exists \val_1,\val_2,\valB_1,\valB_2. \val = (\val_1,\val_2) \ast
    \valB = (\valB_1,\valB_2) \ast \logrelV{\type}{\Delta}{\val_1}{\valB_1}
  \ast \logrelV{\typeB}{\Delta}{\val_2}{\valB_2}}
  {\logrelV{\type \times \typeB}{\Delta}{\val}{\valB}}
\and
\inferHB{val-arr}
  {\always (\All \valB_1\, \valB_2. \logrelV{\type}{\Delta}{\valB_1}{\valB_2} \wand
      \logrel{\Delta}{\val_1\ \valB_1}{\val_2\ \valB_2}{\typeB})}
  {\logrelV{\type \to \typeB}{\Delta}{\val_1}{\val_2}}
\end{mathpar}  

\bigskip
\textbf{Monadic rules}:
\begin{mathpar}
\inferH{rel-return}
{\logrelV{\type}{\Delta}{\val_1}{\val_2}}
{\logrel{\Delta}{\val_1}{\val_2}{\type}}
\and
\inferH{rel-bind}
{\logrel{\Delta}{\expr_1}{\expr_2}{\type}
  \and
  {
  \begin{array}[b]{@{} r @{\spac} l @{}}
  \All \val_1\, \val_2.& \logrelV{\type}{\Delta}{\val_1}{\val_2} \wand \\
  &\logrel{\Delta}{\fillctx{\lctx_1}[\val_1]}{\fillctx{\lctx_2}[\val_2]}{\typeB}
  \end{array}}}
{\logrel{\Delta}{\fillctx{\lctx_1}[\expr_1]}{\fillctx{\lctx_2}[\expr_2]}{\typeB}}
\end{mathpar}

\bigskip
\textbf{Type-directed structural rules}:
\begin{mathpar}
\inferH{rel-fork}
  {\logrel{\Delta}{{\expr_1}}{{\expr_2}}{\tunit}}
  {\logrel{\Delta}{\Fork{\expr_1}}{\Fork{\expr_2}}{\tunit}}
\end{mathpar}

\bigskip
\textbf{Symbolic execution rules}:
\medskip
\begin{mathpar}
\inferH{rel-load-l}
  {\pvs[\top][\mask] {\left(
    \Exists \val. \loc \mapstoI \val *
    \later \big(\loc \mapstoI \val \wand
      \logrel[\mask]{\Delta}{\fillctx\lctx [\val]}{\expr_2}{\type}\big)
  \right)}}
  {\logrel{\Delta}{\fillctx\lctx [\deref \loc]}{\expr_2}{\type}}
\and
\inferH{rel-store-l}
  {\pvs[\top][\mask] {\left(
    \loc \mapstoI - *
    \later \big(\loc \mapstoI \val \wand
      \logrel[\mask]{\Delta}{\fillctx\lctx [\unittt]}{\expr_2}{\type}\big)
  \right) }}
  {\logrel{\Delta}{\fillctx\lctx [\loc \gets \val]}{\expr_2}{\type}}
\and
\inferH{rel-cas-l}
  {\pvs[\top][\mask] {\left(
  \begin{array}{@{} r @{\spac} l @{}}
  \Exists \val'. \loc \mapstoI \val' \mathrel{*} &
  \later \big(\val' \neq \val_1 \wand \later (\loc \mapstoI \val' \wand
    \logrel[\mask]{\Delta}{\fillctx\lctx [\<false>]}{\expr_2}{\type}) \big) \mathrel{\land} \\
  & \later \big(\val' = \val_1 \wand \later (\loc \mapstoI \val_2 \wand
    \logrel[\mask]{\Delta}{\fillctx\lctx [\<true>]}{\expr_2}{\type})\big)
  \end{array} \right) }}
  {\logrel{\Delta}{\fillctx\lctx [\CAS{\loc}{\val_1}{\val_2}]}{\expr_2}{\type}}
\and
\end{mathpar}

\bigskip
\textbf{Invariants rules} (\ruleref{inv-alloc} and \ruleref{inv-access} are inherited from Iris):
\medskip
\begin{mathpar}
\inferH{rel-upd}
  {\pvs[\mask_1][\mask_2] \logrel[\mask_2]{\Delta}{\expr_1}{\expr_2}{\type}}
  {\logrel[\mask_1]{\Delta}{\expr_2}{\expr_2}{\type}}
\and
\inferH{inv-alloc}
  {\later \propMV}
  {\pvs[\mask] \knowInv{\namesp}{\propMV}}
\and
\inferH{inv-access}
  {\namesp \subseteq \mask \and \knowInv{\namesp}{\propMV}}
  {\pvs[\mask][\mask \setminus \namesp] \later \propMV * (\later \propMV \vsW[\mask \setminus \namesp][\mask] \TRUE)}
\end{mathpar}
\caption{Selected primitive rules of \reloc.}
\label{reloc:fig:rules}
\end{figure*}

\subsection{Invariants and the update modality}
\label{reloc:sec:invariants_update}
The rules for invariants in \Cref{reloc:fig:rules_tour} in \Cref{reloc:sec:tour:invariants} are fairly restrictive, \eg they allow us to open at most one invariant at the same time.
Moreover, several of those rules, \eg \ruleref{rel-load-l-inv} and \ruleref{rel-cas-l-inv}, mix together symbolic execution and invariant manipulation.
We now present \reloc's more primitive proof rules, which integrate Iris's flexible mechanism for invariants and ghost state, and which can be used to derive rules such as like \ruleref{rel-load-l-inv} and \ruleref{rel-cas-l-inv}.

Invariants and ghost state in Iris are controlled via the \emph{update modality} $\pvs[\mask_1][\mask_2] \propMV$.
The intuition behind $\pvs[\mask_1][\mask_2] \propMV$ is to express that under the assumption that the invariants in $\mask_1$ are accessible initially, one can obtain $\propMV$, and end up in the situation where the invariants in $\mask_2$ are accessible.
Thus, for showing $P$ we can open the invariants from $\mask_1$ and have to restore the invariants from $\mask_2$ (the invariants from $\mask_1 \setminus \mask_2$ may remain open).
Furthermore, this modality allows one to perform changes to Iris's ghost state via \emph{frame preserving updates}; for a description of those we refer the reader to~\cite{irisJFP}.

The key rules of the update modality are:
\begin{mathpar}
\inferhref{$\pvs$-intro}{upd-intro}{\propMV}{\pvs[\mask][\mask] \propMV}
\and
\inferhref{$\pvs$-mono}{upd-mono}{\propMV \vdash \propMVB}{\pvs[\mask_1][\mask_2] \propMV \vdash \pvs[\mask_1][\mask_2] \propMVB}
\and
\inferhref{$\pvs$-idemp}{upd-idemp}{\pvs[\mask_1][\mask_2] \pvs[\mask_2][\mask_3] \propMV}{\pvs[\mask_1][\mask_3] \propMV}
\and
\inferhrefB{$\pvs$-sep}{upd-sep}{\propMV \ast \pvs[\mask_1][\mask_2] \propMVB}{\pvs[\mask_1][\mask_2] (\propMV \ast \propMVB)}
\end{mathpar}
These rules say that the update modality is a monad, which is indexed (due to the masks), and strong (due to rule \ruleref{upd-sep}).
In \reloc (and Iris) proofs, we often need to eliminate update modalities in the proof context, which is allowed if the goal is an update modality with corresponding source mask.
This is expressed by the following derived rule:
\begin{mathpar}
  \inferhref{$\pvs$-elim}{upd-elim}
  {\pvs[\mask_1][\mask_2] \propMV \and \propMV \wand \pvs[\mask_2][\mask_3] \propMVB}
  {\pvs[\mask_1][\mask_3] \propMVB}
\end{mathpar}
This rule is derivable from \ruleref{upd-mono}, and \ruleref{upd-idemp}.

Before we will describe the rules of the update modality related to invariants, let us describe some syntactic sugar that we inherit from Iris.
We write $\pvs[\mask] \propMV$ for $\pvs[\mask][\mask] \propMV$, and $\pvs \propMV$ for $\pvs[\top] \propMV$, where $\top$ is the set of all invariant names.
Moreover, since the update modality is often combined with the magic wand, we write $\propMV \vsW[\mask_1][\mask_2] \propMVB$ for $\propMV \wand \pvs[\mask_1][\mask_2] \propMVB$, and follow the same conventions for omitting masks on $\vsW$ as used for $\pvs$.

\reloc's main rule for interacting with the update modality is \ruleref{rel-upd}.
It allows to eliminate an update modality around a refinement judgment.
To get an idea of how this rule is used, let us take a look at the primitive rule \ruleref{inv-alloc} for allocating an invariant.
The derived rule \ruleref{rel-inv-alloc'} in \Cref{reloc:fig:rules_tour} is a composition of \ruleref{rel-upd} with Iris's rules \ruleref{upd-elim} and \ruleref{inv-alloc}.

By combining \ruleref{rel-upd} with Iris's rules \ruleref{upd-elim} and \ruleref{inv-access} for accessing invariants, one can turn an invariant $\knowInv \namesp \propMV$ into its content $\propMV$, together with a way of restoring the invariant $\later \propMV \vsW[\mask \setminus \namesp][\mask] \TRUE$.
It is important to notice that by using the combination of these rules, the mask on the refinement judgment changes from $\mask$ into $\mask \setminus \namesp$.
This prohibits access to the invariant $\namesp$ until it has been restored---thus preventing reentrancy.
Restoring the invariant is done by using the rule \ruleref{rel-upd} with the premise $\later \propMV \vsW[\mask \setminus \namesp][\mask] \TRUE$.
This requires one to give up $\propMV$, and in turn transforms the mask of the judgment back into $\mask$.
Note that one can use \ruleref{inv-access} multiple times to open multiple invariants.

\paragraph{Invariants and symbolic execution.}
Opening invariants through \ruleref{rel-upd} and \ruleref{inv-access} as described above is fairly limited.
Once we open an invariant, the mask at the refinement judgment changes from $\top$ into $\top \setminus \namesp$, which prevents any symbolic execution on the left-hand side.
The rules for symbolic execution on that side require the mask to be $\top$.
As we discussed in \Cref{reloc:sec:tour:invariants} already, this restriction to the $\top$ mask on left-hand side rules is crucial.
It is unsound to perform multiple symbolic execution steps on the left while an invariant is open.
To see why this is the case, consider the following refinement:
\[
  \logrel{}{(\Lam x. \Let n = \deref x  in x \gets n + 1; n)}{\FGincrement}{\tref\ {\tint} \to \tint}
\]
This refinement does not hold because the two programs can be distinguished by the context:
\[
  \Let \cntvar = \Ref(0) in \Let f = \hole in \Fork{f\ \cntvar}; f \cntvar.
\]
The left-hand side is basically the coarse-grained increment operation $\CGincrement$ without the lock protection.
Thus, the function on the left-hand side does not guarantee thread-safety: the value of the passed reference can change unpredictably if the function is invoked in parallel with itself.
By contrast, the $\FGincrement$ always increments the counter monotonically.

If we were allowed to perform multiple symbolic execution rules on the left-hand side, then we could have proven the above refinement, using an invariant of $\knowInv{}{\Exists n. \cntvar_s \mapstoS n \ast \cntvar_i \mapstoI n}$.

In order to support symbolic execution with invariants, \reloc provides additional rules to simultaneously access an invariant and perform a single atomic symbolic execution step on the left-hand side.
Examples of such rules are \ruleref{rel-load-l}, \ruleref{rel-store-l} and \ruleref{rel-cas-l}.

We can now explain the derived rule \ruleref{rel-load-l-inv} in terms of the primitive rules.
The proposition $\later \propMV \vsW[\mask \setminus \namesp][\mask] \TRUE$ is used for closing the invariant $\namesp$ because it changes the mask from $\mask \setminus \namesp$ to $\mask$.
Thus $\CLOSE{\namesp}{\propMV} \eqdef (\later \propMV \vsW[\top \setminus \namesp][\top] \TRUE)$.
To prove \ruleref{rel-load-l-inv} from \Cref{reloc:fig:rules_tour}, we apply \ruleref{rel-load-l} to obtain the goal:
\[
\pvs[\top][\top\setminus\namesp] {\left(
  \Exists \val. \loc \mapstoI \val \mathrel{*}
  \later \big(\loc \mapstoI \val \wand
    \logrel[\top\setminus\namesp]{\Delta}{\fillctx\lctx [\val]}{\expr}{\type}\big)
  \right) }.
\]
We then use \ruleref{inv-access} and \ruleref{upd-elim} to get the premise of \ruleref{rel-load-l-inv}.
In the same way \ruleref{rel-cas-l-inv} can be derived from \ruleref{rel-cas-l}.
Finally, the closing rule \ruleref{rel-inv-restore} is a consequence of the definition of $\CLOSE{\namesp}{\propMV}$ and \ruleref{rel-upd}.

Using \reloc's primitive symbolic execution rules such as \ruleref{rel-load-l}, \ruleref{rel-store-l} and \ruleref{rel-cas-l} one can also derive the following weaker, but perhaps more intuitive, symbolic execution rule:
\[
\inferH{rel-store-l'}
  {\loc \mapstoI \val \and
  \later \big(\loc \mapstoI \valB \wand
    \logrel{\Delta}{\fillctx\lctx [\unittt]}{\expr_2}{\type}\big)}
  {\logrel{\Delta}{\fillctx\lctx [\loc \gets \valB]}{\expr_2}{\type}}
\]

Since these rules have a $\top$ mask, they can only be used when no invariants have been opened.
Recall that by contrast, the symbolic execution rules for the right-hand side, such as \ruleref{rel-load-r}, \ruleref{rel-store-r} in \Cref{reloc:fig:rules_tour}, which are of a similar shape, can be performed even with invariants open because they allow the mask to be arbitrary.

\subsection{The persistence modality}
\label{reloc:sec:persistence}

Recall from \Cref{reloc:sec:tour:invariants} that a proposition $\propMV$ is persistent, written as $\persistent{\propMV}$, if $\propMV \vdash \always \propMV$, where $\always$ is Iris's \emph{persistence} modality.
The $\always$ modality plays an important role in \reloc because it makes it possible to express that if two expressions are related, they remain related forever.
For example, the persistence modality plays a crucial role in the rule \ruleref{rel-rec} in \Cref{reloc:fig:rules_tour}---it ensures that ephemeral resources (such as heap assertions) are not used for the verification of the closure's body.
After all, closures can be invoked arbitrarily many times at different points in time (possibly concurrently), and hence it is impossible to guarantee that ephemeral resources will still be available when the closure is called.
For example, without the $\always$ modality in the premise of \ruleref{rel-rec} one would be able to prove the following unsound refinement:
\[
\logrel{}{\Let \loc = \Ref(0) in \Lam \unittt. \loc \gets 1 + \deref\loc;\,\deref \loc\ }{\ \Lam \unittt. 1\ }{\ \tunit \to \tint}.
\]
One would use \ruleref{rel-alloc-l'} to obtain the heap assertion $\loc \mapstoI 0$, and subsequently use that assertion to verify the body of the closure.
Fortunately, the $\always$ modality in \ruleref{rel-rec} prevails---$\loc \mapstoI 0$ is ephemeral, not persistent, so cannot be moved under a $\always$.

In \Cref{reloc:sec:tour:invariants} we gave an idea of the core rules of the persistence modality.
Let us now take a look at the rules in more detail:
\begin{mathparpagebreakable}
\inferhrefB{$\always$-dup}{always-dup}{\always \propMV \ast \always \propMV}{\always \propMV}
\and
\inferhref{$\always$-elim}{always-elim}{\always \propMV}{\propMV}
\and
\inferhref{$\always$-mono}{always-mono}{\propMV \vdash \propMVB}{\always \propMV \vdash \always \propMVB}
\and
\inferhref{$\always$-idemp}{always-idemp}{\always \propMV}{\always \always \propMV}
\and
\inferhrefB{$\always$-sep}{always-sep}{\always \propMV \ast \always \propMVB}{\always (\propMV \ast \propMVB)}
\end{mathparpagebreakable}
The rules \ruleref{always-dup} and \ruleref{always-elim} say that the $\always \propMV$ is duplicable, and one can get $\propMV$ out.
The rule \ruleref{always-idemp} says that $\always \propMV$ itself is persistent.
The rules \ruleref{always-elim}, \ruleref{always-mono} and \ruleref{always-idemp} say that $\always$ is in fact a co-monad.
Finally, $\always$ commutes with most logical connectives, for example, the separating conjunction, as expressed by \ruleref{always-sep}.

If we wish to prove $\always \propMVB$ under the assumptions $\propMV_1, \dots, \propMV_n$, where each $\propMV_i$ is persistent, then we can introduce the $\always$ modality and prove $\propMVB$ from $\propMV_1, \dots, \propMV_n$:
\begin{mathpar}
  \inferhref{$\always$-intro}{always-intro}
  {\persistent{\propMV_1} \and \dots \and \persistent{\propMV_n}
    \and
    {\propMV_1 \ast \dots \ast \propMV_n \vdash \propMVB}
}
  {\propMV_1 \ast \dots \ast \propMV_n \vdash \always \propMVB}
\end{mathpar}
This rule is derivable from the definition of $\persistent{-}$, \ruleref{always-sep}, and \ruleref{always-mono}.

Note that $\persistent{\propMV}$ is defined through the validity relation $\propMV \proves \always \propMV$; \ie it is a meta-logical notion (in terms of the mechanization, $\persistent\propMV$ is a Coq-level predicate, not an Iris-level predicate).
As such, the rule above does not fit the description we have given to the inference rules in \Cref{reloc:subsec:grammar}.
Rather, it should be seen as a family of inference rules indexed by meta-level propositions $\persistent{\propMV_1}, \dots, \persistent{\propMV_n}$.
% In that light, the rule \ruleref{rel-pack} is a family of rules indexed by \persistent{R}.

\subsection{Value interpretation and monadic rules}
\label{reloc:sec:calculus}
In addition to the refinement judgment $\logrel{\Delta}{\expr_1}{\expr_2}{\type}$, which relates expressions $\expr_1$ and $\expr_2$, \reloc provides the value interpretation $\Sem{\type}_\Delta(\val_1,\val_2)$, which relates values $\val_1$ and $\val_2$.
The rule \ruleref{rel-return} expresses that $\Sem{\type}_\Delta(\val_1,\val_2)$ implies $\logrel{\Delta}{\val_2}{\val_2}{\type}$.
However, the inverse direction does not hold, $\Sem{\type}_\Delta(\val_1,\val_2)$ is strictly stronger than $\logrel{\Delta}{\val_1}{\val_2}{\type}$ as its rules (in \Cref{reloc:fig:rules}) are bidirectional, whereas those for the expression judgment are unidirectional.
The bidirectionality is crucial for the rule \ruleref{rel-rec} in \Cref{reloc:fig:rules_tour}, as it contains $\Sem{\type}_\Delta(\val_1,\val_2)$ in negative position---\ie as a client of \ruleref{rel-rec} one gets $\Sem{\type}_\Delta(\val_1,\val_2)$ as an assumption and hence needs to eliminate it.

We want the value interpretation $\logrelV{\type}{\Delta}{\val_1}{\val_2}$ to be persistent, because our type system is not substructural, \ie types denote knowledge, but not ownership of data.
For example, in typing the expression $\expr_1 \gets \expr_2$ with \ruleref{store-typed}, we use the same context $\Gamma$ for type checking both $e_1$ and $e_2$.
In order to semantically validate such rules, we want the propositions $\logrelV{\type}{\Delta}{\val_1}{\val_2}$ to be duplicable.
To that end, we require all the interpretations in the context $\Delta$ to be persistent.
That is why the rule \ruleref{rel-pack} in \Cref{reloc:fig:rules_tour} has a side-condition $\All \val_1, \val_2. \persistent{R(\val_1, \val_2)}$.

The value interpretation also appears in the monadic rules \ruleref{rel-return} and \ruleref{rel-bind} in \Cref{reloc:fig:rules}.
These rules are used to derive all type-directed structural rules of \reloc, with the exception of \ruleref{rel-fork}, which is the sole primitive type-directed structural rule.
As an example, consider the type-directed structural rule for the first projection $\pi_1$.

\begin{lem}
\label[lem]{reloc:lem:related_fst}
The following rule is derivable:
\begin{mathpar}
\infer
  {\logrel{\Delta}{\expr_1}{\expr_2}{\type \times \typeB}}
  {\logrel{\Delta}{\pi_1(\expr_1)}{\pi_1(\expr_2)}{\type}}
\end{mathpar}
\end{lem}
\begin{proof}
By \ruleref{rel-bind} it suffices to show:
\begin{itemize}
\item $\logrel{\Delta}{\expr_1}{\expr_2}{\type \times \typeB}$, but this is exactly our assumption;
\item for any $\val, \valB$: $\logrelV{\type \times \typeB}{\Delta}{\val}{\valB} \wand \logrel{\Delta}{\pi_1(\val)}{\pi_1(\valB)}{\type}$.
\end{itemize}

By \ruleref{val-prod} we have values $\val_i, \valB_i$ for $i \in \{1,2\}$ such that $\val = (\val_1,\val_2)$ and $\valB = (\valB_1,\valB_2)$ and
$\logrelV{\type}{\Delta}{\val_1}{\valB_1} \ast \logrelV{\typeB}{\Delta}{\val_2}{\valB_2}$.
Using \ruleref{rel-pure-l} and \ruleref{rel-pure-r} we reduce the goal
$\logrel{\Delta}{\pi_1(\val_1,\val_2)}{\pi_1(\valB_1,\valB_2)}{\type}$ to
$\logrel{\Delta}{\val_1}{\valB_1}{\type}$.
At this point we apply \ruleref{rel-return}.
\end{proof}

\subsection{Fundamental theorem and refinements of open terms}
\label{reloc:sec:fundamental}
The type-directed structural rules\footnote{Our \emph{type-directed structural rules} are often called \emph{compatibility lemmas} in the logical relation literature.} are also used for proving the following theorem, which is a standard result for logical relation models of type systems:
\begin{thm}[Fundamental theorem for closed terms]
\label{reloc:lem:fundamental_short}
If expression $\expr$ is well typed, \ie $\typed{}{}{\expr}{\type}$, then $\expr$ refines itself, \ie the judgment $\logrel{}{\expr}{\expr}{\type}$ is derivable in \reloc.
\end{thm}

We wish to prove this theorem by induction on the typing derivation.
But in order to make it work, we need to generalize the theorem to open terms (\eg in order to deal with the \ruleref{rec-typed} case).
Consequently, we need to generalize \reloc's refinement judgment $\logrel{\Delta}[\vctx]{\expr_1}{\expr_2}{\type}$ to open terms $\expr_1$ and $\expr_2$ whose free variables are bound by the typing context $\vctx$.
To define the refinement judgment for open terms, we first define a standard notion of a \emph{closing substitution}.

\begin{defi}
  \label[defi]{reloc:def:interp_open_terms}
  A mapping $\gamma : \Var \to \Val \times \Val$ is
  a \emph{closing substitution \wrt the typing environment $\vctx$}, notation
  $\validSubst{\Delta}{\vctx}{\gamma}$, if
  \[
    \All (x, \type) \in \vctx. \logrelV{\type}{\Delta}{\gamma_1(x)}{\gamma_2(x)},
  \]
  where $\gamma_i(x) = \pi_i(\gamma(x))$ is the $i$-th projection of $\gamma(x)$.
\end{defi}

\begin{defi}
  \label[defi]{reloc:def:open_term_refinement}
  The \emph{refinement judgment $\logrel{\Delta}[\vctx]{\expr_1}{\expr_2}{\type}$ for open
  terms} is defined as:
  \[
    \logrel{\Delta}[\vctx]{\expr_1}{\expr_2}{\type}
    ~\eqdef~
    \All \gamma.
      \validSubst{\Delta}{\vctx}{\gamma} \wand \logrel{\Delta}{\gamma_1(\expr_1)}{\gamma_2(\expr_2)}{\type}.
  \]
\end{defi}

Using the refinement judgment for open terms we can now state versions of the type-directed structural rules for open terms.
For example:
\[
\infer
  {\logrel{\Delta}[\vctxinsert{x}{\type}{\vctx}]{\expr_1}{\expr_2}{\typeB}}
  {\logrel{\Delta}[\vctx]{\Lam x. \expr_1}{\Lam x. \expr_2}{\type \to \typeB}}
\]
This rule can be proven by unfolding the definition of the refinement judgment for open terms and proceeding as in \Cref{reloc:lem:related_fst}.
Finally, we can state and prove the generalized version of the fundamental theorem for open terms:

\begin{thm}[Fundamental theorem for open terms]
\label{reloc:lem:fundamental}
If $\typed{\tenv}{\vctx}{\expr}{\type}$, then $\logrel{\Delta}[\vctx]{\expr}{\expr}{\type}$ is derivable in \reloc, for all $\Delta$ which contain the variables from $\tenv$.
\end{thm}
\begin{proof}
 By induction on the typing derivation, using the versions of the type-directed structural rules for open terms.
\end{proof}

With the refinement judgments generalized to open terms, we can state and the generalized version of \Cref{reloc:thm:soundness} for open terms, which we prove in \Cref{reloc:sec:model:soundness}.

\begin{thm}[Soundness for open terms]
\label{reloc:thm:soundness_full}
Let $\tenv$ be a type environment.
Suppose that refinement judgment $\logrel{\Delta}[\vctx]{\expr_1}{\expr_2}{\type}$ is derivable in \reloc for all $\Delta$ which contain the variables from $\tenv$.
Then $\ctxref{\tenv \mid \vctx}{\expr_1}{\expr_2}{\type}$.
\end{thm}

%%% Local Variables:
%%% mode: latex
%%% TeX-master: "reloc"
%%% End:

%  LocalWords:  upd logrel inv alloc cas tlam rules_tour nat thm
%  LocalWords:  casestudies

\section{Relational specifications in \reloc}
\label{reloc:sec:atomicity}
Due to its first-class refinement judgments, \reloc can be used to give \emph{relational specifications} to programs.
Similar to Hoare triples, relational specifications abstract away from a program's implementation by expressing its behavior in terms of a pre- and postcondition.
Relational specifications apply to the situation when the expression on the one side of the refinement contains a program subject to specification, while the expression on the other side is arbitrary.
In \Cref{reloc:fig:locks} in \Cref{reloc:sec:tour} we saw an example of a right-hand side relational specifications for locks, which we then used to verify a counter module.

We start this section by describing the general format of non-atomic relational specifications (\Cref{reloc:subsec:simple_relational_specs}).
Non-atomic relational specifications are sufficient to give strong specifications for the right-hand left, but due to the demonic nature of the left-hand side, we often need stronger specifications for the left-hand side (\Cref{reloc:subsec:why_atomic}).
We therefore introduce \emph{logically atomic relation specifications}, which generalize da Rocha Pinto \etal's TaDA-style specifications~\cite{daRochaPinto:dinsdale-young:gardner:2014} (\Cref{reloc:subsec:tada_style}) and Svendsen \etal's HOCAP-style specifications~\cite{HOCAP} (\Cref{reloc:subsec:hocap_style}) from the Hoare-logic setting to the relational setting.
Finally, we show how to use logically atomic specifications to verify a ticket lock (\Cref{reloc:sec:ticket_lock}).

\subsection{Non-atomic relational specifications}
\label{reloc:subsec:simple_relational_specs}

\subsubsection{Right-hand side relational specifications}
\label{reloc:subsec:specs_rhs}

Consider the following implementation of a lock, which we refer to as a \emph{spin lock}:
\begin{align*}
\newlock \eqdef{} & \Lam \unittt. \Ref(\<false>) \\
\acquire \eqdef{} & \Lam \loc. \If \CAS{\loc}{\<false>}{\<true>} then \unittt \Else
                        \acquire\ \loc \\
\release \eqdef{} & \Lam \loc. \loc \gets \<false>
\end{align*}
For this specific implementation, we can prove the rules in \Cref{reloc:fig:locks} in \Cref{reloc:sec:counterproof} by defining the lock predicates as follows:
\[
\isLock{\lkvar}{b} \eqdef \lkvar \in \Loc \ast \lkvar \mapstoS b.
\]
The rules for locks in \Cref{reloc:fig:locks} follow a certain pattern.
For an expression $\expr_2$ that, under precondition $\textcolor{precond}{\propMV}$, reduces to $\val$, with postcondition $\textcolor{postcond}{\propMVB(\var,\val)}$, we have the following rule:
\[
\infer
  {{\color{precond}\propMV} \and
  \All \var,\val. {\color{postcond}\propMVB(\var,\val)} \wand \logrel[\mask]{\Delta}{\expr_1}{\fillctx\lctx[\val]}{\type}}
  {\logrel[\mask]{\Delta}{\expr_1}{\fillctx\lctx[\expr_2]}{\type}}
\]
The postcondition $\propMVB : X \times \Val \to \Prop$ also depends on a type $X$, provided by the provider of the rule.
This rule pattern can be considered a relational version of a Hoare triple for a program on the right-hand side of the refinement judgment.

\subsubsection{Left-hand side relational specifications}
\label{reloc:sec:simpl_specs_lhs}

\begin{figure}[t]
\begin{mathpar}
\inferH{newlock-l}
{R \and \All \lkvar,\, \gname. \islockI(\gname, \lkvar, R) \wand
\logrel{}{\fillctx\lctx[\lkvar]}{\expr_2}{\type}}
{\logrel{}{\fillctx\lctx[\newlock\ \unittt]}{\expr_2}{\type}}
\and
\inferH{is-lock-pers}
{\islockI(\gname, \lkvar, R)}
{\always \islockI(\gname, \lkvar, R)}
\and
\inferH{acquire-l}
{\islockI(\gname, \lkvar, R) \and
(\locked{\gname} \wand R \wand \logrel{}{\fillctx\lctx[\unittt]}{\expr_2}{\type})}
{\logrel{}{\fillctx\lctx[\acquire\ \lkvar]}{\expr_2}{\type}}
\and
\inferH{release-l}
{\islockI(\gname, \lkvar, R) \and \locked{\gname} \and R \and
\logrel{}{\fillctx\lctx[\unittt]}{\expr_2}{\type}}
{\logrel{}{\fillctx\lctx[\release\ \lkvar]}{\expr_2}{\type}}
\end{mathpar}
\caption{Left-hand side relational specification for locks.}
\label{reloc:fig:locks_l}
\end{figure}

We can formulate a similar pattern for programs on the left-hand side of the refinement judgment:
\[
\infer
  {{\color{precond}\propMV} \and
  \All \var,\val. {\color{postcond}\propMVB(\var,\val)} \wand \logrel{\Delta}{\fillctx\lctx[\val]}{\expr_2}{\type}}
  {\logrel{\Delta}{\fillctx\lctx[\expr_1]}{\expr_2}{\type}}
\]
Using this pattern we can state and prove a relational version of the standard separation logic specification for locks, which is shown in \Cref{reloc:fig:locks_l}.
This specification makes use of the lock predicate $\islockI(\gname, \lkvar, R)$, which states that $\lkvar$ protects the resources described by the proposition $R$.
When creating a lock using $\newlock$, the resources $R$ have to be given up, and the persistent lock predicate $\islockI(\gname, \lkvar, R)$ is given in return.
A thread that acquires a lock by calling $\acquire$ gets access to $R$ for the duration of the critical section, and has to give $R$ back when calling $\release$.
The token $\locked{\gname}$, where $\gname$ is a \emph{ghost name} associated to the lock, makes sure that a lock can only be released when it has been acquired.
To prove the left-hand side specification for the spin lock, we define the lock predicate $\islockI(\gname, \lkvar, R)$ following the usual definition in Iris:
\[
\islockI(\gname, \lkvar, R) \eqdef
  \lkvar \in \Loc \ast
  \knowInv{\namesp_{\mathsf{lock}}}{
    (\lkvar \mapstoI \<false> * \locked{\gname} * R) \lor
    (\lkvar \mapstoI \<true>)
  }
\]
This definition uses an Iris invariant to express that the lock is either unlocked ($\lkvar \mapstoI \<false>$), in which case it holds the token $\locked{\gname}$ and the resources $R$, or locked ($\lkvar \mapstoI \<true>$), in which case it holds no resources, since those are held by the thread that acquired the lock.
The token $\locked{\gname}$ is an exclusive resource that is obtained from Iris's ghost theory.

\subsubsection{Left- versus right-hand side relational specifications}
As we have seen in the specifications for the lock, there is an asymmetry between the left- and right-hand side specifications.
The left-hand side specification of $\acquire$ (\ruleref{acquire-l}) can be used regardless of whether the lock is unlocked, whereas the right-hand specification (\ruleref{acquire-r}) can be used solely if the lock is unlocked (\ie if $\isLock{\lkvar}{\<false>}$).
This is due to the demonic nature of the left-hand side and the angelic nature of the right-hand side.
For $\acquire$ on the left-hand side, we have to consider an arbitrary execution, whereas for $\acquire$ on the right-hand side we have to provide an execution ourselves.
That is, for $\acquire$ on the right-hand side we have to show that it actually acquires the lock and reduces to $\unittt$,
which is only possible when the lock is unlocked.
For this reason we use the predicate $\isLock{\lkvar}{b}$, which tracks the state $b$ of the lock.

\subsection{The need for logically atomic specifications}
\label{reloc:subsec:why_atomic}

Recall from \Cref{reloc:sec:invariants_update} that for any primitive (stateful) operation we have a symbolic execution rule that allows the client to access shared resourced stored in an invariant.
For example, the rule \ruleref{rel-store-l} for the store operation is as follows:
\begin{mathpar}
\infer
  {\pvs[\top][\mask] {\left(
    \loc \mapstoI - *
    \later \big(\loc \mapstoI \val \wand
      \logrel[\mask]{\Delta}{\fillctx\lctx [\unittt]}{\expr_2}{\type}\big)
  \right) }}
  {\logrel{\Delta}{\fillctx\lctx [\loc \gets \val]}{\expr_2}{\type}}
\end{mathpar}
Concretely, the update modality $\pvs[\top][\mask]$ in the premise of this rule allows users to use \ruleref{inv-access} to access an invariant for the duration of the operation.
The mask $\mask$ in the refinement judgment $\logrel[\mask]{\Delta}{\fillctx\lctx [\unittt]}{\expr_2}{\type}$ (that appears in the premise of the rule) forces the user to close the invariant at the end of the duration of the operation.
The ability to open an invariant is sound because operations such as store are \emph{physically atomic}---\ie they reduce in one step.
As a consequence of being physically atomic, other threads cannot observe that the invariant has been broken during the execution of the operation.

In contrast, methods of a concurrent program module are typically composed of several operations and hence they are not physically atomic.
For example, consider the increment function $\FGincrement$ of the fine-grained counter module from \Cref{reloc:fig:impl}:
\[
\FGincrement \eqdef{} \Rec {\lvar{inc}} \cntvar =
	\begin{array}[t]{@{} l}
  \Let n = \deref \cntvar in \\
  \If \CAS{\cntvar}{n}{1 + n} then n \Else \lvar{inc}\ \cntvar
  \end{array}
\]
This function is a compound expression that does not reduce to a value in a single step.
Nevertheless, during the execution of this function there is a single instant at which the whole operation actually appears to take the effect---namely the successful reduction of the compare-and-set operation ($\<CAS>$).
This instant is called the \emph{linearization point}.
What it means is that, for an outside observer, the method $\FGincrement$ behaves \emph{as if} it was atomic, and we wish to express that in this function's relational specification.

This phenomenon is called \emph{logical atomicity} in the literature, and has been studied extensively in the context of Hoare-style logics~\cite{jacobs:piessens:2011,daRochaPinto:dinsdale-young:gardner:2014,HOCAP,iris1,irisProph}.
In the upcoming subsections we will how to generalize the concept of logical atomicity to the relational setting, and how that gives rise to \emph{logically atomic relational specifications}.
Concretely, we generalize da Rocha Pinto \etal's TaDA-style specifications~\cite{daRochaPinto:dinsdale-young:gardner:2014} (\Cref{reloc:subsec:tada_style}) and Svendsen \etal's HOCAP-style specifications~\cite{HOCAP} (\Cref{reloc:subsec:hocap_style}) from the Hoare-logic setting to the relational setting.
Establishing the formal comparison between the two styles is out of the scope of this paper.
Rather, we demonstrate that both approaches can be applied to the context of relational specifications.

\subsection{TaDA-style relational specifications}
\label{reloc:subsec:tada_style}
\subsubsection{Formulating TaDA-style specifications}
\label{reloc:subsec:tada_style_formulating}

We take inspiration from the encoding of TaDA-style logically atomic Hoare triples in Iris~\cite{iris1} and assign the following logically atomic relational specification to $\FGincrement$:
\begin{mathpar}
\inferH{inc-i-l-tada}
  {\always
  \pvs[\top][\mask] \Exists n. \cntvar \mapstoI n *
  \left(
  {\begin{array}{@{} l @{}}
  \big(\cntvar \mapstoI n \vsW[\mask][\top] \TRUE\big)
  \mathrel{\land} \\
  \big(\cntvar \mapstoI (n + 1) \wand \logrel[\mask]{}{\fillctx\lctx[n]}{\expr}{\type}\big)
  \end{array}}
  \right)}
  {\logrel{}{\fillctx\lctx[\FGincrement\ \cntvar]}{\expr}{\type}}
\end{mathpar}

Contrary to the non-atomic specification, we do not have $\cntvar \mapstoI n$ as a premise of the rule directly, but instead the premise contains a way of obtaining $\cntvar \mapstoI n$.
The typical way of obtaining $\cntvar \mapstoI n$ is by accessing an invariant, which is formally done by using the update modality $\pvs[\top][\mask]$ in the premise combined with \ruleref{inv-access} from \Cref{reloc:fig:rules}.

To justify the remaining part of the premise of the rule we need to take a closer took at the behavior of $\FGincrement\ \cntvar$, whose implementation (\Cref{reloc:fig:impl}) we recall to be as follows:
\[
\FGincrement \eqdef{} \Rec {\lvar{inc}} \cntvar =
	\begin{array}[t]{@{} l}
  \Let n = \deref \cntvar in \\
  \If \CAS{\cntvar}{n}{1 + n} then n \Else \lvar{inc}\ \cntvar
  \end{array}
\]
The compare-and-set operation ($\<CAS>$) can either succeed or fail.
If it succeeds, then we have managed to update our resources to $\cntvar \mapstoI (n + 1)$, and we can proceed with proving $\logrel[\mask]{}{\fillctx\lctx[n]}{\expr}{\type}$ under that premise.
This explains the $(\cntvar \mapstoI (n + 1) \wand \logrel[\mask]{}{\fillctx\lctx[n]}{\expr}{\type})$ clause.
If, however, the compare-and-set fails, then we need to be able to restart the whole computation of $\FGincrement\ \cntvar$.
For that we must be able to return $\cntvar \mapstoI n$ to the invariant.
Hence the $(\cntvar \mapstoI n \vsW[\mask][\top] \TRUE)$ clause.
(The same clause is used for performing operations that do not modify the state, such as dereferencing.)

Finally, we know that the computation can either succeed or be restarted---but not both.
We have to accommodate for both situations, just not at the same time.
Hence the last two clauses described here are connected by an intuitionistic conjunction ($\wedge$), instead of the separating conjunction ($*$).

\subsubsection{Using TaDA-style specifications}
\label{reloc:subsec:tada_style_using}
We use the logically atomic relational specification \ruleref{inc-i-l-tada} to prove the refinement that we have seen in \Cref{reloc:sec:tour:invariants}.
The new proof is more modular since it does not appeal to the definition of $\FGincrement$.
The refinement that we want to prove is as follows:
\[
\knowInv{\counterN}{\counterinv} \wand{}
\logrel{}{\FGincrement\ c_i}{\CGincrement\ c_s\ l}{\tint}
\]
Recall that $\counterinv \eqdef {\Exists n \in \nat. c_i \mapstoI n * c_s \mapstoS n * \isLock{\lkvar}{\<false>}}$.
To prove this goal, we use \ruleref{inc-i-l-tada}.
After introducing the persistence modality (using \ruleref{always-intro}, which is allowed, because there are no ephemeral assumptions in our context), this gives the following new goal (under the assumption $\knowInv{\counterN}{\counterinv}$):
\begin{equation*}
  \pvs[\top][\top \setminus \counterN] \Exists n. c_i \mapstoI n \ast
  \left(
  \begin{array}{@{} l @{}}
  \big(c_i \mapstoI n \vsW[\top \setminus \counterN][\top] \TRUE\big)
  \mathrel{\land} \\
  \big(c_i \mapstoI (n + 1) \wand \logrel[\top \setminus \counterN]{}{n}{\CGincrement\ c_s\ l}{\tint}\big)
  \end{array}
  \right)
\end{equation*}
At this point we can open up the invariant $\counterinv$ (using \ruleref{inv-access}), and thereby introduce the update modality.
The contents of the invariant provides us with a witness for the existential quantifier and allows us to discharge $\cntvar \mapstoI n$.
We are left with proving the conjunction:
\begin{equation*}
  \big(c_i \mapstoI n \vsW[\top \setminus \counterN][\top] \TRUE\big)
  \quad \land \quad
  \big(c_i \mapstoI (n + 1) \wand \logrel[\top \setminus \counterN]{}{n}{\CGincrement\ c_s\ l}{\tint}\big)
\end{equation*}
under the assumption of the unused resources $\isLock{l}{\<false>}$ and $c_s \mapstoS n$ from the invariant, and the invariant closing resource $\later \counterinv \vsW[\top \setminus \counterN][\top] \TRUE$.

The first conjunct corresponds to the case in which we close the invariant without modifying anything in our current context (\ie the compare-and-set has failed).
It follows directly from the invariant closing resource.
It thus remains to prove the second conjunct (\ie the compare-and-set has succeeded), which means we should prove
$
\logrel[\top \setminus \counterN]{}{n}{\CGincrement\ c_s\ l}{\tint}
$
under the assumptions $c_i \mapstoI (n+1)$ and $\isLock{l}{\<false>}$ and $c_s \mapstoS n$ and $\later \counterinv \vsW[\top \setminus \counterN][\top] \TRUE$.
At this point we finish the proof by symbolically executing $\CGincrement\ c_s\ l$ on the right-hand side before closing the invariant using invariant closing resource.

\subsubsection{General format of TaDA-style specifications}
\label{reloc:subsec:tada_general_form}
The general format of logically atomic rules for logical refinements is the following:
\begin{mathpar}
  \inferH{}
  {\propMVC \and
  \always \pvs[\top][\mask]
  \Exists x. \textcolor{precond}{\propMV(x)} *
  \left({
  \begin{array}{@{} l @{}}
  \big(\textcolor{precond}{\propMV(x)} \vsW[\mask][\top] \TRUE\big)
  \mathrel{\land} \\
  \big(\All \val.
    \textcolor{postcond}{\propMVB(x,\val)} * \propMVC \wand
    \logrel[\mask]{}{\fillctx\lctx[\val]}{\expr_2}{\type}\big)
  \end{array}
  }\right)}
  {\logrel{}{\fillctx\lctx[\expr_1]}{\expr_2}{\type}}
\end{mathpar}
Here, $\propMV : X \to \Prop$ is a predicate describing consumed resources, and $\propMVB : X \times \Val \to \Prop$ is a predicate describing produced resources, both dependent on a type $X$ supplied by the provider of the rule (\eg a library that exports the program $\expr_1$).
This parameter $X$ is selected on per-specification basis.
For example, for the counter module $X$ is going to be the type of natural numbers.

We include a frame $\propMVC$, which can be chosen by the client, for the following reason.
The second premise of the rule resides below a persistence modality.
Whenever we prove a goal of the form $\always P$ we must prove $P$ using only persistent resources, and thus have to throw all the ephemeral resources away (see \ruleref{always-intro} in \Cref{reloc:sec:persistence}).
However, we do not want to throw away all the ephemeral resources that we have for eternity (as they might be needed to close invariants afterwards or to proceed otherwise with the proof), so we give them up only temporarily, by collecting them in $\propMVC$.

\subsubsection{Proving TaDA-style specifications}
Following the general scheme, we now state and prove the TaDA-style specification of our increment function:
\begin{mathpar}
\inferH{inc-i-l-tada-gen}
  {\propMVC \and
  \pvs[\top][\mask] \always \Exists n. \cntvar \mapstoI n *
  \left(
  {\begin{array}{@{} l @{}}
  \big(\cntvar \mapstoI n \vsW[\mask][\top] \TRUE\big)
  \mathrel{\land} \\
  \big(\cntvar \mapstoI (n + 1) \ast \propMVC \wand \logrel[\mask]{}{\fillctx\lctx[n]}{\expr}{\type}\big)
  \end{array}}
  \right)}
  {\logrel{}{\fillctx\lctx[\FGincrement\ \cntvar]}{\expr}{\type}}
\end{mathpar}

To prove this specification, we proceed by L\"ob induction and symbolic execution.
At the point when we need to symbolically dereference $\cntvar$ we apply \ruleref{rel-load-l}.
We then use the update that we have as a premise of the specification to obtain $\cntvar \mapstoI n$ for some $n$.
After providing $\cntvar \mapstoI n$ for the load operation, we use the first conjunct
$
  \cntvar \mapstoI n * \vsW[\mask][\top] \TRUE
$
to restore the mask on the refinement judgment.

After that we have to symbolically execute the compare-and-set operation; we apply \ruleref{rel-cas-l} and use the update that we have as a premise again to obtain $\cntvar \mapstoI m$ for some $m$, as needed for \ruleref{rel-cas-l}.
If $m \neq n$, then the compare-and-set operation has failed, and we can restart the proof first by restoring the mask on the refinement judgment (using the closing update), and then using the L\"ob induction hypothesis.
If $m = n$, then the compare-and-set operation has succeeded, and the points-to connective is updated to $\cntvar \mapstoI (n+1)$.
Then we can use the second conjunct
$
\cntvar \mapstoI (n + 1) \ast \propMVC \wand \logrel[\mask]{}{\fillctx\lctx[n]}{\expr}{\type}
$
to arrive at the exact conclusion that we need: $\logrel[\mask]{}{\fillctx\lctx[n]}{\expr}{\type}$.

\subsection{HOCAP-style relational specifications}
\label{reloc:subsec:hocap_style}

We now present another form of logically atomic relational specifications---\emph{HOCAP-style logical atomic relational specifications}, which are based on the logically atomic specifications by Svendsen \etal~\cite{HOCAP} in the eponymous logic.
Contrary to TaDA-style specifications, which come in a precisely specified format (\Cref{reloc:subsec:tada_general_form}), HOCAP-style specifications do not have a precise format.
This provides the flexibility that not only can they be used to represent linearization points, but they can also be used to represent arbitrary observable interactions with the abstract state.
This flexibility allows us to give strong specifications to non-linearizable methods (\Cref{reloc:subsec:hocap_nonlin_operations}), which we use in the ticket lock refinement proof (\Cref{reloc:sec:ticket_lock}).

\subsubsection{Formulating HOCAP-style specifications}
\label{reloc:subsec:counter_hocap_specs}

\begin{figure}
\raggedright
\textbf{Rules for abstract predicates}:
\medskip
\begin{mathpar}
\inferH{cnt-agree}
  {\cntAuth{\gname}{n} \and \cnt{\gname}{q}{m}}
  {n = m}
\and
\inferH{cnt-agree'}
  {\cnt{\gname}{q_1}{n} \and \cnt{\gname}{q_2}{m}}
  {n = m}
\and
\inferHB{cnt-combine}
  {\cnt{\gname}{q_1}{n} \ast \cnt{\gname}{q_2}{n}}
  {\cnt{\gname}{q_1 + q_2}{n}}
\and
\inferH{cnt-update}
  {\cntAuth{\gname}{n} \and \cnt{\gname}{1}{m}}
  {\pvs \cntAuth{\gname}{k} \ast \cnt{\gname}{1}{k}}
\and
\inferH{Cnt-persistent}
  {\Cnt{\cntvar}{\gname}}
  {\always \Cnt{\cntvar}{\gname}}
\end{mathpar}

\bigskip
\textbf{Relational specification}:
\medskip
\begin{mathparpagebreakable}
\inferH{new-i-l-hocap}
  {\All \cntvar\,\gname. \Cnt{\cntvar}{\gname} \ast \cnt{\gname}{1}{n} \wand{} \logrel{}{\fillctx\lctx[\cntvar]}{\expr_2}{\type}}
  {\logrel{}{\fillctx\lctx[\Ref(n)]}{\expr_2}{\type}}
\and
\inferH{inc-i-l-hocap}
  {\mask \cap \namecl{\namesp} = \emptyset \and
   \Cnt{\cntvar}{\gname} \and \\
   (\All n. \cntAuth{\gname}{n}
   \vsW[\top\setminus \namesp][\top\setminus \namesp \setminus \mask]
   \cntAuth{\gname}{n+1} \ast{} \logrel[\top\setminus\mask]{}{\fillctx\lctx[n]}{\expr_2}{\type})}
  {\logrel{}{\fillctx\lctx[\FGincrement\ \cntvar]}{\expr_2}{\type}}
\and
\inferH{read-i-l-hocap}
  {\mask \cap \namecl{\namesp} = \emptyset \and
   \Cnt{\cntvar}{\gname} \and \\
   (\All n. \cntAuth{\gname}{n}
   \vsW[\top\setminus \namesp][\top\setminus  \namesp \setminus \mask]
   \cntAuth{\gname}{n} \ast{} \logrel[\top\setminus\mask]{}{\fillctx\lctx[n]}{\expr_2}{\type})}
  {\logrel{}{\fillctx\lctx[\counterread\ \cntvar]}{\expr_2}{\type}}
\end{mathparpagebreakable}
\caption{HOCAP-style logically atomic relational specification for a fine-grained concurrent counter.}
\label{reloc:fig:hocap_cnt_model}
\end{figure}

Let us consider the HOCAP-style specification in \Cref{reloc:fig:hocap_cnt_model} for the fine-grained concurrent counter from \Cref{reloc:fig:impl} in \Cref{reloc:sec:intro}.
Contrary to the TaDA-style specification, the HOCAP-style specification does not expose the underlying state of the counter (\ie $\loc \mapstoI n$) directly, but instead provides an abstract view of the state through abstract predicates.

The persistent predicate $\Cnt{\cntvar}{\gname}$ asserts that the value $\cntvar$ represents a counter.
The specification is parameterized by a namespace $\namesp$ for the internal invariants associated with the specification.
The ghost name $\gname$ is used to link $\cntvar$ with the predicates $\cnt{\gname}{q}{n}$ and $\cntAuth{\gname}{n}$, which describe the abstract state of the counter.
The predicate $\cnt{\gname}{q}{n}$ provides the \emph{client view} of the abstract state of the counter.
It is similar to the fractional heap points-to connective from separation logic---it associates a value (a natural number $n$) to the counter, and can be split and combined according to the fractional component $q$ (\ruleref{cnt-agree'}, \ruleref{cnt-combine}).
The predicate $\cntAuth{\gname}{m}$ provides the \emph{module view} of the abstract state of the counter.
It agrees with the client view (\ruleref{cnt-agree}) and can be used together with $\cnt{\gname}{1}{n}$ to update the value associated to the counter (\ruleref{cnt-update}).

Ownership of the module view predicate $\cntAuth{\gname}{m}$ is given to the user only during the execution of the counter operations.
Consider, for example, the specification \ruleref{inc-i-l-hocap} for the atomic increment function $\FGincrement$.
From the client's point of view, there is only one place where $\FGincrement$ observably interacts with the abstract state of the counter, namely during its linearization point.
For this point of access, the user has to provide the update:
\[
\cntAuth{\gname}{n}
\vsW[\top\setminus \namesp][\top\setminus  \namesp \setminus \mask]
\cntAuth{\gname}{n+1} \ast{} \logrel[\top\setminus\mask]{}{\fillctx\lctx[n]}{\expr_2}{\type}.
\]
The user is given the module view $\cntAuth{\gname}{n}$ of the counter, and has to update it to $\cntAuth{\gname}{n+1}$.
For that, the user has to appeal to \ruleref{cnt-update} and has to provide $\cnt{\gname}{q}{n}$ themselves (either as an immediate resource or from an invariant in $\mask$, which can be opened thanks to the update modality).
After the abstract state is updated, the user has to prove the refinement judgment $\logrel[\top\setminus\mask]{}{\fillctx\lctx[n]}{\expr_2}{\type}$.
Similar to the TaDA-style specifications, the mask on the refinement is set to $\top\setminus\mask$, allowing the user to perform some reasoning on the right-hand side before closing all the invariants from $\mask$.

\subsubsection{Using HOCAP-style specifications}
We use the HOCAP-style logically atomic relational specification from \Cref{reloc:fig:hocap_cnt_model} to prove the refinement that we have seen in \Cref{reloc:sec:tour:invariants}:
\[
\logrel{}{\FGcounter}{\CGcounter}
	{(\tunit \to \tint) \times (\tunit \to \tint)}.
\]
Since the HOCAP-style specifications are stated in terms of abstract predicates, instead of the $\loc \mapstoI n$ connective, we need a slightly different invariant than the one we used for the proof using the TaDA-style specification in \Cref{reloc:subsec:tada_style_using}, namely:
\[
\CntN{c_i}{\gname}{\namesp\!.{\mathsf{cnt}}} \ast
\knowInv{\namesp\!.{\mathsf{inv}}}{\Exists n \in \nat. \cnt{\gname}{1}{n} * c_s \mapstoS n * \isLock{\lkvar}{\<false>}}.
\]
After having established this invariant, the refinement proofs for the increment and read methods proceed similar to the corresponding TaDA-style proofs (\Cref{reloc:subsec:tada_style_using}), except that in the increment case the user has to use \ruleref{cnt-update} alongside \ruleref{inc-i-l-hocap} in order to update the ghost state $\cnt{\gname}{1}{n}$ accordingly.
We do not give the proof here, and direct the reader to the accompanying Coq mechanization.

In \Cref{reloc:sec:ticket_lock} we will another example of a client using the HOCAP-style specification for the counter.

\subsubsection{Proving HOCAP-style specifications}
We discuss how to prove that the implementation of the fine-grained counter meets the HOCAP-style specifications.
To do so, we first use Iris's ghost theory to define the predicates $\cnt{\gname}{q}{n}$ and $\cntAuth{\gname}{n}$ (the details of this definition are omitted).
We then define the predicate $\Cnt{\cntvar}{\gname}$, which provides the internal invariant of the module:
\[
\Cnt{\cntvar}{\gname} \eqdef \cntvar \in \Loc \ast
\knowInv{\namesp}{\Exists n \in \mathbb{N}. \cntvar \mapstoI n \ast \cntAuth{\gname}{n}}
\]
This invariant states that the physical value $n$ of $\cntvar$ corresponds to the logical value $n$ of the predicate $\cntAuth{\gname}{n}$.
To see how this invariant is used, let us consider the proof of \ruleref{inc-i-l-hocap} for the $\FGincrement$ operation.
Since $\FGincrement$ is defined recursively, we prove this rule by L\"ob induction.
We proceed by symbolically executing the left-hand side, accessing the invariant $\namesp$ to dereference $\cntvar$ for some value $n$.
It then remains to show:
\[
  \logrel{}{\If \CAS{\cntvar}{n}{1 + n} then n \Else \FGincrement\ \cntvar}{\expr_2}{\type}.
\]
At this point we use the atomic symbolic execution rule for compare-and-set \ruleref{rel-cas-l} (with $\mask = \namesp$).
We introduce the update modality $\pvs[\top][\top \setminus \namesp]$ we obtain by accessing the invariant $\namesp$ using Iris's strong invariant access rule \ruleref{inv-access-strong}, which is needed because we need to access invariants in a non-well-bracketed way:
\begin{mathpar}
\inferH{inv-access-strong}
{\namecl{\namesp} \subseteq \mask}
{\knowInv{\namesp}{P} \vsW[\mask][\mask \setminus \namesp]
\later P \ast (\All \mask'. \later P \vsW[\mask'][\mask' \cup  \namesp] \TRUE)}
\end{mathpar}
It remains to consider two cases.
If the compare-and-set ($\<CAS>$) has failed, we close the invariant (by setting $\mask' = \top \setminus \namesp$) and appeal to the induction hypothesis.
If the compare-and-set has succeeded, we are left to show the following:
\[
  \dots \ast \cntvar \mapstoI (n+1) \ast \cntAuth{\gname}{n} \wand{} \logrel[\top \setminus  \namesp]{}{n}{\expr_2}{\type}.
\]
We first use the update $\vsW[\top\setminus \namesp][\top\setminus \namesp \setminus \mask]$ that is provided by the premise of the rule to update $\cntAuth{\gname}{n}$ into $\cntAuth{\gname}{n+1}$.
This moreover provides a proof of $\logrel[\top \setminus \mask]{}{n}{\expr_2}{\type}$.
At this point our goal becomes
$
  \logrel[\top \setminus  \namesp \setminus \mask]{}{n}{\expr_2}{\type}
$.
We close the invariant $\namesp$ (by setting $\mask' = \top \setminus  \namesp \setminus \mask$) and restore the mask on the refinement proposition in our goal, resulting in
$
  \logrel[\top \setminus \mask]{}{n}{\expr_2}{\type}
$, which is exactly what we obtained from the update $\vsW[\top\setminus \namesp][\top\setminus \namesp \setminus \mask]$.

\subsubsection{HOCAP-style specifications for non-linearizable operations}
\label{reloc:subsec:hocap_nonlin_operations}
Using HOCAP-style specifications we can also specify operations that are not linearizable.
Consider the following ``weak increment'' function that we can add to the counter module:
\[
\wkincr \eqdef{} \Lam \cntvar. \cntvar \gets (\deref \cntvar + 1)
\]
The function increments the value in the location $\cntvar$ non-atomically.
What kind of specification can we give to $\wkincr$?
To answer this we have to examine what the update $\vsW$ represents in the HOCAP-style specifications.
In the previous examples with linearizable functions, the updates represented observations about the abstract state that the clients could make, and they corresponded to the linearization points.
But there is no reason why we should pin them to linearization points only.
Rather, we can let the update correspond to any operation that is observable through the abstract state.
In $\wkincr$ there are two points where such operations happen, which we can represent through two nested updates:
\begin{mathpar}
\inferH{cnt-wk-incr-l}
  {\mask \cap \namecl{\namesp} = \emptyset \and
  \Cnt{\cntvar}{\gname} \and \\
  \All n. \cntAuth{\gname}{n} \vsW[\top\setminus \namesp]
  \cntAuth{\gname}{n} \ast{}
  (\All m.
    {\begin{array}[t]{@{} l}
    \cntAuth{\gname}{m} \vsW[\top\setminus \namesp][\top\setminus  \namesp \setminus \mask] \\
    \cntAuth{\gname}{n+1} \ast{} \logrel[\top\setminus\mask]{}{\fillctx\lctx[\unittt]}{\expr_2}{\type})
    \end{array}}}
  {\logrel{}{\fillctx\lctx[\wkincr\ \cntvar]}{\expr_2}{\type}}
\end{mathpar}
The first update binds the value $n$, which is obtained from the initial read operation $\deref \cntvar$.
In the conclusion of this update the client needs to return the $\cntAuth{\gname}{n}$ predicate, as in the specification \ruleref{read-i-l-hocap}.
In addition to that predicate, the client has to provide the second update, in which $\cntAuth{\gname}{m}$ has to be updated to $\cntAuth{\gname}{n+1}$, corresponding to the assignment $\cntvar \gets n+1$.
The value $m$ corresponds to the intermediate state of the counter, which might have changed in between the dereferencing of $\cntvar$ and the assignment to it.
The presence of two updates differentiates methods that have a linearization point, such as $\FGincrement$, and non-linearizable methods, such as $\wkincr$.

In the next section we will see how a client might use the specification \ruleref{cnt-wk-incr-l}.

\subsection{Ticket lock refinement from HOCAP-style specs}
\label{reloc:sec:ticket_lock}
\begin{figure}
\begin{align*}
\newlockTicket \eqdef{} & \Lam \unittt. (\Ref(0),\Ref(0)) \displaybreak[4] \\
\acquireTicket \eqdef{} & \Lam (\llo, \lln).
  \Let n = \FGincrement\ \lln in
  \textlog{wait\_loop}\ n\ \llo \\
\textlog{wait\_loop} \eqdef{}& \Lam n\ \llo.
   \If {(n = \counterread\ \llo)} then \unittt \Else \textlog{wait\_loop}\ n\ \llo \\
\releaseTicket \eqdef{}& \Lam (\llo,\lln). \wkincr\ \llo
\end{align*}
\caption{Ticket lock implementation.}
\label{reloc:fig:ticket_lock_impl}
\end{figure}

We show how our HOCAP-style relational specifications for the fine-grained concurrent counter (\Cref{reloc:fig:hocap_cnt_model}) is used to prove that a \emph{ticket lock} refines a \emph{spin lock} (or, rather, that a ticket lock refines any lock satisfying the specification in \Cref{reloc:fig:locks}).
The proof in this section demonstrates several important features of \reloc.
First, it demonstrates compositionality of proofs in \reloc both by employing relational specifications for the left- and right-hand sides, and by reducing the refinement proof of a program module into separate reusable refinement proofs of the module functions.
Second, the proof highlights the integration of Iris ghost state to facilitate CAP-style~\cite{CAP} reasoning with abstract predicates.

A ticket lock~\cite[Section 2.2]{mellor:crummey:scott:1991} is a ticket-based data structure for mutual exclusion, which is fair---threads racing to enter a critical section will gain access to it in the order of arrival at the critical section.
The implementation of the ticket lock is given in \Cref{reloc:fig:ticket_lock_impl}.
The two locations associated with the lock, $\llo$ and $\lln$, point to the identifiers of the current owner of the lock, and to the total number of issued tickets, respectively.
When a thread wants to enter a critical section using the $\acquireTicket$ function, it first requests a new ticket (by atomically increasing the value of $\lln$ using the $\FGincrement$ function), and then spins until the value of the current owner of the lock matches the ticket number (using $\textlog{wait\_loop}$).

The function $\releaseTicket$ uses the weak increment $\wkincr$ (\Cref{reloc:subsec:hocap_nonlin_operations}) on the location $\llo$.
It does not need to use an atomic increment (\ie $\FGincrement$), because, if the lock is used in a well-bracketed manner, only the owner of the lock will be calling the $\releaseTicket$ function.

Concretely, the refinement that we wish to show is the following:
\begin{align*}
\logrel{}{{}& \<pack> (\newlockTicket,\acquireTicket,\releaseTicket)}
            {\\ &\<pack> (\newlock,\acquire,\release)}{\texists{\Ret \tvar. (\tunit \to \tvar) \times (\tvar \to \tunit) \times (\tvar \to \tunit)}}.
\end{align*}
Here, $\newlock$, $\acquire$ and $\release$ are any operations that satisfy the relational specification from \Cref{reloc:fig:locks} (for example, the spin lock from \Cref{reloc:subsec:specs_rhs}).

\paragraph{Proof outline.}
To prove the refinement above we employ our general strategy for proving refinements for stateful program modules in \reloc:
\begin{enumerate}
\item We define an invariant $\lockInv$ linking together the underlying representations of each individual pair of locks, which we use to define a witness for the existential type $\tvar$.
\item We prove the refinements for each method in the signature.
\item Finally, we combine those proofs together into a module refinement proof.
  This is what we refer to as a \emph{component-wise} refinement proof.
\end{enumerate}

We stress that to carry out the proof we neither need to refer to the implementation of the fine-grained concurrent counter (on the left-hand side), nor to the implementation of the spin lock (on the right-hand side).
Rather, we only refer to the HOCAP-style specification for the fine-grained counter and the relational specification for the spin lock.

\paragraph{Proof of the refinement.}

\begin{figure}[t]
\begin{mathpar}
\inferH{newIssuedTickets}
  {}
  {\pvs \Exists \gname. \issuedTickets{\gname}{0}}
\and
\inferH{issueNewTicket}
  {\issuedTickets{\gname}{m}}
  {\pvs \issuedTickets{\gname}{m+1} \ast \ticket{\gname}{m}}
\and
\inferH{ticket-nondup}
  {\ticket{\gname}{n} \and \ticket{\gname}{n}}
  {\FALSE}
\end{mathpar}
\caption{The ticket ghost theory.}
\label{reloc:fig:ticket_ghost}
\end{figure}

To match up the physical representation of tickets in the lock we use Iris's ghost theory to define abstract predicates tracking the tickets.
We will use two ghost predicates: $\issuedTickets{\gname}{m}$ saying that $m$ tickets have been issued in total, and $\ticket{\gname}{n}$ representing the $n$-th ticket.
The predicates satisfy the rules in \Cref{reloc:fig:ticket_ghost}.

To prove the refinement of lock modules, we need to pick a relation (serving as the interpretation for the witness $\tvar$ of the existential type) that links the two modules together.
We use the the relation $\lockInt$ defined as follows:
\begin{align*}
 \lockInv_{\gname}(\loname,\lnname,\lkvar) \eqdef{} &
  \Exists (o\, n : \nat)\, (b : \mathbb{B}).
  \begin{array}[t]{@{} l}
  \cnt{\loname}{1}{o} \ast \cnt{\lnname}{1}{n} \ast \isLock{\lkvar}{b} \mathrel{*} \\
  \issuedTickets{\gname}{n} \ast (\ticket{\gname}{o} \lor b = \<false>)
  \end{array} \\
\lockInt((\llo, \lln), \lkvar) \eqdef{}&
  \Exists \loname,\lnname,\gname.
  \begin{array}[t]{@{} l}
  \CntN{\llo}{\loname}{\lonamesp} \ast
  \CntN{\lln}{\lnname}{\lnnamesp} \ast \\
  \knowInv{\locknamesp}{\lockInv_{\gname}(\loname,\lnname, \lkvar)}
  \end{array}
\end{align*}
Here, $(\llo, \lln)$ is the ticket lock on the left-hand side, and $\lkvar$ is the specification lock on the right-hand side.
The $\lockInt$ relation states that $\llo$ and $\lln$ are concurrent counters with ghost names $\loname$ and $\lnname$, that satisfy the invariant $\lockInv$.
This invariant describes the relation between the values representing two locks.
It states that the values of the counters $\llo$ and $\lln$ are $o$ and $n$, respectively, and that exactly $n$ tickets have been issued.
Furthermore, the right-hand side lock $\lkvar$ is locked iff the ticket $\ticket{\gname}{o}$ of the current owner of the lock is in the invariant; that is, $\ticket{\gname}{o}$ was given up by a thread that acquired the lock.

Using \ruleref{rel-pack} we subdivide the main refinement proof into three refinements for the functions that constitute the lock module:
\begin{enumerate}
\item $\logrel{\dinssingle{\tvar}{\lockInt}}{\newlockTicket}{\newlock}{\tunit \to \tvar}$;
\item $\logrel{\dinssingle{\tvar}{\lockInt}}{\acquireTicket}{\acquire}{\tvar \to \tunit}$;
\item $\logrel{\dinssingle{\tvar}{\lockInt}}{\releaseTicket}{\release}{\tvar \to \tunit}$.
\end{enumerate}

\begin{prop}
$\logrel{\dinssingle{\tvar}{\lockInt}}{\newlockTicket}{\newlock}{\tunit \to \tvar}$.
\end{prop}

\begin{proof}
By rule \ruleref{rel-rec} it suffices to show
$\logrel{\dinssingle{\tvar}{\lockInt}}{\newlockTicket\ \unittt}{\newlock\ \unittt}{\tvar}$.
By applying the symbolic execution rules and \ruleref{new-i-l-hocap} we are left with the goal:
\begin{align*}
  & \CntN{\lln}{\lnname}{\lnnamesp} \ast \CntN{\llo}{\loname}{\lonamesp} \ast{} \\
  & \cnt{\loname}{1}{0} \ast \cnt{\lnname}{1}{0} \ast \isLock{\lkvar}{\<false>}
  \wand \logrel{\dinssingle{\tvar}{\lockInt}}{(\llo, \lln)}{\lkvar}{\tvar}.
\end{align*}
From the premises we can allocate the invariant $\lockInv_{\gname}(\loname, \lnname, \lkvar)$, and 
  we can obtain $\lockInt((\llo, \lln), \lkvar)$.
We finish the proof with \ruleref{rel-return}.
\end{proof}

To prove the $\acquireTicket$ refinement we need the following helper lemma.
\begin{lem}
\label[lem]{reloc:lem:wait_loop}
$\ticket{\gname}{m} \proves
\logrel{\dinssingle{\tvar}{\lockInt}}
  {\mathsf{wait\_loop}\ m\ \llo}{\acquire\ \lkvar}{\tunit}$, provided
we have
$\knowInv{\locknamesp}{\lockInv_{\gname}(\loname, \lnname, \lkvar)}$.
\end{lem}

\begin{proof}
We prove the refinement by L\"ob induction and symbolic execution.
After some pure symbolic executions steps we are left with the goal:
\[
\logrel{\dinssingle{\tvar}{\lockInt}}
{\If (m = \counterread\ \llo) then \unittt \Else \textlog{wait\_loop}\ m\ \llo}
{\acquire\ \lkvar}{\tunit}.
\]
We then apply \ruleref{read-i-l-hocap}, after which it remains to prove
\begin{align*}
  \cntAuth{\loname}{o} & \vsW[\mask'][\mask' \setminus \namesp]
  \cntAuth{\loname}{o}{} \ast{} \\
& \logrel{\dinssingle{\tvar}{\lockInt}}
{\If (m = o) then \unittt \Else \textlog{wait\_loop}\ m\ \llo}
{\acquire\ \lkvar}{\tunit}.
\end{align*}
for any number $o$.
In case $m \neq o$, we symbolically execute the left-hand side and appeal to the induction hypothesis.
In case $m = o$, we proceed by accessing the invariant $\knowInv{\locknamesp}{\lockInv_{\gname}(\loname, \lnname, \lkvar)}$.
Note that it cannot be the case that $b = \<true>$, because then we would have two copies of $\ticket{\gname}{o}$: one from the assumption of the lemma and one from the invariant.
Thus, the case $b = \<true>$ can be eliminated by \ruleref{ticket-nondup}.
Then it must be the case that $b = \<false>$.
In that case we have $\isLock{\lkvar}{\<false>}$ and we can apply \ruleref{acquire-r} to update it to $\isLock{\lkvar}{\<true>}$, changing the right-hand side to $\unittt$.

We finish by closing the invariant, picking this time $b = \<true>$ and storing the $\ticket{\gname}{o}$ from the assumption of the lemma in the invariant.
\end{proof}

\begin{prop}
$\logrel{\dinssingle{\tvar}{\lockInt}}{\acquireTicket}{\acquire}{\tvar \to \tunit}$.
\end{prop}

\begin{proof}
We use \ruleref{rel-rec} and symbolic execution rules, and then \ruleref{inc-i-l-hocap} and \Cref{reloc:lem:wait_loop}.
When we apply \ruleref{inc-i-l-hocap}, we use the update to issue a new ticket using \ruleref{issueNewTicket}.
This ticket will be used for the assumption of \Cref{reloc:lem:wait_loop}.
\end{proof}

\begin{prop}
$\logrel{\dinssingle{\tvar}{\lockInt}}{\releaseTicket}{\release}{\tvar \to \tunit}$.
\end{prop}

\begin{proof}
We use \ruleref{rel-rec}, and then symbolic execution and \ruleref{cnt-wk-incr-l}, after which the new goal becomes
\begin{align*}
\cntAuth{\loname}{n} {}
& \vsW[\mask'] \cntAuth{\loname}{n} \ast{}
(\All m. \cntAuth{\loname}{m}
\vsW[\mask'][\mask'\setminus  \namesp] \\
& {} \cntAuth{\loname}{n+1} \ast{} \logrel[\top\setminus\namesp]{}{\unittt}{\release\ \lkvar}{\type})
\end{align*}
for an arbitrary $n$.
By framing $\cntAuth{\loname}{n}$, it suffices to show
\[
\cntAuth{\loname}{m}
\vsW[\mask'][\mask'\setminus  \namesp] \cntAuth{\loname}{n+1} \ast{} \logrel[\top\setminus\namesp]{}{\unittt}{\release\ \lkvar}{\type}
\]
for an arbitrary $m$.
We utilize this update by accessing the invariant and getting access to
$\cnt{\loname}{1}{o}$.
Using this proposition and $\cntAuth{\loname}{m}$ we apply \ruleref{cnt-update} to get
\[
  \cnt{\loname}{1}{n+1} \ast \cntAuth{\loname}{n+1}.
\]
We frame the second separating conjunct, and use \ruleref{release-r} to reduce the right-hand side to~$\unittt$.
Finally we close the invariant and finish the proof with \ruleref{rel-return}.
\end{proof}

%%% Local Variables:
%%% mode: latex
%%% TeX-master: "reloc"
%%% End:

%%% Local Variables:
%%% mode: latex
%%% TeX-master: "reloc"
%%% End:

\section{Speculative reasoning using prophecy variables}
\label{reloc:sec:prophecies}
In addition to Iris's ordinary ghost state mechanism, which allows to reason about the history of a program, Iris has recently been extended with a mechanism for speculative reasoning based on \emph{prophecy variables}, which allows to reason about the future of a program~\cite{abadi:lamport:1991,irisProph}.
A prophecy variable is a ghost variable that can reference a value that is determined in the future of a program's execution.
While the program that is subject to verification cannot make use of the value of a prophecy variable itself---they are ghost variables---the value can be used in proofs, \eg to speculatively choose a reduction step in the right-hand program.
In this section we show how Iris's mechanism for prophecy variables is integrated into \reloc and can be put to action to prove challenging refinements.

We start this section by illustrating the need for prophecy variables with a motivational example (\Cref{reloc:sec:proph:example}).
We then introduce the proof rules for prophecy variables in \reloc (\Cref{reloc:sec:proph:rules}), and use them to verify the motivational example (\Cref{reloc:sec:proph:proof}).
We finish with another example demonstrating the applicability of prophecy variables to algebraic reasoning for concurrent programs (\Cref{reloc:sec:proph:alg}).

\subsection{Motivational example}
\label{reloc:sec:proph:example}
\begin{figure}
\raggedright
\textbf{The eager and lazy implementation:}
\begin{align*}
\newcoin \eqdef{}&
  \begin{array}[t]{@{} l}
  \Let c = \Ref(\<false>) in \\[-0.2em]
  ((\Lam \unittt. \flipcoin\ c), (\Lam \unittt. \readcoin\ c))
  \end{array}
  &
\newcoinLazy \eqdef{}&
  \begin{array}[t]{@{} l}
  \Let c = \Ref(\SOME(\<false>)) in \\[-0.2em]
  ((\Lam \unittt. \flipcoinLazy\ c), (\Lam \unittt. \readcoinLazy\ c))
  \end{array}
  \\
\flipcoin \eqdef{}& \Lam c.
  c \gets \rand\ \unittt &
\flipcoinLazy \eqdef{}& \Lam c.
  c \gets \NONE
  \\
\readcoin \eqdef{}& \Lam c.
  \deref c &
\readcoinLazy \eqdef{}& \Lam c.
  \begin{array}[t]{@{} l}
  \Match \deref c with \\[-0.2em]
  \mid \SOME(v) \to v \\[-0.2em]
  \mid \NONE \to \\[-0.2em]
  \qquad \Let x = \rand\ \unittt in \\[-0.2em]
  \qquad \If \CAS{c}{\NONE}{x} \\[-0.2em]
  \qquad \!\! then x \Else \readcoinLazy\ c\ \unittt
  \end{array}
\end{align*}
\textbf{The instrumented lazy implementation with prophecy variables:}
\begin{align*}
\newcoinLazyI \eqdef{}&
 \begin{array}[t]{@{} l}
   \Let c = \Ref(\SOME(\<false>)) in\\[-0.2em]
   \Let p = \newProph in            \\[-0.2em]
   \Let \lkvar = \newlock\ \unittt in\\[-0.2em]
   ((\Lam \unittt. \flipcoinLazyI\ c\ \lkvar), (\Lam \unittt. \readcoinLazyI\ c\ \lkvar\ p))
 \end{array} \\
\flipcoinLazyI \eqdef{}& \Lam c\ \lkvar.
  \begin{array}[t]{@{} l}
    \acquire\ \lkvar;\
    c \gets \NONE;\
    \release\ \lkvar
  \end{array}
  \\
\readcoinLazyI \eqdef{}& \Lam c\ \lkvar\ p.
  \begin{array}[t]{@{} l}
  \acquire\ \lkvar; \\[-0.2em]
  \<let> r = 
    \begin{array}[t]{@{} l}
      \Match \deref c with\\
      \mid \SOME(v) \to v \\[-0.2em]
      \mid \NONE \to \\[-0.2em]
      \qquad \<let> x = \rand\ \unittt \<in>\\[-0.2em]
      \qquad c \gets \SOME(x);\ \resolveProph{p}{x};\ x
    \end{array}\\[-0.2em]
   \<in> \release\ \lkvar;\ r
\end{array}
  &
\end{align*}
\caption{The implementations of the coin module.}
\label{reloc:fig:coin_prog}
\end{figure}

\begin{figure}
\begin{minipage}[t]{0.34\textwidth}
\vspace{-0.5em}
\[
\rand \eqdef{} \Lam \unittt.
  \begin{array}[t]{@{} l}
  \Let y = \Ref(\<false>) in \\
  \Fork{y \gets \<true>};\\ \deref y
  \end{array}
\]
\end{minipage}
\begin{minipage}[t]{0.65\textwidth}
\begin{mathpar}
\quad
\inferH{rel-rand-l}
  {\All b \in \mathbb{B}. \left(\logrel{}{\fillctx\lctx[b]}{\exprB}{\type}\right)}
  {\logrel{}{\fillctx\lctx[\rand\ \unittt]}{\exprB}{\type}}
\quad
\inferH{rel-rand-r}
  {b \in \mathbb{B} \and
   \logrel{}{\expr}{\fillctx\lctx[b]}{\type}}
  {\logrel{}{\expr}{\fillctx\lctx[\rand\ \unittt]}{\type}}
\end{mathpar}
\end{minipage}
\caption{The implementation and relational specification of the $\rand$ function.}
\label{reloc:fig:rand_fn}
\end{figure}

The ghost state mechanism of Iris that we have seen so far has allowed us to reason about the history of the program's execution.
However, in some cases that is not enough, and it is required to take the future of the program's execution into account.
As a simple motivating example let, us consider the implementations $\newcoin$ and $\newcoinLazy$ of a coin module in \Cref{reloc:fig:coin_prog}.
They both implement a virtual coin that can be flipped (using the first closure) and whose value can be read (using the second closure), but there is an important difference.
The eager version ($\newcoin$) calculates the flipped value in the $\flipcoin$ function immediately---using the $\rand$ function in \Cref{reloc:fig:rand_fn}, which gives a non-deterministic Boolean value---and the $\readcoin$ function just reads that value.
In contrast, the lazy version ($\newcoinLazy$) does not perform any non-deterministic calculations in its flip operation $\flipcoinLazy$.
Instead, it sets the value of the coin to ``undetermined'' (\ie $\NONE$), and postpones the actual calculation to the $\readcoinLazy$ function.

While these two implementations are rather different, they are contextually equivalent---for clients of the module it is not observable if the coin is flipped eagerly or lazily.
To prove that, we wish to establish the following refinements in \reloc:
\[
\logrel{}{\newcoin}{\newcoinLazy}{(\tunit \to \tunit) \times (\tunit \to \tbool)}
\]
\[
\logrel{}{\newcoinLazy}{\newcoin}{(\tunit \to \tunit) \times (\tunit \to \tbool)}
\]

The first refinement can be proved with the tools that we have already described.
We start by symbolically executing both implementations, obtaining references $c$ and $c_l$ to the internal state of the eager coin and lazy coin, respectively.
We then establish the following invariant linking together the two internal states:
\[
\knowInv{}{\Exists (b : \mathbb{B}). c \mapstoI b \ast (c_l \mapstoS \<None> \vee c_l \mapstoS \<Some>(b))}.
\]
This invariant can be easily shown to be preserved during the $\flipcoin \lrel \flipcoinLazy$ refinement proof.
During the $\readcoin \lrel \readcoinLazy$ refinement proof we can choose which value $\rand\ \unittt$ reduces to based on the current value of $c$, using \ruleref{rel-rand-r}.

However, we cannot prove the second refinement with the same strategy.
The problem is that during the $\flipcoinLazy \lrel \flipcoin$ refinement we reset the value of $c_l$ (on the left-hand side) to $\<None>$, and then we have to establish a simulation on the right-hand side by picking a value for $\rand\ \unittt$ that will be assigned to $c$.
But this value has to be the same value that is picked by $\rand\ \unittt$ in $\readcoinLazy$.
Thus we have to pick a value ``from the future''.

To facilitate this style of reasoning, prophecy variables have been introduced into Iris \cite{irisProph}.
Originally, prophecy variables were used to prove refinements between state machines \cite{abadi:lamport:1991}.
Lately they have been used in Iris for establishing linearizability of concurrent data structures without a fixed linearization point.
In the rest of this section we show how we integrated prophecy variables into \reloc.

\subsection{Prophecy instructions and proof rules}
\label{reloc:sec:proph:rules}

\begin{figure}
\begin{mathpar}
\inferH{rel-newproph-l}
  {
     \All \vec{v}\ p. \Proph{p}{\vec{v}} \wand
     \logrel{\Delta}{\fillctx\lctx[p]}{\expr_2}{\type}
   }
  {\logrel{\Delta}{\fillctx\lctx[\newProph]}{\expr_2}{\type}}
\and
\inferH{rel-resolveproph-l}
  {\pvs[\top][\mask] {
   \Exists \vec{v}. \Proph{p}{\vec{v}} \ast{}
   \big(
     \All \vec{\valB}. (\vec{v}=\conS{\valB}{\vec{\valB}}) \ast \Proph{p}{\vec{\valB}} \wand
     \logrel[\mask]{\Delta}{\fillctx\lctx[\unittt]}{\expr_2}{\type}
   \big)}}
  {\logrel{\Delta}{\fillctx\lctx[\resolveProph{p}{\valB}]}{\expr_2}{\type}}
\and
\inferH{rel-newproph-r}
  {\All p. \left(\logrel[\mask]{\Delta}{\expr_1}{\fillctx\lctx[p]}{\type}\right)}
  {\logrel[\mask]{\Delta}{\expr_1}{\fillctx\lctx[\newProph]}{\type}}
\and
\inferH{rel-resolveproph-r}
  {\logrel[\mask]{\Delta}{\expr_1}{\fillctx\lctx[\unittt]}{\type}}
  {\logrel[\mask]{\Delta}{\expr_1}{\fillctx\lctx[\resolveProph{p}{\valB}]}{\type}}
\end{mathpar}
\caption{The \reloc proof rules for prophecy variables.}
\label{reloc:fig:proph_rules}
\end{figure}

While Iris's ordinary ghost state mechanism only appears at the level of the logic, prophecy variables appear as instrumented instructions in the source program.\footnote{The semantics of \HeapLang has to be instrumented to support prophecy variables, we refer the reader to \cite[Section 3]{irisProph} for details.}
The instruction $\newProph$ creates a new prophecy variable.
The instruction $\resolveProph{p}{v}$ resolves a prophecy variable $p$ to a value $v$.

The symbolic execution rules for the prophecy instructions are given in \Cref{reloc:fig:proph_rules}.
In the right-hand side, the prophecy instructions are no-ops and therefore do not have any pre- or post-conditions.
Prophecy instructions that appear on the left-hand side, however, operate on additional ghost state, and thus have pre- and postconditions.
The ghost predicate $\Proph{p}{\vec{v}}$ says that the prophecy variable $p$ will be resolved, in the future, with values from the vector $\vec{v}$.
Initially, a prophecy variable created with $\newProph$ has an arbitrary vector $\vec{v}$ associated with it.
Only after symbolically executing $\resolveProph{p}{\valB}$ we learn that this vector $\vec{v}$ contains $\valB$ at the head position.
The trick behind the prophecy variables ghost state is that we can already refer to the head element of $\vec{v}$ before resolving it to some $\valB$.
We will see how to use this in establishing the refinement between the lazy coin and the eager coin in the next section.
Note that the rules for the left-hand side are written in logically atomic style: compare, for example, \ruleref{rel-resolveproph-l} and \ruleref{rel-store-l}.

To see how the instrumented instructions for prophecy variables are used, suppose we want to prove a contextual refinement $\ctxref{}{\expr_1}{\expr_2}{\type}$ that involves speculative reasoning.
We first prove a refinement $\hat \expr_1 \lrel \expr_2 : \type$, where $\hat \expr_1$ is a version of $\expr_1$ instrumented with prophecy variables, and then prove $\expr_1 \lrel \hat \expr_1 : \type$ to show that the prophecy variables can be erased.
By soundness of \reloc and transitivity of contextual refinement, this gives a contextual refinement $\ctxref{}{\expr_1}{\expr_2}{\type}$ that refers only to the original programs.

\subsection{Proving the coin refinement}
\label{reloc:sec:proph:proof}

To prove the refinement $\logrel{}{\newcoinLazy}{\newcoin}{(\tunit \to \tunit) \times (\tunit \to \tbool)}$ from \Cref{reloc:sec:proph:example} we instrument the lazy implementation $\newcoinLazy$ with prophecy variables so we can speculate on the outcome of $\rand$ in $\readcoinLazy$.
The instrumented implementation $\newcoinLazyI$ is shown in \Cref{reloc:fig:rand_fn}.
In addition to prophecy variables, we also instrumented the implementation with locks to ensure that there is no interference between updating the reference $c$ and resolving the prophecy variable $p$.
With the instrumented program at hand, we will prove the chain of refinements:
\begin{align*}
& \logrel{}{\newcoinLazy}{\newcoinLazyI}{(\tunit \to \tunit) \times (\tunit \to \tbool)}\\
& \logrel{}{\newcoinLazyI}{\newcoin}{(\tunit \to \tunit) \times (\tunit \to \tbool)}.
\label{reloc:eq:coin_instrumented_refinements}
\end{align*}
Via \reloc's soundness theorem, we can compose these refinements at the level of contextual refinement to obtain:
\[
  \ctxref{}{\newcoinLazy}{\newcoin}{(\tunit \to \tunit) \times (\tunit \to \tbool)}.
\]
Note that that we only use the instrumented implementation $\newcoinLazyI$ for the intermediate step, which means that prophecy variables and locks do not appear at all in the final statement above.
The approach of using prophecies as an intermediate step works not just for closed programs, but also for open programs, as it does not rely on an erasure theorem \cite[Section 3.5]{irisProph}.
Moreover, as the example demonstrates, it allows us to make use of locks in the instrumented program.\footnote{\emph{Atomic prophecy resolution} was introduced in~\cite{irisProph} as an alternative to locks to deal with atomicity of prophecy resolution.}.

The first refinement ($\newcoinLazy \lrel \newcoinLazyI$) is easy to prove, we simply use the no-op symbolic execution rules for prophecies on the right-hand side (\Cref{reloc:fig:proph_rules}).
The second refinement ($\newcoinLazyI \lrel \newcoin$) is where the mechanism of prophecy variables comes to help.
We symbolically execute the allocation parts of $\newcoinLazyI$ and $\newcoin$.
We then use the relational specification for locks (\Cref{reloc:sec:simpl_specs_lhs}) with the following lock invariant:
\[
  \Exists \vec{v}. \Proph{p}{\vec{v}} \ast
    \big(
      (c_l \mapstoI \NONE \ast c \mapstoS (hd\ \vec{v}))
      \vee \spac (\Exists (b : \mathbb{B}). c_l \mapstoI \SOME(b) \ast c \mapstoS b)
    \big).
\]
This invariant says that if the value of the lazy coin is $\NONE$, then the value of the eager coin is determined by the prophecy variable $p$.
There are two main implications of this:
\begin{enumerate}
\item In the refinement between $\flipcoinLazyI$ and $\flipcoin$, the invariant can be (re)established, because we can pick the value of $\rand\ \unittt$ on the right-hand side to be the head element of $\vec{v}$---the future value of the lazy coin is already bound at this point.
\item In the refinement between $\readcoinLazyI$ and $\readcoin$ (specifically, in the $\NONE$ branch), we obtain a non-deterministic Boolean $x$ from symbolically executing $\rand\ \unittt$ on the left-hand side, and we update the value of $c_l$ to be $x$.
  Moreover, we resolve the prophecy variable $p$ to $x$, which gives us much desired information: the head element of $\vec{v}$ was $x$ all along!
  This information allows us to transition from the left disjunct to the right disjunct in the invariant and complete the proof.
\end{enumerate}

\subsection{Algebraic reasoning about non-deterministic choice}
\label{reloc:sec:proph:alg}
In this section we give another example of the use of prophecy variables: we verify several algebraic properties of non-deterministic choice modulo contextual equivalence.
(In)equational theories of the non-deterministic choice operator were previously considered in the context of domain theory, where non-determinism is usually modeled using power domains~\cite{plotkin:1976,smyth:1976}, and in the context of algebraic effects~\cite{simpson:voorneveld:2020,johann:simpson:voigtlander:2010}.
Power domains and the denotational semantics approach does not seem to scale easily to languages with concurrency and higher-order store.
An operational approach to equational theory of a programming language with non-determinism was considered in \cite{birkedal:bizjak:schwinghammer:2013} using step-indexed logical relations.
There the authors show several contextual equivalences involving non-determinism, both finite (\eg picking a Boolean) and countable (\eg picking a natural number).
In this subsection, we provide conceptually simple proofs for contextual equivalences involving finite non-determinism only.
However, we were also able to prove that non-deterministic choice and sequential composition distribute over each other.
Proving this crucially relies on speculative reasoning which we formalize using prophecy variables.

We do not have a non-deterministic choice operation built-in the language, but we can define it using the $\rand$ function from \Cref{reloc:sec:proph:example}.
The operation $\orFn$ non-determini\-sti\-ca\-lly executes one of its thunked arguments:
\[
  \orFn\spac \exprB_1\spac \exprB_2 \eqdef{}
  \If \rand\ \unittt then \exprB_1\ \unittt \Else \exprB_2\ \unittt
\]
We write $\expr_1 \orOp \expr_2$ for $\orFn\ (\Lam \unittt. \expr_1)\ (\Lam \unittt. \expr_2)$.
The expression $\expr_1 \orOp \expr_2$ thus non-deterministically reduces to either $\expr_1$ or $\expr_2$.
From the rules for the $\rand$ function (\Cref{reloc:fig:rand_fn}), we derive the following symbolic execution rules for $\orOp$:
\begin{mathpar}
\inferH{rel-or-l}
  {(\logrel{\Delta}{\fillctx\lctx[\expr_1]}{\exprB}{\type})
   \wedge
   (\logrel{\Delta}{\fillctx\lctx[\expr_2]}{\exprB}{\type})}
  {\logrel{\Delta}{\fillctx\lctx[\expr_1 \orOp \expr_2]}{\exprB}{\type}}
\\
\inferH{rel-or-r-1}
  {\logrel[\mask]{\Delta}{\exprB}{\fillctx\lctx[\expr_1]}{\type}}
  {\logrel[\mask]{\Delta}{\exprB}{\fillctx\lctx[\expr_1 \orOp \expr_2]}{\type}}
\and
\inferH{rel-or-r-2}
  {\logrel[\mask]{\Delta}{\exprB}{\fillctx\lctx[\expr_2]}{\type}}
  {\logrel[\mask]{\Delta}{\exprB}{\fillctx\lctx[\expr_1 \orOp \expr_2]}{\type}}
\end{mathpar}
The rules for $\orOp$ are reminiscent of the rules for disjunction ($\vee$) in sequent calculus.
To symbolically execute $\orOp$ on the left-hand side (\cf to eliminate $\vee$) it is necessary to establish refinements for both operands (\cf to consider both disjuncts), and to symbolically execute $\orOp$ on the right-hand side (\cf to introduce $\vee$) it suffices to establish a refinement for one of the operands (\cf prove one of the disjuncts).

Assume that $\expr_1, \expr_2, \expr_3$ are closed programs of type $\type$.
Then using \ruleref{rel-or-r-1}, \ruleref{rel-or-r-2}, and \ruleref{rel-or-l}, we prove the following equivalences:
\begin{mathpar}
\ctxequiv{}{\expr_1}{\expr_1 \orOp \expr_1}{\type}
\and
\ctxequiv{}{\expr_1 \orOp \expr_2}{\expr_2 \orOp \expr_1}{\type}
\and
\ctxequiv{}{\expr_1}{\expr_1 \orOp \diverge}{\type}
\and
\ctxequiv{}{\expr_1 \orOp (\expr_2 \orOp \expr_3)}{(\expr_1 \orOp \expr_2) \orOp \expr_3}{\type}
\and
\ctxequiv{}{(\expr_1 \orOp \expr_2) ; \expr_3}{(\expr_1 ; \expr_3) \orOp (\expr_2 ; \expr_3)}{\type}
\end{mathpar}
The equational theory that we obtain here is similar to the one obtained from the Hoare power domain, as ${\expr_1 \orOp \diverge}$ (where $\diverge$ is an infinite loop) is identified with $\expr_1$.
The last equation states that non-deterministic choice distributes over sequential composition, and is standard in, \eg process calculi.
What is less standard is the following equation, which is not validated by models based on bisimulation:
\[
\ctxequiv{}{\expr_1 ; (\expr_2 \orOp \expr_3)}{(\expr_1 ; \expr_2) \orOp (\expr_1 ; \expr_3)}{\type}.
\]
This equation, however, holds in Kleene algebra-like models~\cite{hoare:moller:struth:wehrman:2011,kozen:1994}.
If we think about proving this equation using the symbolic execution rules for $\orOp$, then we can observe that proving the refinement in right-to-left direction
\[
  \logrel{}{(\expr_1 ; \expr_2) \orOp (\expr_1 ; \expr_3)}{\expr_1 ; (\expr_2 \orOp \expr_3)}{\type}
\]
is possible, and, by \ruleref{rel-or-l}, it boils down to proving two refinements:
\[
  \logrel{}{\expr_1 ; \expr_2}{\expr_1 ; (\expr_2 \orOp \expr_3)}{\type}
\qquad
  \logrel{}{\expr_1 ; \expr_3}{\expr_1 ; (\expr_2 \orOp \expr_3)}{\type}.
\]
However, proving the refinement in left-to-right direction is harder:
\[
  \logrel{}{\expr_1 ; (\expr_2 \orOp \expr_3)}{(\expr_1 ; \expr_2) \orOp (\expr_1 ; \expr_3)}{\type}.
\]
If we want to use the symbolic execution rules for $\orOp$, we have to ``synchronize'' both sides on $\expr_1$.
To do that, we have to pick a branch for $\orOp$ on the right-hand side before we get to use \ruleref{rel-or-l} on the left-hand side, but we do not know ahead of time which branch to pick.
To resolve this dependency, we use a prophecy variable to speculate on which branch $\expr_2 \orOp \expr_3$ will be taken on the left-hand side, and use the value of this prophecy variable to choose the appropriate branch of ${(\expr_1 ; \expr_2) \orOp (\expr_1 ; \expr_3)}$ on the right-hand side.

The intermediate program that is instrumented with prophecy variables is as follows:
\begin{align*}
& \Let p {} = \newProph in \\
& \expr_1 ;\
  \big((\resolveProph{p}{0}; \expr_2) \orOp (\resolveProph{p}{1}; \expr_3)\big)
\end{align*}
We can easily verify that the original program ${\expr_1 ; (\expr_2 \orOp \expr_3)}$ refines the instrumented one.
To verify that the instrumented program refines ${(\expr_1 ; \expr_2) \orOp (\expr_1 ; \expr_3)}$ we symbolically execute $\newProph$ and obtain a predicate $\Proph{p}{\vec{v}}$ associating a vector of future values $\vec{v}$ to the newly created prophecy variable $p$.
Then we examine the head element $\valB$ of the prophecy values $\vec{v}$.
If $\valB$ is $0$, then we apply \ruleref{rel-or-r-1}, otherwise we apply \ruleref{rel-or-r-2}.
Without loss of generality, suppose that $\valB$ is $0$; that is, $\vec{v} = \conS{0}{\vec{\valB}}$ for some tail $\vec{\valB}$.
First we ``synchronize'' the refinement proof on $\expr_1$ on both sides.
Then we apply \ruleref{rel-or-l}.
Because the premises of \ruleref{rel-or-l} are joined by intuitionistic conjunction $\wedge$, we can use the resource $\Proph{p}{\vec{v}}$ for verifying both refinements:
\begin{align*}
& \Proph{p}{\conS{0}{\vec{v'}}} \wand \logrel{}{\resolveProph{p}{0}; \expr_2}{\expr_2}{\type} \\
& \Proph{p}{\conS{0}{\vec{v'}}} \wand \logrel{}{\resolveProph{p}{1}; \expr_3}{\expr_2}{\type}
\end{align*}
The first refinement is reduced to $\logrel{}{\expr_2}{\expr_2}{\type}$, which follows from the fundamental property (\Cref{reloc:lem:fundamental}) and the assumption that $\expr_2$ is well-typed.
To prove the second refinement we symbolically execute $\resolveProph{p}{1}$ on the left-hand side, at which point we reach a contradiction $0=1$.

%%% Local Variables:
%%% mode: latex
%%% TeX-master: "reloc"
%%% End:

\section{The logical relations model of \reloc}
\label{reloc:sec:model}
\reloc extends Iris with logical connectives and corresponding proof rules for reasoning about refinements.
In this section we show how this is achieved by modeling the connectives of \reloc through a shallow embedding in Iris and proving the logical rules of \reloc as mere lemmas in Iris.
We describe how the refinement judgment $\logrel{}{\expr_1}{\expr_2}{\type}$ is modeled through Iris's weakest preconditions and a \emph{ghost thread pool} construction (\Cref{reloc:sec:model:refinement}) combined with a \emph{binary logical relation} $\Sem{\type}_\Delta$ that describes when values are related (\Cref{reloc:sec:model:logrel}).
We then summarize how the \reloc proof rules (\Cref{reloc:sec:model:rules}) and soundness theorem (\Cref{reloc:sec:model:soundness}) are proved.
The key definitions of the \reloc model are shown in \Cref{reloc:fig:logrel_def}.

The construction of our model generalizes prior work by Turon \etal~\cite{turon:thamsborg:ahmed:birkedal:dreyer:2013,CaReSL}, which culminated in the CaReSL logic, and was subsequently mechanized in Iris by Krebbers \etal~\cite{irisIPM} and Timany~\cite{amin:thesis}.
We discuss the differences in \Cref{reloc:sec:soundness:related}.

\begin{figure}
\raggedright
\textbf{Refinement judgments}:
\medskip
\begin{align*}
  \logrel[\mask]{\Delta & }{\expr_1}{\expr_2}{\type} \eqdef
  \All i,\lctx.
  \begin{array}[t]{@{} l}
  \specctx \wand i \tpto \fillctx\lctx[\expr_2] \vsW[\mask][\top] \\
  \wpre{\expr_1}{\Ret \val_1. \Exists \val_2. i \tpto \fillctx\lctx[\val_2] \ast \Sem{\type}_{\Delta}(\val_1, \val_2)}
  \end{array}
\end{align*}

\bigskip
\textbf{Interpretation of types}:
\medskip
\begin{align*}
\Sem{\tvar}_{\Delta} \eqdef{}& \Lam (\val_1, \val_2). \Delta(\tvar)(\val_1, \val_2) \\
%%%%%%%%%%%%%%%%%%%%%%%%%%%%%%%%%%%%%%%%%%%%%%%%%%
\Sem{\tunit}_\Delta \eqdef{}& \Lam (\val_1, \val_2). \val_1 = \val_2 = \unittt \\
%%%%%%%%%%%%%%%%%%%%%%%%%%%%%%%%%%%%%%%%%%%%%%%%%%
\Sem{\tbool}_\Delta \eqdef{}& \Lam (\val_1, \val_2). (\val_1 = \val_2 = \<true>) \vee (\val_1 = \val_2 = \<false>) \\
%%%%%%%%%%%%%%%%%%%%%%%%%%%%%%%%%%%%%%%%%%%%%%%%%%
\Sem{\tint}_\Delta \eqdef{}& \Lam (\val_1, \val_2). \Exists n \in \mathbb{Z}. \val_1 = \val_2 = n \\
%%%%%%%%%%%%%%%%%%%%%%%%%%%%%%%%%%%%%%%%%%%%%%%%%%
\Sem{\type \times \typeB}_\Delta \eqdef{}& \Lam (\val_1, \val_2).
\Exists \valB_1,\valB_2,\valB'_1,\valB'_2.
\begin{array}[t]{@{}l}
   \val_1 = (\valB_1, \valB_2) \ast \val_2 = (\valB'_1, \valB'_2) \ast{} \\
   \Sem{\type}_\Delta(\valB_1, \valB'_1) \ast
     \Sem{\typeB}_\Delta(\valB_2, \valB'_2)
\end{array}\\
%%%%%%%%%%%%%%%%%%%%%%%%%%%%%%%%%%%%%%%%%%%%%%%%%%
\Sem{\type + \typeB}_\Delta \eqdef{}& \Lam (\val_1, \val_2).
   \Exists \valB_1,\valB_2.
\begin{array}[t]{@{}l}
    \left( \val_1 = \<inl>(\valB_1) \ast \val_2 = \<inl>(\valB_2) \ast{}
           \Sem{\type}_{\Delta}(\valB_1, \valB_2) \right) \vee{} \\
    \left( \val_1 = \<inr>(\valB_1) \ast \val_2 = \<inr>(\valB_2) \ast
           \Sem{\typeB}_{\Delta}(\valB_1, \valB_2) \right)
\end{array}\\
%%%%%%%%%%%%%%%%%%%%%%%%%%%%%%%%%%%%%%%%%%%%%%%%%%
\Sem{\type \to \typeB}_\Delta \eqdef{}& \Lam (\val_1, \val_2).
  \always \left(\All \valB_1,\valB_2. \Sem{\type}_{\Delta}(\valB_1, \valB_2) \wand
                                       (\logrel{\Delta}{\val_1\ \valB_1}{\val_2\ \valB_2}{\typeB})\right) \\
%%%%%%%%%%%%%%%%%%%%%%%%%%%%%%%%%%%%%%%%%%%%%%%%%%
\Sem{\tforall{\Ret \tvar.\type}}_\Delta \eqdef{}& \Lam (\val_1, \val_2).
\always \left(\All \pred \in \Val \times \Val \to \PersistentProp.
   (\logrel{\deltainsert{\tvar}{\pred}{\Delta}}{\tapp{\val_1}}{\tapp{\val_2}}{\type})\right) \\
%%%%%%%%%%%%%%%%%%%%%%%%%%%%%%%%%%%%%%%%%%%%%%%%%%
\Sem{\texists{\Ret \tvar.\type}}_\Delta \eqdef{}& \Lam (\val_1, \val_2).
\Exists \pred \in \Val \times \Val \to \PersistentProp. \Sem{\type}_{\deltainsert{\tvar}{\pred}{\Delta}}(\val_1, \val_2) \\
%%%%%%%%%%%%%%%%%%%%%%%%%%%%%%%%%%%%%%%%%%%%%%%%%%
\Sem{\trec{\Ret \tvar.\type}}_{\Delta} \eqdef{}& \MU \pred. \Lam (\val_1, \val_2).
\Exists \valB_1,\valB_2. \val_1 = \<fold>(\valB_1) \ast
\val_2 = \<fold>(\valB_2) \ast
   \later \Sem{\type}_{\deltainsert{\tvar}{\pred}{\Delta}}(\valB_1, \valB_2) \\[0.3em]
%%%%%%%%%%%%%%%%%%%%%%%%%%%%%%%%%%%%%%%%%%%%%%%%%%
\Sem{\tref\ \type}_{\Delta} \eqdef{}& \Lam (\val_1, \val_2).
 \Exists \loc_1,\loc_2 \in \Loc.
 \begin{array}[t]{@{}l}
\val_1 = \loc_1 \ast \val_2 = \loc_2 \ast{} \\
\knowInv{(\loc_1, \loc_2)}
{
\Exists \valB_1,\valB_2. \loc_1 \mapstoI \valB_1 \ast \loc_2 \mapstoS \valB_2 \ast \Sem{\type}_{\Delta}(\valB_1, \valB_2)
}
\end{array}
\end{align*}
\caption{The model of \reloc in Iris.}
\label{reloc:fig:logrel_def}
\end{figure}

\newcommand{\ghostThreadPoolFigure}{
\begin{figure}
\begin{mathpar}
\inferH{step-pure}
  {\specctx \and
   i \tpto \expr \and
   \pureexec{\expr}{\expr'}}
  {\pvs[\mask] i \tpto \expr'}
\and
\inferH{step-alloc}
  {\specctx \and
   i \tpto \fillctx\lctx[\Ref(\val)]}
  {\pvs[\mask] \Exists \loc.
   i \tpto \fillctx\lctx[\loc] \ast \loc \mapstoS \val}
\and
\inferH{step-store}
  {\specctx \and
   i \tpto \fillctx\lctx[\loc \gets \valB] \and
   \loc \mapstoS \val}
  {\pvs[\mask]
   i \tpto \fillctx\lctx[\unittt] \ast \loc \mapstoS \valB}
\and
\inferH{step-fork}
  {\specctx \and
   i \tpto \fillctx\lctx[\Fork{\expr}]}
  {\pvs[\mask] \Exists j.
   i \tpto \fillctx\lctx[\unittt] \ast j \tpto \expr }
\end{mathpar}
\caption{Selected rules for the ghost thread pool.}
\label{reloc:fig:ghost_thread_pool}
\end{figure}}

\subsection{The refinement judgment}
\label{reloc:sec:model:refinement}
Recall from \Cref{reloc:subsec:grammar} that the intuitive meaning of the refinement proposition $\logrel{}{\expr_1}{\expr_2}{\type}$ is that \emph{any} behavior of $\expr_1$ can be simulated by \emph{some} behavior of~$\expr_2$.
This intuitive idea is modeled in Iris as follows:
\begin{align*}
  \logrel[\mask]{\Delta & }{\expr_1}{\expr_2}{\type} \eqdef
  \All i,\lctx.
  \begin{array}[t]{@{} l}
  \specctx \wand i \tpto \fillctx\lctx[\expr_2] \vsW[\mask][\top] \\
  \wpre{\expr_1}{\Ret \val_1. \Exists \val_2. i \tpto \fillctx\lctx[\val_2] \ast \Sem{\type}_{\Delta}(\val_1, \val_2)}
  \end{array}
\end{align*}

This definition is quite a mouthful, so let us go over it piece by piece.
First, it involves Iris's \emph{weakest precondition} connective $\wpre {\expr} \pred$, which gives the weakest precondition under which execution of $\expr$ is safe, and when $\expr$ returns with value $\val$, the postcondition $\pred(\val)$ holds.
Second, it involves the \emph{ghost thread pool} connective $i \tpto \expr$, which is defined through Iris's ghost theory, and states that the $i$-th ghost thread is executing a program $\expr$.
Putting these pieces together (ignoring $\specctx$ and $\vsW[\mask][\top]$ for now), this definition states that if a (ghost) thread $i$ is executing right-hand side $\expr_2$, and left-hand side $\expr_1$ reduces to some value $\val_1$, then a corresponding execution can be made so that (ghost) thread $i$ is executing right-hand side~$\val_2$.
The result values $\val_1$ and $\val_2$ of the left-hand and right-hand side should be related via the value interpretation $\Sem{\type}_{\Delta}(\val_1, \val_2)$, which we model in \Cref{reloc:sec:model:logrel} via a logical relation.
The quantification over $\lctx$ closes the definition under evaluation contexts.
The expression $\expr_1$ on the left-hand side does not need to be closed under evaluation contexts because weakest preconditions enjoy the rule: $\wpre \expr {\Ret \valB. \wpre {\fillctx\lctx[\valB]} \pred} \wand
\wpre {\fillctx\lctx[\expr]} \pred$.

The ghost thread pool predicates satisfy a number of symbolic execution rules corresponding to executions in the operational semantics.
A selection of these rules is given in \Cref{reloc:fig:ghost_thread_pool}.
The $\specctx$ proposition is an Iris invariant that ties together the thread pool connectives $i \tpto \expr$ and the heap assertions $\loc \mapstoS \val$ with a matching execution on the right-hand side.
We will explain the role of $\specctx$ in \Cref{reloc:sec:model:soundness}.

We should emphasize that the combination of the weakest precondition and the ghost thread pool in the definition of $\logrel[\mask]{\Delta}{\expr_1}{\expr_2}{\type}$ model the demonic nature of $\expr_1$ and the angelic nature of $\expr_2$.
To prove the weakest precondition $\wpre{\expr_1}{\Ret \val_1. \Exists \val_2. i \tpto \fillctx\lctx[\val_2] \ast \Sem{\type}_{\Delta}(\val_1, \val_2)}$ one has to consider all behaviors of $\expr_1$, but has to establish only a single matching execution for $\expr_2$ by using the appropriate rules for the ghost thread pool.

\ghostThreadPoolFigure

\subsection{The logical relation}
\label{reloc:sec:model:logrel}

The interpretation of types $\Sem{\type}_{\Delta}(\val_1, \val_2)$, as defined in \Cref{reloc:fig:logrel_def}, expresses when two values $\val_1$ and $\val_2$ are related at type $\type$ (in context $\Delta$).
The definition of $\Sem{\type}_{\Delta}(\val_1, \val_2)$ follows the usual structure of a logical relation, it is defined recursively on the structure of the type $\type$ and uses the corresponding logical connectives via the Curry-Howard isomorphism.
For example, products are defined via (separating) conjunction, sums are defined via disjunction, functions are defined via (separating) implication, universal types are defined via universal quantification, \etc

The interpretation of recursive types and reference types are somewhat more interesting, as they make use of Iris-specific connectives.
The interpretation of the recursive type $\trec{\Ret \tvar.\type}$ makes use of Iris's guarded fixed point operator $\MU \var. \term$, which is used to define recursive predicates without a restriction of the variance of the recursive occurrence $\var$ in $\term$, but requires $\var$ to appear in \emph{guarded} position, \ie under the later modality $\later$ \cite[Section 5.6]{irisJFP}.
To define the interpretation of the reference type $\tref\ \type$, we use the invariant
\[
\knowInv{(\loc_1, \loc_2)}
{
\Exists \valB_1,\valB_2. \loc_1 \mapstoI \valB_1 \ast \loc_2 \mapstoS \valB_2 \ast \Sem{\type}_{\Delta}(\valB_1, \valB_2)
},
\]
which states that whatever values are stored in $\loc_1$ and $\loc_2$ are always related at type $\Sem{\type}_{\Delta}$.

The persistence modality $\always$ in the interpretation for function types and universal types is used to ensure that the type interpretation is persistent and prevents the kind of issues described in \Cref{reloc:sec:persistence}.
Similarly, in the interpretation of the universal and existential types we quantify over a \emph{persistent} predicate $\pred \in \Val \times \Val \to \PersistentProp$, where $\PersistentProp$ is the subset of Iris propositions that is persistent.

\subsection{Differences with prior work.}
\label{reloc:sec:soundness:related}

The definition of the refinement $\logrel[\mask]{\Delta}{\expr_1}{\expr_2}{\type}$ and value interpretation $\Sem{\type}_{\Delta}(\val_1, \val_2)$ generalize the versions by Krebbers \etal~\cite{irisIPM} and Timany \etal~\cite{amin:thesis}, which in turn adapted ghost thread pools by Turon \etal~\cite{turon:thamsborg:ahmed:birkedal:dreyer:2013,CaReSL} by modeling these in Iris.
The main novelty is that our refinement judgment $\logrel[\mask]{\Delta}{\expr_1}{\expr_2}{\type}$ is a first-class Iris proposition, instead of a meta-logical proposition.
As we have demonstrated throughout this paper, this modification is simple, albeit crucial for writing conditional refinements and to obtain high-level proof rules for refinements.

Furthermore, to obtain high-level proof rules for invariants, we have equipped the refinement judgment with a mask $\mask$, which keeps track of the invariants that may be opened.
To give the appropriate semantics to the mask $\mask$, our definition involves the update modality $\vsW[\mask][\top]$.
Note that the definition by Krebbers \etal~\cite{irisIPM} and Timany~\cite{amin:thesis} is logically equivalent to $\vdash \logrel[\top]{\Delta}{\expr_1}{\expr_2}{\type}$, where the derivability relation $\proves$ of Iris is used to turn the judgment into a meta theoretical proposition, and the mask is set to $\top$.

\subsection{Deriving the primitive rules}
\label{reloc:sec:model:rules}

In \Cref{reloc:sec:calculus} we have demonstrated that \reloc's primitive monadic (\ruleref{rel-return} and \ruleref{rel-bind}) and symbolic execution rules can be used to derive \reloc's high-level proof rules, such as its type-directed structural rules.
In this section, we indicate how \reloc's primitive rules are proved by unfolding the definition of the refinement judgment.
We prove the symbolic execution rules through the following auxiliary rules, which allow us to lift Iris's rules for weakest preconditions and the ghost thread pool rules (\Cref{reloc:fig:ghost_thread_pool}) to the refinement judgment:
\begin{mathparpagebreakable}
\inferH{rel-wp-l}
  {\wpre{\expr_1}[\top]{\Ret \val_1. \logrel{}{\fillctx\lctx[\val_1]}{\expr_2}{\type}}}
  {\logrel{}{\fillctx\lctx[\expr_1]}{\expr_2}{\type}}
\and
\inferH{rel-wp-atomic-l}
  {\pvs[\top][\mask] \wpre{\expr_1}[\mask]{\Ret \val_1. \logrel[\mask]{}{\fillctx\lctx[\val_1]}{\expr_2}{\type}}
   \and \physatomic{\expr_1}}
  {\logrel{}{\fillctx\lctx[\expr_1]}{\expr_2}{\type}}
\and
  \inferH{rel-step-r}
  {\All j, \lctx'. \specctx \ast j \tpto \fillctx\lctx'[\fillctx\lctx[\expr_2]]
   \vsW[\mask] \Exists \val_2. j \tpto \fillctx\lctx'[\fillctx\lctx[\val_2]] {}\ast{}
      \logrel[\mask]{}{\expr_1}{\fillctx\lctx[\val_2]}{\type}}
  {\logrel[\mask]{}{\expr_1}{\fillctx\lctx[\expr_2]}{\type}}
\end{mathparpagebreakable}
The rule \ruleref{rel-wp-l} says that we can ``take out'' an expression $\expr_1$ in context $\lctx$ on the left-hand side, and reason about it using Iris's weakest precondition.
The rule \ruleref{rel-wp-atomic-l} is similar, but it also allows for opening an invariant around $\expr_1$, in case $\expr_1$ is atomic.\footnote{Iris's weakest precondition connective $\wpre \expr [\mask] \pred$ is also equipped with a mask to keep track of which invariants may be opened.
This was the inspiration for the mask annotation at \reloc's refinement judgment.}
The rule \ruleref{rel-step-r} says that if we have an expression $\expr_2$ on the right-hand side in an evaluation context $\lctx$, and we can reduce $\expr_2$ to a value $\val_2$, using the ghost thread pool rules,
then we can reduce the refinement proposition to $\logrel[\mask]{}{\expr_1}{\fillctx\lctx[\val_2]}{\type}$.

\subsection{Soundness}
\label{reloc:sec:model:soundness}
Utilizing the definitions in this section, we outline the proof of the soundness theorem (\Cref{reloc:thm:soundness_full}), which says that \reloc's refinement judgment is sound \wrt contextual refinement.
Formally, if $\logrel{\Delta}[\vctx]{\expr_1}{\expr_2}{\type}$ is derivable in \reloc for any $\Delta$ with $\tenv \subseteq \dom(\Delta)$, then $\ctxref{\tenv \mid \vctx}{\expr_1}{\expr_2}{\type}$.
To prove this theorem we make use of two key lemmas: adequacy of the refinement judgment (\Cref{reloc:lem:refines_adequate}), and the fact that the refinement judgment is a precongruence (\Cref{reloc:lem:refines_precongruence}).

\begin{thm}[Adequacy of \reloc]
  \label[thm]{reloc:lem:refines_adequate}
  If $\proves \logrel{\Delta}{\expr_1}{\expr_2}{\type}$ is derivable in \reloc,
  and $(\expr_1, \stateS) \tpstep^{\ast} (\conS{\val_1}\efs_1, \stateS'_1)$,
  then there exists $\val_2, \efs_2$, and $\stateS'_2$ such that
  $(\expr_2, \stateS) \tpstep^{\ast} (\conS{\val_2}{\efs_2}, \stateS'_2)$.
\end{thm}

\begin{lem}
\label[lem]{reloc:lem:refines_precongruence}
Let $\pctx$ be a well-typed context
$\typedctx{\pctx}{\tenv}{\vctx}{\type}{\tenv'}{\vctx'}{\type'}$, then we have
\[
  \always (\All \Delta. \logrel{\Delta}[\vctx]{\expr_1}{\expr_2}{\type})
  \wand (\All \Delta'. \logrel{\Delta'}[\vctx']{\fillctx\pctx[\expr_1]}{\fillctx\pctx[\expr_2]}{\type'})
\]
where $\Delta$ and $\Delta'$ contain at least the type variables in $\tenv$ and $\tenv'$, respectively.
\end{lem}

\Cref{reloc:lem:refines_precongruence} is proved by induction on $\pctx$ making using of \reloc's type-directed structural rules (\Cref{reloc:sec:fundamental}).
The proof of \Cref{reloc:lem:refines_adequate} is rather involved, so before we discuss that, let us see how we prove the soundness theorem by putting these two lemmas together.

\begin{proof}[Proof of \Cref{reloc:thm:soundness_full} (Soundness for open terms).]
Let $\tenv$ be a type environment, and suppose that $\logrel{\Delta}[\vctx]{\expr_1}{\expr_2}{\type}$ is derivable in \reloc for any $\Delta$ with $\tenv \subseteq \dom(\Delta)$.
To prove $\ctxref{\tenv \mid \vctx}{\expr_1}{\expr_2}{\type}$,
suppose we have typed context $\typedctx{\pctx}{\tenv}{\vctx}{\type}{\emptyset}{\emptyset}{\type'}$,
and reduction $(\fillctx\pctx[\expr_1], \emptyset) \tpstep^{\ast} (\conS{\val_1}{\efs_1}, \stateS_1)$.
By \Cref{reloc:lem:refines_precongruence}, we have
$
  \logrel{}{\fillctx\pctx[\expr_1]}{\fillctx\pctx[\expr_2]}{\type'}
$.
Then, by \Cref{reloc:lem:refines_adequate}, we get that $(\fillctx\pctx[\expr_2], \emptyset) \tpstep^{\ast} (\conS{\val_2}{\efs_2}, \stateS_2)$ for some $\val_2, \efs_2$ and $\stateS_2$, which concludes the proof.
\end{proof}

\begin{proof}[Proof of \Cref{reloc:lem:refines_adequate} (Adequacy of \reloc)]
Suppose that $\logrel{\Delta}{\expr_1}{\expr_2}{\type}$ is derivable in \reloc, and we have $(\expr_1, \stateS) \tpstep^{\ast} (\conS{\val_1}\efs_1, \stateS'_1)$.
Now we should exhibit $\val_2, \efs_2$, and $\stateS'_2$ such that $(\expr_2, \stateS) \tpstep^{\ast} (\conS{\val_2}{\efs_2}, \stateS'_2)$.
The high-level structure of the proof is as follows.
First, we allocate the thread pool invariant $\specctx$ and $0 \tpto \expr_2$ for the main-thread of the right-hand side.
Second, by definition of the refinement judgment, we obtain a weakest precondition $\wpre{\expr_1}{\Ret \val_1. \Exists \val_2. 0 \tpto \val_2 \ast \Sem{\type}_{\Delta}(\val_1, \val_2)}$.
Third, by opening $\specctx$ and using adequacy of Iris's weakest preconditions, we obtain $(\expr_2, \stateS) \tpstep^{\ast} (\conS{\val_2}{\efs_2}, \stateS'_2)$.

Carrying out these steps in detail---notably, setting up the required ghost theory for the ghost thread pool---involves some intricate reasoning using Iris features that are out of scope for this paper.
We thus refer the interested reader to the Coq mechanization, and only highlight the key part---the definition of the thread pool invariant $\specctx$:
\[
  \specctx \eqdef \Exists \vec\expr_0,\stateS_0. \knowInv{\namesp_{\textnormal{\reloc}}}{
    \Exists \vec{\expr},\stateS.
    \specinv(\vec{\expr}, \stateS) \ast
    (\vec{\expr}_0, \stateS_0) \tpstep^{\ast} (\vec{\expr}, \stateS)
  }.
\]
The invariant asserts that given an initial configuration $(\vec\expr_0,\stateS_0)$ for the right-hand side (which we set to be $(\expr_2,\stateS)$ when allocating the invariant), the configuration $(\vec{\expr}, \stateS)$ can be reached via the reduction $(\vec{\expr}_0, \stateS_0) \tpstep^{\ast} (\vec{\expr}, \stateS)$.
Here, $\specinv(\vec{\expr}, \stateS)$ is a connective defined using Iris's ghost theory that keeps track of the configuration of the ghost thread pool and ensures it is consistent with the $\tpto$ and $\mapstoS$ connectives.
The latter is essential, as it allows us to conclude from $\specinv(\vec{\expr}, \stateS)$ and $0 \tpto \val_2$ (as given by the post condition of the weakest precondition in the definition of the refinement judgment) that $\vec{\expr}$ is equal to $\conS{\val_2}{\efs_2}$ for some $\efs_2$.
By definition of the invariant $\specctx$, this gives us a reduction $(\expr_2, \stateS) \tpstep^{\ast} (\conS{\val_2}{\efs_2}, \stateS'_2)$ for the right-hand side, which is needed to conclude the third step of the proof.
\end{proof}

%%% Local Variables:
%%% mode: latex
%%% TeX-master: "reloc"
%%% End:

\section{The Coq mechanization of \reloc}
\label{reloc:sec:formalization}
The Coq mechanization of \reloc provides a soundness proof of \reloc and infrastructure to carry out interactive tactic-based refinement proofs.
It is built on top of the mechanization of Iris in Coq~\cite{irisWWW} and the Iris Proof Mode/MoSeL framework for tactic-based proofs in separation logic~\cite{irisIPM,MoSeL}.
In this section we examine the way \reloc's language and type system are defined (\Cref{reloc:subsec:coq:language}), and how the \reloc logic is defined on top of that (\Cref{reloc:subsec:coq:logic}).
We then describe \reloc's tactic support for interactive refinement proofs, which allows us to seamlessly carry out proofs in Coq similar to those we have seen in this paper (\Cref{reloc:subsec:coq:tactic}).
Finally, we give an overview of the source code (\Cref{reloc:subsec:coq:overview}).

\subsection{The programming language}
\label{reloc:subsec:coq:language}

Iris is a programming language independent framework, which means that it can be instantiated with a programming language of choice.
In this paper, we do not make use of this generality, and use \HeapLang---the default language shipped with Iris's Coq development, which is essentially an untyped version of the language we considered in \Cref{reloc:sec:language}.
\HeapLang is represented via a deep embedding and comes with a set of notations so that programs can be written in Coq-style syntax.
For example, the Boolean implementation $\bitbool$ of the bit module from \Cref{reloc:subsec:representation_independence} is written as follows:
\begin{lstlisting}[language=Coq]
Definition bit_bool : expr :=
  (#true, (λ: "b", ~"b"), (λ: "b", "b")).
\end{lstlisting}

Binders in \HeapLang are represented as strings, which makes it possible to write programs in a human-readable way.
This works well in practice because expression-level substitution only acts on closed terms, and thus does not need to be capture avoiding.

We equip \HeapLang with a type system in the usual way---types \coqe{type} are defined as an inductive data type, and the typing judgment \coqe{typed} is defined as an inductive relation:
\begin{lstlisting}[language=Coq]
Inductive type :=
  | TVar : var → type
  | TProd : type → type → type
  | TArrow : type → type → type
  | TExists : {bind 1 of type} → type
  | (* ... *).
Inductive typed : stringmap type → expr → type → Prop :=
  | Var_typed Γ x τ :
     Γ !! x = Some τ → (Γ ⊢ₜ Var x : τ)
  | Pair_typed Γ e1 e2 τ1 τ2 :
     (Γ ⊢ₜ e1 : τ1) → (Γ ⊢ₜ e2 : τ2) → (Γ ⊢ₜ (e1, e2) : τ1 * τ2)
  | Fst_typed Γ e τ1 τ2 :
     (Γ ⊢ₜ e : τ1 * τ2) → (Γ ⊢ₜ Fst e : τ1)
  | (* ... *).
\end{lstlisting}
We use the notation \coqe{Γ ⊢ₜ e : τ} for \coqe{typed Γ e τ}, and overload the standard Coq notations for types, \eg we use the notation \coqe{τ1 * τ2}\, for \coqe{TProd τ1 τ2}\, and \coqe{τ1 → τ2}\, for \coqe{TArrow τ1 τ2}\,.
Since type-level substitution acts on (potentially) open terms, and therefore needs to be capture avoiding,
we use De Bruijn indices to represent type-level binders through the Autosubst Coq library \cite{schafer:tebbi:smolka:2015}.
For example, the type $\tbit \eqdef \texists{\tvar. \tvar \times (\tvar \to \tvar) \times (\tvar \to \tbool)}$ from \Cref{reloc:subsec:representation_independence} is represented in Coq as follows (\texttt{\#} is notation for \coqe{TVar}):
\begin{lstlisting}[language=Coq]
Definition bitτ : type := ∃: #0 * (#0 → #0) * (#0 → TBool).
\end{lstlisting}

\subsection{The \reloc logic}
\label{reloc:subsec:coq:logic}

Recall from \Cref{reloc:sec:model} that \reloc is defined as a shallow definition in Iris---the \reloc connectives are definitions in Iris, and the \reloc proof rules are lemmas in Iris.
In Coq we follow the same approach.
At the core of \reloc we have the definition \coqe{lrel} of \emph{semantic types}, \ie persistent Iris relations over \HeapLang values:
\begin{lstlisting}[language=Coq]
Record lrel Σ := LRel {
  lrel_car :> val → val → iProp Σ;
  lrel_persistent v1 v2 : Persistent (lrel_car v1 v2)
}.
\end{lstlisting}
Here, \coqe{iProp Σ} is the type of Iris propositions.\footnote{The parameter $\Sigma$ describes the kind of ghost state available in Iris.
It is an important but technical detail that can safely be ignored for the purpose of this paper.
An interested reader is directed to \cite[\S 4.7]{irisJFP}.}
The record bundles together a relation together with a proof that it is persistent.
The notation \coqe{:>} declares the field \coqe{lrel_car} as a coercion.
In the Coq mechanization of \reloc we generalize the refinement judgment $\logrel[\mask]{\Delta}{\expr_1}{\expr_2}{\type}$ to range over semantic types (\coqe{lrel}) instead of syntactic types (\coqe{type}):
\begin{lstlisting}[language=Coq]
Definition refines (E : coPset) (e1 e2 :$\,$expr) (A$\,$:$\,$lrel Σ) : iProp Σ :=
  ∀ j K, spec_ctx -∗ j ⤇ fill K e2 ={E,⊤}=∗
         WP e1 {{ v1, ∃ v2, j ⤇ fill K (of_val v2) ∗ A v1 v2 }}.
\end{lstlisting}
The parameter \coqe{E} corresponds to the mask $\mask$, and the semantic type \coqe{A} corresponds to the type interpretation $\Sem{\type}_{\Delta}$.
We use the notation \coqe{REL e1 << e2 @ E : A} for \coqe{refines E e1 e2 A}.
This definition makes use of the ghost thread pool connectives \coqe{spec_ctx} and \coqe{j ⤇ e}, as discussed in \Cref{reloc:sec:model:refinement}, and which were originally defined in Coq in~\cite{irisIPM}.

To formalize the refinement judgment on syntactic types, we first define the semantic interpretation $\Sem{\type}_{\Delta}$, denoted as \coqe{interp τ Δ} in Coq, which maps syntactic types $\type$ to semantic types.
To define the semantic interpretation, we define semantic type formers, which are combinators on semantic types corresponding to each syntactic type former.
For example, the semantic product type is defined as follows:
\begin{lstlisting}[language=Coq]
Definition lrel_prod (A B : lrel Σ) : lrel Σ := LRel (λ v1 v2,
  ∃ w1 w2 w1' w2', ⌜v1 = (w1,w1')%V⌝ ∧ ⌜v2 = (w2,w2')%V⌝ ∧ A w1 w2 ∗ B w1' w2').
\end{lstlisting}
Here, we use Iris's notion \coqe{⌜φ⌝} to embed Coq propositions \coqe{φ : Prop} into Iris, although on paper we take the equality predicate to be primitive.
With the above definitions at hand, we can now define \reloc's refinement judgment $\logrel[\mask]{\Delta}{\expr_1}{\expr_2}{\type}$ as \coqe{REL e1 << e2 @ E : interp τ Δ}\,.

\paragraph{The proof rules.}
For example, the rule \ruleref{rel-load-r} is formalized as the following lemma:
\begin{lstlisting}[language=Coq]
Lemma refines_load_r E K l q v e1 A :
  ↑ relocN ⊆ E →
  l ↦ₛ{q} v -∗
  (l ↦ₛ{q} v -∗ REL e1 << fill K (of_val v) @ E : A) -∗
  REL e1 << fill K !#l @ E : A.
\end{lstlisting}
The lemma states that, under the assumption that $\namecl{\textlog{relocN}} \subseteq \mask$ (\ie \reloc's internal invariants are available in the mask $\mask$), the following separation logic formula holds:
\[
  \loc \fmapstoS[q] \val \wand (\loc \fmapstoS[q] \val \wand{} \logrel[\mask]{}{\exprB}{\fillctx\lctx[\val]}{A})
  \wand \logrel[\mask]{}{\exprB}{\fillctx\lctx[\deref \loc]}{A}
\]
This is exactly the internalization of \ruleref{rel-load-r}.
The other \reloc proof rules are mechanized in a similar way.

\paragraph{Soundness.}
The versions of \reloc's soundness theorem for closed (\Cref{reloc:thm:soundness}) and open terms (\Cref{reloc:thm:soundness_full}) are stated in Coq as follows:
\begin{lstlisting}[language=Coq]
Lemma refines_sound Σ `{relocPreG Σ} e1 e2 τ :
  (∀ `{relocG Σ} Δ, ⊢ REL e1 << e2 : interp τ Δ) →
  ∅ ⊨ e1 ≤ctx≤ e2 : τ.
Lemma refines_sound_open Σ `{relocPreG Σ} Γ e1 e2 τ :
  (∀ `{relocG Σ} Δ, ⊢ {Δ;Γ} ⊨ e1 ≤log≤ e2 : τ) →
  Γ ⊨ e1 ≤ctx≤ e2 : τ.
\end{lstlisting}
Here, \coqe{Γ ⊨ e1 ≤ctx≤ e2 : τ} is the notion of contextual refinement, \coqe{{Δ;Γ$\}$ ⊨  e1 ≤log≤ e2 : τ} is the refinement judgment lifted to open expressions, and \coqe{⊢ P} expresses that the Iris proposition \coqe{P} is derivable.

\paragraph{Example proof: refinement of the bit module.}
In order to prove the contextual refinement \coqe{∅ ⊨ bit_bool ≤ctx≤ bit_nat : bitτ} from \Cref{reloc:subsec:representation_independence}, it suffices to prove the following:
\begin{lstlisting}[language=Coq]
Lemma bit_refinement Δ : ⊢ REL bit_bool << bit_nat : interp bitτ Δ.
\end{lstlisting}
To prove this lemma, we use the relation \coqe{R}, which is the same as the one in \Cref{reloc:subsec:representation_independence}, but wrapped into a semantic type (\coqe{lrel}) to ensure it is persistent:
\begin{lstlisting}[language=Coq]
Definition R : lrel Σ := LRel (λ v1 v2,
  (⌜v1 = #true⌝ ∧ ⌜v2 = #1⌝) ∨ (⌜v1 = #false⌝ ∧ ⌜v2 = #0⌝)).
\end{lstlisting}
Using the relation \coqe{R}, a Coq proof of the desired refinement is as follows:
\begin{lstlisting}[language=Coq]
Lemma bit_refinement Δ : ⊢ REL bit_bool << bit_nat : interp bitτ Δ.
Proof.
  unfold bitτ; simpl. iApply (refines_exists R). (* apply $\mbox{\ruleref{rel-pack}}$ *)
  progress repeat iApply refines_pair. (* repeatedly apply $\mbox{\ruleref{rel-pair}}$ *)
  - rel_values. (* apply $\mbox{\ruleref{rel-return}}$ and solve the goal *)
  - (* ... *)
Qed.
\end{lstlisting}
Finally, we combine \coqe{bit_refinement} with the soundness theorem to get a closed proof of contextual refinement:
\begin{lstlisting}[language=Coq]
Theorem bit_ctx_refinement : ∅ ⊨ bit_bool ≤ctx≤ bit_nat : bitτ
Proof. auto using (refines_sound relocΣ), bit_refinement. Qed.
\end{lstlisting}
It is important to emphasize that the contextual refinements, which we obtain in theorems like \coqe{bit_ctx_refinement} above, are closed propositions in Coq.
The statement (the type) of \coqe{bit_ctx_refinement} does not refer to \reloc or Iris.
This illustrates that the only parts of the trusted code base of our development are the notions that are involved in the definition of contextual refinement, \ie the operational semantics and the typing of contexts.

\subsection{Tactic support for interactive proofs}
\label{reloc:subsec:coq:tactic}
To prove refinement judgments, like the bit refinement \coqe{REL bit_bool << bit_nat : interp bitτ Δ} from the previous section, we can repeatedly apply the Iris lemmas corresponding to the \reloc proof rules.
However, doing so directly quickly becomes unwieldy, as the user has to manually provide the resources (like the precondition \coqe{l ↦ₛ{q$\}$ v} of \coqe{refines_load_r}), and manually select the evaluation context \coqe{K}.
For better usability we provide tactic support for symbolic execution.

\paragraph{Interactive separation logic proofs.}
To explain the tactics for \reloc that we have defined, let us first look at the general tactic support in Iris.
The Iris Proof Mode (IPM)~\cite{irisIPM} and its successor MoSeL~\cite{MoSeL} allow us to carry out separation logic proofs interactively, in the style of regular tactic-based proofs in Coq.
\begin{figure}
\begin{minipage}[b]{.4\linewidth}
\begin{lstlisting}[language=Coq]
-------------------∗
P -∗ (P -∗ Q) -∗ Q
\end{lstlisting}
\subcaption{Before executing any tactics.}
\end{minipage}%
\qquad\qquad
\begin{minipage}[b]{.4\linewidth}
\begin{lstlisting}[language=Coq]
"H1" : P
"H2" : P -∗ Q
-------------------∗
Q
\end{lstlisting}
\subcaption{After \coqe{iIntros "H1 H2".}}
\end{minipage}
\caption{Interactive proof of lemma \coqe{example} in IPM.}
\label{reloc:fig:ipm_simple}
\end{figure}
IPM provides a convenient representation of sequents for separation logic and tactics for manipulating them, allowing for interactive proof development in the style of regular proofs in Coq.
To illustrate this, consider the following separation logic tautology:
\begin{lstlisting}[language=Coq]
Lemma example (P Q : iProp Σ) : P -∗ (P -∗ Q) -∗ Q.
Proof. iIntros "H1 H2". iApply ("H2" with "H1"). Qed.
\end{lstlisting}
The intermediate results can be seen in \Cref{reloc:fig:ipm_simple}.
Applying \coqe{iIntros "H1 H2".} introduces the hypothesis \coqe{P} and \coqe{P -∗ Q} into the IPM context, giving them names \coqe{H1} and \coqe{H2}, respectively.
Then, \coqe{iApply ("H2" with "H1").} applies the separating implication \coqe{P -∗ Q} to the goal, using the hypothesis \coqe{H1 : P} as the assumption.

\paragraph{Symbolic execution tactics.}
In addition to tactics like \coqe{iIntros} and \coqe{iApply}, IPM provides tactic for symbolic execution in weakest preconditions.
We built similar tactics on top of IPM for symbolic execution in refinement judgments.
For example, consider:
\begin{lstlisting}[language=Coq]
Lemma example_load l : l ↦ₛ #0 -∗ REL #2 << (!#l + #2) : lrel_int.
Proof. iIntros "Hl". rel_load_r. rel_pures_r. rel_values. Qed.
\end{lstlisting}
\begin{figure}
\begin{minipage}[b]{.48\linewidth}
\begin{lstlisting}[language=Coq]
"Hl" : l ↦ₛ #0
-------------------------------∗
REL #2 << (!#l + #2) : lrel_int
\end{lstlisting}
\subcaption{Before applying \coqe{rel\_load\_r}.}
\end{minipage}%
\qquad
\begin{minipage}[b]{.44\linewidth}
\begin{lstlisting}[language=Coq]
"Hl" : l ↦ₛ #0
------------------------------∗
REL #2 << (#0 + #2) : lrel_int
\end{lstlisting}
\subcaption{After applying \coqe{rel\_load\_r}.}
\end{minipage}
\begin{minipage}[b]{.49\textwidth}
\begin{lstlisting}[language=Coq]
"Hl" : l ↦ₛ #0
-----------------------------∗
REL #2 << #(0 + 2) : lrel_int
\end{lstlisting}
\subcaption{After applying \coqe{rel\_pures\_r}.}
\end{minipage}
\caption{Interactive refinement proof of lemma \coqe{example_load} in \reloc.}
\label{reloc:fig:load_r_ipm}
\end{figure}
The results of \coqe{rel_load_r} and \coqe{rel_pures_r} can be seen in \Cref{reloc:fig:load_r_ipm}.
The tactic \coqe{rel_load_r} symbolically executes the dereferencing operation, and the tactic \coqe{rel_pures_r} symbolically executes as many pure reduction steps as possible.
The tactic \coqe{rel_values} finishes the goal since both sides are values.
Similarly, we built tactics for all other language connectives (both on the left- and right-hand side).
The tactics were developed in a similar way to the weakest-precondition tactics from IPM, and we refer the reader to~\cite{irisIPM} for details.

\subsection{Overview of the source code.}
\label{reloc:subsec:coq:overview}
The Coq mechanization contains around 10300 lines of code, of which approximately
\begin{enumerate*}
\item 1315 lines for mechanization of the model of \reloc (\Cref{reloc:sec:model}), including the adequacy theorem (\Cref{reloc:lem:refines_adequate}), and the primitive and derived rules (\Cref{reloc:sec:rules});
\item 1200 lines for the tactics (\Cref{reloc:subsec:coq:tactic});
\item 1450 lines for the mechanization of the type system (\Cref{reloc:sec:language}), and the soundness theorem for open term (\Cref{reloc:thm:soundness_full});
\item 6050 lines for the examples and case studies (including the case studies we describe in the upcoming \Cref{reloc:sec:other_examples,reloc:sec:escape-hatch});
\item and 140 lines for tests (mainly regression tests for the tactics).
\end{enumerate*}

%%% Local Variables:
%%% mode: latex
%%% TeX-master: "reloc"
%%% End:

\section{Related work}
\label{reloc:sec:related-work}
We described some of the most closely related work in the introduction (\Cref{reloc:sec:intro}), we now
discuss other related work on logical relations models, relational logics, atomic specifications, speculative reasoning, and linearizability.

\paragraph{Logical relations models.}
Logical relations models over denotational and operational semantics have an extensive history.
To cover advanced programming language features such as recursive types and higher-order references, logical relations with step-indexing have been introduced~\cite{ahmed:appel:virga:2002,ahmed:thesis,ahmed:dreyer:rossberg:2009,birkedal:etal:11}.
Step-indexing has shown to be very effective by a large body of work on step-indexed logical relations models, \eg~\cite{neis:dreyer:rossberg:2011,hur:dreyer:2011,birkedal:Sieczkowski:thamsborg:2012,cicek:paraskevopoulou:garg:2016,rajani:garg:2018}.
However, in these papers step-indices appear explicitly in the definition of the logical relations model and the proofs about it.
In contrast in this paper we have used the ``logical approach'' to step-indexed logical relations.
This approach, pioneered by Dreyer \etal in the LSLR logic~\cite{dreyer:ahmed:birkedal:2009}, hides step-indices by abstracting and internalizing them in a logic using the later modality ($\later$)~\cite{appel:mellies:richards:vouillon:2007}.
Dreyer \etal used this approach to construct a binary logical relations model for System F with recursive types~\cite{dreyer:ahmed:birkedal:2009}, and later extended the approach as part of the LADR logic to cover existential types and references~\cite{dreyer:neis:rossberg:birkedal:2010}.

The logical approach to logical relations was further refined by Turon \etal~\cite{turon:thamsborg:ahmed:birkedal:dreyer:2013,CaReSL}, culminating in the CaReSL logic, who showed
how Hoare triples and ghost thread pools can be used to define a binary logical relation for fine-grained concurrency.
Subsequently, a version of this binary logical relation was defined and mechanized in Iris by Krebbers \etal~\cite{irisIPM} and Timany~\cite{amin:thesis}.
However, in these papers, logical refinement judgments are meta-logical statements, and because of that, there are no high-level proof rules for establishing and combining refinements.
Instead, to prove a refinement judgment, the user of the logic had to unfold the definition of the refinement judgment, and reason directly in CaReSL or Iris.
In this work we provide a generalization that makes refinement judgments first-class logical statements, which is crucial to reason abstractly about invariants and formulate atomic specifications.
The technical differences are discussed in \Cref{reloc:sec:soundness:related}.
Thus we really make use of the fact that Iris is a \emph{higher-order} logic---CaReSL is only a second-order logic and it would not be possible to make refinement judgments first-class in CaReSL (indeed Iris is not only based on CaReSL, but just as much on the \emph{higher-order} iCAP logic of Svendsen and Birkedal~\cite{svendsen:birkedal:2014}).
We also provide a mechanization in Coq with tactical support that supports the same backwards reasoning style that is employed for proving weakest preconditions in Iris~\cite{irisIPM}.

Apart from the directions that we explored in this paper, there has been an abundance of work on logical relations models in Iris.
Binary logical relations models in Iris have been used for proving contextual equivalence in the context of Haskell's \texttt{ST} monad~\cite{timany:stefanesco:krogh-jespersen:birkedal:2018}, first-class per-thread continuations~\cite{timany:birkedal:2019}, and types-and-effect systems~\cite{krogh-jespersen:svendsen:birkedal:2017}.
Unary logical relations models in Iris have been used for proving type safety and data-race freedom of the Rust type system~\cite{jung:jourdan:krebbers:dreyer:2018,dang:jourdan:kaiser:dreyer:2020,jung:jourdan:krebbers:dreyer:2020}, type safety of session types~\cite{hinrichsen:louwrink:krebbers:bengtson:2020}, type safety of Scala's core calculus DOT~\cite{giarrusso:stefanseso:timany:birkedal:krebbers:2020}, and robust safety~\cite{swasey:garg:dreyer:2017,sammler:garg:dreyer:litak:2020}.
Logical relations in Iris have also been used for showing other relational properties such as termination-preserving refinement~\cite{tassarotti:jung:harper:2017}, non-interference of concurrent programs~\cite{SeLoC}, and recovery refinements (refinements in the presence of potential crashes)~\cite{chajed:tassarotti:kaashoek:zeldovich:2019}.
Nearly all of the aforementioned developments have accompanying mechanizations in Coq, and in some of those mechanizations the authors define their own tactics.
They define tactics for either their version of weakest preconditions or for derived operations, but, to the best of our knowledge, they do not define tactics for reasoning about the logical relation directly.

\paragraph{Relational logics.}
Logics for proving relational properties of programs have a long history, going back to the earlier work of Plotkin and Abadi~\cite{plotkin:abadi:1993}.
Since then many relational logics have been developed addressing various applications,
\eg probabilistic properties in security \cite{barthe:gregoire:zanella:2009,barthe:kopf:olmedo:zanella:2012,barthe:dupressoir:gregoire:kunz:schmidt:strub:2014} and cost analysis~\cite{cicek:barthe:gaboardi:garg:hoffmann:2017,radivcek:barthe:gaboardi:garg:zuleger:2017}.
Here we discuss some more recent work on relational logics that are capable of proving program refinements, with a focus on logics with support for higher-order languages, languages with mutable state, and languages with concurrency.

Earlier work on relational logics targeted programming languages with mutable state, but no concurrency.
Relational Hoare logic \cite{benton:2004} and Relational Separation logic \cite{yang:2007} can be used for reasoning about relational properties for first-order imperative programs, and they have inspired several extensions, for example to probabilistic languages \cite{barthe:gregoire:zanella:2009}.

Relational Higher Order Logic (RHOL)~\cite{aguirre:barthe:gaboardi:garg:strub:2019} is a recent relational higher-order logic for reasoning about relational properties of programs using relational refinement types.
The main judgment of RHOL allows one to prove that a relational formula $\varphi$ holds for two expressions, which do not necessarily have the same type.
While it is not directly possible to reason about expressions with different types in \reloc,
we can relate them by using a type variable $\tvar$ and a suitable interpretation of $\tvar$ in the environment $\Delta$.
The authors prove soundness of RHOL and show how to embed a number of type systems into it.
They provide proofs of various relational properties such as non-interference and relative cost, as provided by the systems they embed into RHOL.
In our work we consider only one (family of) relation(s), namely the logical relation for contextual refinement.
The programming language considered in RHOL is a pure terminating variant of simply-typed PCF, while we consider a much richer programming language with general references and concurrency.

Liang and Feng developed a relational rely-guarantee style logic~\cite{liang:feng:2013},
which can be used to prove refinement for fine-grained concurrent algorithms (including those with helping) but, in contrast
to \reloc, it can only be used to reason about first-order programs.

A relational logic for a sequential class-based language with dynamically allocated objects has been introduced by Banerjee \etal~\cite{banerjee:naumann:nikouei:2016}.
Their relational logic is based on region logic~\cite{banerjee:naumann:rosenberg:2013}, a first-order logic, which is amenable to SMT-based automation.
Their relational logic is aimed at proving refinement and non-interference.
The approach was further extended in~\cite{nikouei:banerjee:naumann:2019} to cover representation independence proofs using per-modules invariants and coupling relations.
In contrast, we focus on reasoning about refinements, but also treat concurrent programs and higher-order store, and we provide tool support for tactic-based interactive verification in Coq.

While not a logic in the strict sense, Relational Hoare Type Theory (RHTT)~\cite{nanevski:banerjee:garg:2013} is a dependent type theory for specification and verification of relational properties of higher-order programs with mutable first-order state, capable of expressing information flow and access control properties.
The object programming language of RHTT and the type system itself are shallowly embedded in Coq.

\paragraph{Atomic specifications.}
To our knowledge, we are the first to study logically atomic specifications in the relational setting.
Logically atomic specifications originate in Hoare-style program logics.
Jacobs and Piessens~\cite{jacobs:piessens:2011} have originally developed a methodology for specifying logically atomic operations.
In their approach, specifications are parameterized by auxiliary code that is performed at the linearization point.
This approach was refined to what we refer to as HOCAP-style specifications, originally introduced in the context of the eponymous logic~\cite{HOCAP}, where the role of auxiliary code is filled by \emph{view shifts} \cite{dinsdale-young:birkedal:gardner:parkinson:yang:2013}, which in this paper are given by Iris's update modality $\vsW$ (\Cref{reloc:subsec:hocap_style}).
Compared to the original Jacobs-Piessens approach, in HOCAP-style specifications, the physical state that a logically atomic function operates on is hidden behind an abstract predicate.
Furthermore, HOCAP-style specifications can also be formulated for non-logically atomic operations, as we have seen in \Cref{reloc:subsec:hocap_nonlin_operations}.
The HOCAP-style specifications were later adopted in the iCAP logic~\cite{svendsen:birkedal:2014} and Iris logic~\cite[Chapter 11]{lecturenotes}.

Because Jacobs-Piessens and HOCAP-style specifications require parameterizing the (ghost) functions that are executed at the linearization points, such specifications are often referred to as higher-order.
As an alternative to this higher-order approach, da Rocha Pinto \etal have introduced the notion of logically atomic triples in their program logic TaDA \cite{daRochaPinto:dinsdale-young:gardner:2014,daRochaPinto:2017}.
Logically atomic triples are a first-order construct, built in as a primitive construct into the logic, which can be used to specify the atomic updates that a program performs.
The atomic triples can be systematically composed in the style of Hoare logic.
A more detailed comparison between the first-order and higher-order approach is given in~\cite{dinsdale-young:daRochaPinto:gardner:2018}.
TaDA-style logically atomic triples were adapted for Iris by Jung \etal~\cite{iris1,irisProph}.
Specifically, they are encoded as derived constructs, using the Jacobs-Piessens approach, that satisfy the TaDA-style rules.

\paragraph{Speculative reasoning.}
To facilitate speculative reasoning, we employ the mechanism for prophecy variables recently introduced in Iris~\cite{irisProph}.
Prophecy variables were first introduced by Abadi and Lamport~\cite{abadi:lamport:1991} for the purpose of proving refinements of state machines.
The idea to use prophecy variables in program logic originates in the rely-guarantee style logic of Vafeiadis \cite{vafeiadis:thesis}, although his treatment of prophecy variables is informal, and he appeals to Abadi and Lamport~\cite{abadi:lamport:1991} for soundness.

Prophecy variables are not the only tool for carrying out speculative reasoning.
Both CaReSL~\cite{CaReSL} and extended LRG~\cite{liang:feng:2013} are program logics capable of proving refinements of programs with future-dependent linearization points.
Both employ, albeit in different forms, a mechanism for recording multiple potential logical states of the program.
These multiple states can then be coalesced into a single one, once the linearization point is determined, and that resulting state is used for establishing the refinement.

Other approaches~\cite{khyza:dodds:gotsman:parkinson:2017,delbianco:sergey:nanevski:banerjee:2017} for proving linearizability of algorithms with future-dependent linearization points use Hoare logics with auxiliary state to track the abstract history of a program as a partial order.
The crucial property is that all total extensions of the partial order result in valid linear histories of the program.

\paragraph{Other work on linearizability.}
One of the main application of \reloc is to prove linearizability of concurrent algorithms, by reducing it to contextual refinements.
Proving linearizability has a long history, and the program logic based approach is not the only one.
Other methods include automated model checking based solutions \cite{yang:wei:liu:sun:2009,vechev:yahav:yorsh:2009,cerny:radhakrishna:zufferey:chaudhuri:alur:2010,burckhardt:dern:musuvathi:tan:2010}
and static analysis, in particular shape analysis, \cite{amit:rinetzky:reps:sagiv:yahav:2007,berdine:lev-ami:manevich:ramalingam:sagiv:2008,vafeiadis:2009}.
The model checking approaches in question do not prove linearizability, but automatically \emph{check} execution traces for linearizability, bounding the heap or the number of threads.
Indeed, model checking approaches are designed to find bugs in a ``push-button'' fashion and can generate counterexample traces.
Approaches based on static analysis are usually sound even for unbounded heaps and threads, but limited to first-order programs.

%%% Local Variables:
%%% mode: latex
%%% TeX-master: "reloc"
%%% End:

\section{Discussion and conclusion}
\label{reloc:sec:conclusion}
In this paper we have presented \reloc---the first mechanized relational logic for proving refinements of fine-grained concurrent higher-order programs.
We have demonstrated that \reloc is expressive enough to formally prove contextual refinements of concurrent programs in a modular way, by employing relational specifications of programs.
Moreover, the mechanization of \reloc in Coq allows us to carry out tactic-based interactive proofs in an intuitive way, by using \reloc's type-directed structural rules and symbolic execution rules, coupled with the powerful mechanisms from Iris, such as invariants, ghost state, and prophecy variables.

In the remainder of this paper we discuss other case studies that we have mechanized in \reloc (\Cref{reloc:sec:other_examples}), discuss the ``escape hatch'' of \reloc (\Cref{reloc:sec:escape-hatch}) for verifying programs that cannot be handled by \reloc, and outline some directions for future work (\Cref{reloc:sec:future-work}).

\subsection{Other examples and case studies}
\label{reloc:sec:other_examples}
In addition to the examples that we have presented in the paper, we have mechanized a number of examples from the literature on logical relations in \reloc in Coq.
Below we give a short summary of those examples.
\begin{itemize}
\item Linearizability of the Treiber stack \cite{treiber:1986};
\item Refinement of higher-order cell objects from \cite{koutavas:wand:2006,ahmed:dreyer:rossberg:2009};
\item Refinement of a symbol lookup table and a name generation module from \cite{ahmed:dreyer:rossberg:2009};
\item Many equivalences from \cite{dreyer:neis:birkedal:2012}, adapted for the concurrent setting, including variations of the ``awkward example'' from \cite{pitts:stark:1998}, and the ``higher-order profiling'' example modified to use the atomic increment function $\FGincrement$;
\item Equivalence between different ways of defining the fixed point combinators;
\item Equivalence between late-choice and early-choice examples from \cite{turon:thamsborg:ahmed:birkedal:dreyer:2013};
\item Algebraic laws for the parallel composition operation and its interaction with non-deterministic choice and sequential composition, inspired by the work on Concurrent Kleene Algebra \cite{hoare:moller:struth:wehrman:2011};
\item Linearizability of the Michael-Scott queue~\cite{michael:scott:1996}, mechanized by Friis Vindum and Birkedal~\cite{vindum:birkedal:2021}.
\end{itemize}

\subsection{The ``escape hatch''}
\label{reloc:sec:escape-hatch}
The rules of \reloc are sound, but not complete.
In particular, there are some examples that cannot be verified in \reloc completely.
One class of such examples that we know of, are refinements of fine-grained concurrent data structures with \emph{external} linearization points (as opposed to fixed linearization points or future-dependent linearization points; see \cite{dongol:derrick:2015} for a survey outlining the differences).
Such external linearization points are present, for example, in algorithms that use \emph{helping} or \emph{work-stealing}.
Fortunately, \reloc's model on top of Iris provides an ``escape hatch'' that still allows us to verify some data structures with helping.

In the appendix~\cite{appendix} we consider an example of such a data-structure: a fine-grained concurrent stack with helping, a simplified version of the elimination-backoff stack from \cite{hendler:shavit:yerushalmi:2004}.
We prove that this stack with helping refines a coarse-grained stack (thus showing that the stack with helping is linearizable).
The stack with helping is interesting because two threads that perform a push and pop operation concurrently can \emph{eliminate} each other, by exchanging data through a side channel, thus reducing the contention for the top node of the stack.
To verify this example we make use of \reloc's ``escape hatch''---we unfold the definition of \reloc's refinement judgment, and perform an explicit proof in terms of \reloc's model in Iris so we can explicitly manipulate the ghost thread pool.
As we demonstrate, the ``escape hatch'' does not render \reloc useless for this example: we still use \reloc's proof rules to carry out the majority of the proof.
Only for a small part of the proof we need to work in the model.
This is achieved by encapsulating the \emph{elimination} mechanism of the stack, for which we can provide a logically atomic relational specification that is proved in the model of \reloc.
This specification can then be used through \reloc's high-level rules to verify the complete data structure without further breaking the abstraction.

\subsection{Future work}
\label{reloc:sec:future-work}
In future work we would like to examine the possibility of a more principled approach to specifying and verifying algorithms with helping, without having to reason in the model of \reloc.
In addition, it would be interesting to explore alternative approaches to speculative reasoning that do not involve prophecy variables.
Furthermore, we would like to study applications of \reloc to type-directed program transformations (for example typed closure conversion \cite{ahmed:blume:2008}) and message-passing programs (for example, by integration with the Iris-based Actris logic~\cite{hinrichsen:bengtson:krebbers:2020,hinrichsen:louwrink:krebbers:bengtson:2020}).

It would also be interesting to see how the ReLoC approach can be used for verifying other kinds of refinements, for example termination-sensitive refinements~\cite{tassarotti:jung:harper:2017} or refinements in the presence of crashes~\cite{chajed:tassarotti:kaashoek:zeldovich:2019}.

%%% Local Variables:
%%% mode: latex
%%% TeX-master: "reloc"
%%% End:

\section*{Acknowledgments}
We thank the anonymous reviewers of this paper and the conference version at LICS'18 for their comments and suggestions.
We thank Herman Geuvers and Simon Friis Vindum for discussions,
and Amin Timany for his contributions to the linearizability proof of the stack with helping.

Dan Frumin was supported by the Dutch Research Council (NWO)
under
STW project 14319 (Sovereign) and VIDI Project No. 016.Vidi.189.046
(Unifying Correctness for Communicating Software).
Robbert Krebbers was supported by the Dutch Research Council (NWO), project 016.Veni.192.259.
Lars Birkedal was supported by a Villum Investigator grant (no.\ 25804), Center for Basic Research in Program Verification (CPV), from the VILLUM Foundation and by the ModuRes Sapere Aude Advanced Grant from The Danish Council for Independent Research for the Natural Sciences (FNU).

%%% Local Variables:
%%% mode: latex
%%% TeX-master: "reloc"
%%% End:

\bibliographystyle{alpha}
\bibliography{reloc}

\end{document}